%% file: main.tex
\documentclass[11pt, a4paper]{article}
\input{notation}
\title{Scalability of Metropolis-within-Gibbs schemes for high-dimensional Bayesian models}
\author{Filippo Ascolani\thanks{Department of Statistical Science, Duke University, \url{filippo.ascolani@duke.edu}}\; , Gareth O. Roberts\thanks{Department of Statistics, University of Warwick, \url{gareth.o.roberts@warwick.ac.uk}}\; and Giacomo Zanella\thanks{Department of Decision Sciences and BIDSA, Bocconi University, \url{giacomo.zanella@unibocconi.it}
}}

\begin{document}
\maketitle

\abstract{
We study general coordinate-wise MCMC schemes (such as Metropolis-within-Gibbs samplers), which are commonly used to fit Bayesian non-conjugate hierarchical models.
We relate their convergence properties to the ones of the corresponding (potentially not implementable) random scan Gibbs sampler through the notion of conditional conductance.
This allows us to study the performances of popular Metropolis-within-Gibbs schemes for non-conjugate hierarchical models, in high-dimensional regimes where both number of datapoints and parameters increase. 
Given random data-generating assumptions, we establish dimension-free convergence results, which are in close accordance with numerical evidences. 
Applications to Bayesian models for 
binary regression with unknown hyperparameters and discretely observed diffusions are also discussed.
Motivated by such statistical applications, auxiliary results of independent interest on approximate conductances and perturbation of Markov operators are provided.
}

\section{Introduction}\label{sec:introduction}

Over 30 years ago, the Gibbs sampler (GS) and more general {\em coordinate-wise samplers} (often termed {\em Metropolis-within-Gibbs (MwG) samplers}) were introduced as powerful tools to enable Bayesian inference for structured data (featuring spatial, hierarchical or temporal dependence), for example see \cite{gelfand1990sampling,SmiRob93,besag1993spatial,C92,B11}.
The substantial impact of these methods has spread across a wide variety of scientific fields, for example see \cite{G13,G15,MFR23}.

Coordinate-wise Markov chain Monte Carlo (MCMC) schemes work by partitioning the vector of parameters in different blocks and updating them one at the time, conditional on all the others.
Working in such a coordinate-wise manner can be computationally beneficial in many cases (see e.g.\ Section \ref{sec:motivating}) and it has been observed empirically \citep{SPZ23, luu2024gibbs} that such samplers often have extremely good convergence properties. However theoretical understanding of this phenomenon has proved difficult, as the analysis of the resulting algorithms is often subtle and case-specific. Substantial progress has now been made on understanding the pure GS (see Section \ref{sec:literature}), but implementation of these algorithms is generally restricted to contexts with specialised conditional conjugacy properties and therefore it is important to understand more general coordinate-wise samplers which remain comparatively under-studied.

In this work we study general coordinate-wise samplers (which include MwG as a special case), with a particular focus on relating their convergence properties to the ones of exact GS. A key theoretical result (Corollary \ref{cor:bound_conductance}) states that the performances of a generic coordinate-wise scheme, measured in terms of the conductance of the associated operator, differ from the ones of GS by a multiplicative factor, which we explicitly control through the goodness of the conditional updates. As motivated in the next section, we apply such findings to high-dimensional non-conjugate Bayesian models, where we provide theoretical justification for the empirically observed good performances of coordinate-wise samplers.

\subsection{Motivating example: non-conjugate hierarchical models}\label{sec:motivating}
Our motivating example is given by classical and widely used Bayesian hierarchical models \citep{GH07,G13}, where the observed dataset $Y_{1:J}=(Y_j)_{j=1,\dots,J}$ is divided into $J$ groups, each featuring some local (i.e.\ group-specific) parameters $\theta_j$. 
Consider for example a hierarchical logistic model defined as
\begin{equation}\label{eq:one_level_nested_intro}
\begin{aligned}
Y_j
\mid \theta_j & \sim f(y \mid \theta_j) = \binom{m}{y}\frac{e^{y\theta_j}}{(1+e^{\theta_j})^m} & j = 1, \dots, J,\\
 \theta_j\mid \mu, \tau &\overset{\text{iid}}{\sim} N(\theta \mid \mu, \tau^{-1})& j = 1, \dots, J,\\
 (\mu, \tau)&\sim p_0(\cdot)\,,
\end{aligned}
\end{equation}
with $y \in \{0,\dots, m\}$ and $m$ being a positive integer. Thus, conditional on $\theta_j$, for each group $j$ a sequence of $m$ independent Bernoulli experiments are performed with success probability $e^{\theta_j}/(1+e^{\theta_j})$. Under model \eqref{eq:one_level_nested_intro}, the conditional distribution of $\bm{\theta}$ given $(Y_{1:J},\mu, \tau)$
 factorizes as $\L(\d\bm{\theta}|Y_{1:J},\mu, \tau)=\otimes_{j=1}^J\L(\d\theta_j|Y_j,\mu, \tau)$, where $\otimes$ denotes the product of independent distributions. 
This makes model  \eqref{eq:one_level_nested_intro}
 particularly well-suited for coordinate-wise samplers, since the high-dimensional update of $\bm{\theta}$ given $(\mu, \tau)$ decouples into $J$ independent low-dimensional updates (see Section \ref{sec:complexity} for more discussion on the computational implications of this). Figure \ref{fig:binom_simulations} compares the efficiency of the resulting samplers with two gradient-based MCMC methods. The target is the joint posterior distribution $\L(\d\bm{\theta}, \d \mu,\d \tau|Y_{1:J})$, and we consider the high-dimensional regime where $J\to\infty$, so that both the number of datapoints and parameters, i.e. $n=Jm$ and $p=J + 2$ respectively, diverge. 
Full details on the simulation set-up, including algorithmic and prior specifications, are postponed to Section \ref{sec:discrete_data}. 
 \begin{figure}
\centering
\includegraphics[width=.42\textwidth]{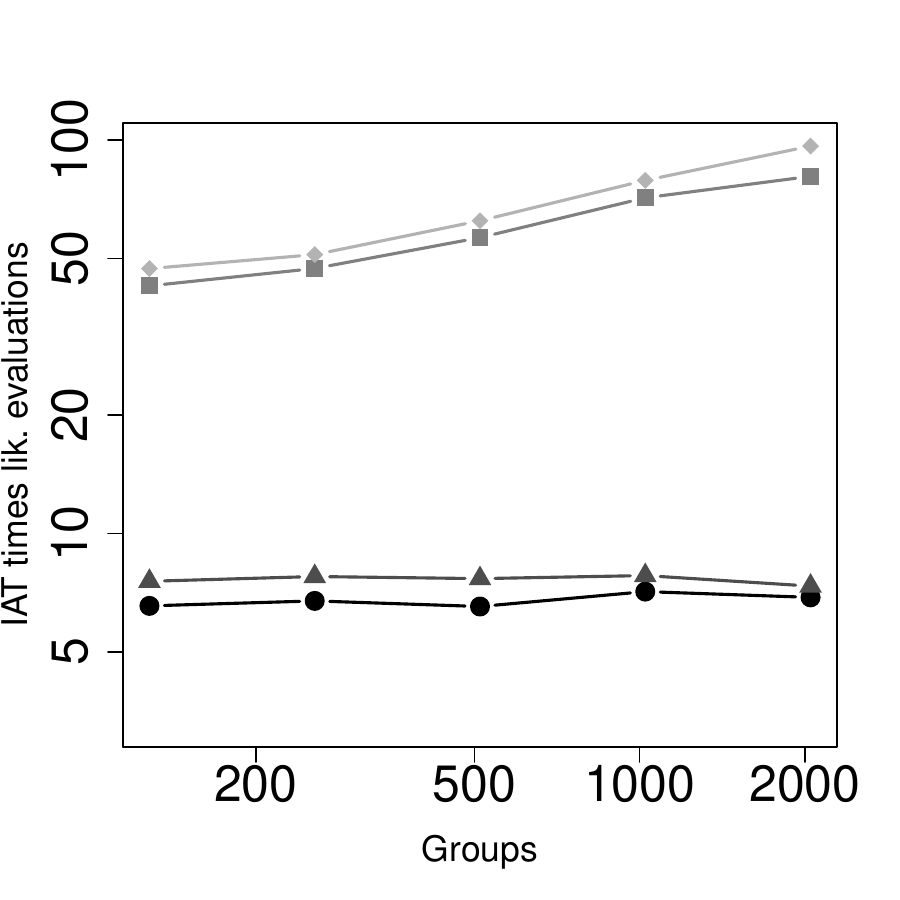} \quad
\includegraphics[width=.48\textwidth]{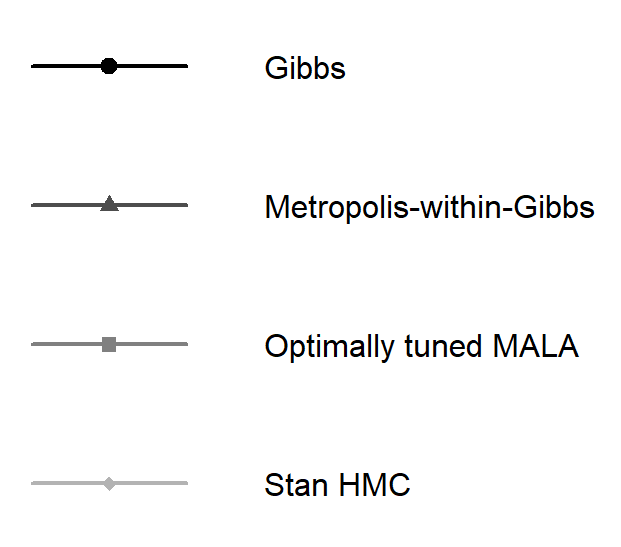}
 \caption{\small{
Median of the integrated autocorrelation times multiplied by the average number of likelihood evaluations per iterations  (on log-scale) for four MCMC schemes targeting the posterior distribution of model \eqref{eq:one_level_nested_intro}, as a function of $J$ (number of groups).
The median refers to repetitions over datasets randomly generated according to the model with true parameters $\mu^* = \tau^* = 1$. See Section \ref{sec:discrete_data} for more details.
  }}
 \label{fig:binom_simulations}
\end{figure}
Note that $\L(\d\theta_j|Y_j,\mu, \tau)$ is not available in closed form for model \eqref{eq:one_level_nested_intro}. %
Exact sampling is possible through an Adaptive Rejection sampler \citep{GW92}, so that exact GS can still be used, but this is potentially hard to implement and its computational cost scales badly with the dimensionality of $\theta_j$. Instead, it is much more common and computationally convenient to implement a MwG scheme performing a $\L(\d\theta_j|Y_j,\mu, \tau)$-invariant Metropolis update of $\theta_j$. 
Both GS and MwG are implemented in Figure \ref{fig:binom_simulations}.

As Figure \ref{fig:binom_simulations} suggests, coordinate-wise samplers can provide state of the art performances for hierarchical models.
In particular, both GS and MwG exhibit dimension-free convergence properties (i.e.\ the number of iterations per effective sample does not grow with $J$) and the slowdown of MwG relative to GS seems to be constant with respect to dimensionality. 
Intuitively, we attribute the good empirical performances of coordinate-wise schemes to the sparse conditional independence structure of hierarchical models, which allows to perform computationally cheap high-dimensional block updates.
This is not a peculiar feature of \eqref{eq:one_level_nested_intro} but rather a common phenomenon occurring in many Bayesian models \cite{SPZ23}.

Whilst the phenomenon of good convergence properties for coordinate-based samplers for hierarchical models has been long recognised, see e.g.\ \cite{SmiRob93}, there has until recently been very little theory to explain this
phenomenon (see Section \ref{sec:literature}). 
In this paper we reduce this gap by proving dimension-free convergence of the mixing times associated to MwG schemes targeting hierarchical models (see Corollary \ref{mixing_times_warm_hier} in Section \ref{sec:dim_free}).

\subsection{Objective and structure of the paper}

The main contributions of the paper can be divided in two parts. 
In the first one, we provide bounds on the approximate conductance of a generic coordinate-wise scheme in terms of the corresponding quantity for the GS (Corollary \ref{cor:bound_conductance}). Working with the approximate version of the conductance 
is crucial for our purposes and subsequent applications (see e.g.\ Remark \ref{asymptotic_cond_Gibbs}).
The general theory naturally applies to MwG schemes, such as those targeting conditionally log-concave distributions (Proposition \ref{prop:RWM_logconcave}).
In the second part, we analyze performances of coordinate-wise samplers for relevant statistical applications, combining the bounds discussed above with specific model properties, statistical asymptotics and some novel auxiliary results on approximate conductances and perturbation of Markov operators. Much emphasis is placed on coordinate-wise schemes for generic two-levels hierarchical models with non-conjugate likelihood, such as \eqref{eq:one_level_nested_intro}, for which we are able to prove dimension-free behaviour of total variation mixing times, under warm (Corollary \ref{mixing_times_warm_hier}) and feasible (Proposition \ref{prop:feasible}) starts. MwG schemes for Bayesian logistic regression with unknown hyper-parameters and inference on discretely observed diffusions are also discussed in some detail.

The structure of the paper is as follows. 
Section \ref{sec:mixing_conductance} briefly recalls the notion of conductance and its connection to mixing times of Markov chains.
Section \ref{sec:main} introduces the class of coordinate-wise schemes that we consider (which include random scan GS and MwG) and provides general results relating the conductance of coordinate-wise schemes to the one of GS (Theorem \ref{main} and Corollary \ref{cor:bound_conductance}). Section \ref{sec:mwg} consider the specific case of MwG schemes and the example of conditionally log-concave targets. 
After discussing independent results about approximate conductances and perturbations of Markov operators (Section \ref{sec:auxiliary}), we move to statistical applications, where we consider Bayesian hierarchical models (Section \ref{sec:hierarchical}), logistic regression with unknown hyperparameters (Section \ref{sec:log_regr_hyper}) and inference for discretely observed diffusion processes (Section \ref{sec:diffusions}).
Mathematical proofs, together with some additional examples and results, are postponed to the Appendix.
R code to replicate all the numerical experiments can be found at \url{https://github.com/gzanella/Metropolis-within-Gibbs/}.

\subsection{Related literature}\label{sec:literature}

Compared to other MCMC schemes (such as gradient-based ones), there are relatively few quantitative theoretical results for coordinate-wise sampling schemes. Moreover, most works focus on the exact Gibbs sampler, e.g.\ \cite{R95,amit1996convergence,R97,K09,K14,Y17,Q19,JH21,Q22,qin2022telescope,CLM23}, for which specific techniques exploiting the exact updating structure and the associated alternating projection representation can be exploited \cite{R01,D08,diaconis2010stochastic,AZ23, ascolani2024entropy}. 
Results for general coordinate-wise MCMC, including MwG, are instead quite limited. Exceptions include \cite{roberts1997geometric,neath2009variable,johnson2013component,jones2014convergence,qin2022convergence}, which mostly focus on geometric or uniform ergodicity, and \cite{tong2020mala}, which provides some analysis of MALA-within-Gibbs schemes.

While working on our manuscript, it came to our attention that \cite{qin2023spectral} independently
and concurrently developed results relating the spectral Gap of coordinate-wise samplers to the ones of exact GS. 
Their results are similar in spirit to the ones we develop in Sections \ref{sec:main} and \ref{sec:mwg}. On the other hand, the specific results and the type of applications are significantly different. 
For example, \cite{qin2023spectral} study the spectral gap, while in this work we work with the approximate version of the conductance defined in \eqref{s_conductance}, which is crucial for the subsequent applications we consider (see e.g.\ Remarks \ref{asymptotic_cond_Gibbs} and \ref{rmk:approx_2}). Also, we consider and develop in some detail various statistical applications (Sections \ref{sec:hierarchical}, \ref{sec:log_regr_hyper} and \ref{sec:diffusions}) where we combine our techniques with posterior asymptotics,  dimensionality reduction and perturbation arguments.

\section{Mixing times and conductance}\label{sec:mixing_conductance}

Let $P$ be a $\pi$-invariant Markov transition kernel on $\sX$, where $\pi \in \mathcal{P}(\sX)$ and $ \mathcal{P}(\sX)$ denotes the collection of probability measures on a space $\sX$. When studying the convergence of a Markov chain to its invariant distribution, a typical object of interest is given by the total variation mixing times starting from $\mu$, defined as
\begin{align*}
t_{mix}(P, \epsilon,\mu)&=\inf\left\{t\geq 0\,:\, \lTV \mu P^t-\pi \rTV<\epsilon \right\},
&\epsilon\in[0,1],\,\mu\in \mathcal{P}(\sX)\,,
\end{align*}
where $P^t$ denotes the $t$-th power of $P$, $\mu P^t(A) = \int_{\sX} P^t(\x,A)\mu(\d \x)$ for any $A\subseteq \sX$ and $\|\cdot\|_{TV}$ denotes the total variation norm. By definition, the mixing times quantify the number of Markov chain's iterations required to obtain a sample from the target distribution $\pi$ up to error $\epsilon$. We will focus on worst-case mixing times with respect to $M$-warm starts. The latter are starting distributions defined as
\begin{align}\label{N_class}
\sN\left(\pi, M \right)=&
\left\{\mu\in\mathcal{P}(\sX)\,:\,\mu(A)\leq M\pi(A) \hbox{ for all }A\subseteq \sX\right\}, &M \geq 1,\,\pi \in \mathcal{P}(\sX) \,.
\end{align}
The associated worst-case mixing times for $P$ targeting $\pi$ are
\begin{equation}\label{absolute_warm_start}
t_{mix}(P, \epsilon, M)=\sup_{\mu\in \sN \left(\pi, M \right)} t_{mix}(P, \epsilon,\mu).
\end{equation}
In order to give explicit quantitative bounds on \eqref{absolute_warm_start}, we will assume that $P$ is $\pi$-reversible, i.e. $\pi(\d \x)P(\x, \d \y) = \pi(\d \y) P(\y, \d \x)$, and positive semi-definite, meaning that
\[
\int \left(Pf(\x)\right)f(\x)\pi(\d \x) \geq 0, \quad Pf(\x) = \int f(\y) P(\x, \d \y),
\]
for every $f \, : \, \sX \, \to \, \sX$ such that $\int f^2(\x) \pi(\d \x) < \infty$.  The latter are common assumptions, which are often satisfied in practice. A typical strategy to determine the convergence properties of a reversible and positive semi-definite Markov chain is to study the \textbf{$s$-conductance} of $P$, i.e.
\begin{align}\label{s_conductance}
\Phi_s(P) &= \inf \left\{\frac{P(\partial A)
}{\pi(A) }; \, s < \pi(A) \leq \frac{1}{2}, A \subset \sX \right\}\,,
&
P(\partial A)=\int_A P(\x, A^c)\pi(\d \x)
\end{align}
with $s \in[0,1/2)$. 
If $s = 0$, we write $\Phi(P) := \Phi_0(P)$ and call it \textbf{conductance}. Also, $P(\partial A)$ is sometimes called the flux of $P$ through $A$ and $P(\partial A)/\pi(A)$ coincides with the probability that the Markov chain exits $A$ in one step, given that it starts from $\pi$ restricted to $A$. It is well known (see e.g. Corollary $1.5$ in \cite{LS93}) that a strictly positive conductance implies exponential convergence of $P^t$ to $\pi$ and thus a bound on the mixing times. This is summarized in the following lemma.
\begin{lemma}\label{conductance_convergence}
Let $P$ be a $\pi$-reversible and positive semi-definite Markov transition kernel. For every $s \in [0, 1/2)$, $t \geq 0$, $M \geq 1$ and $\mu \in \sN\left(\pi, M \right)$, it holds
\[
\lTV \mu P^t - \pi \rTV \leq Ms + M\left(1-\frac{\Phi^2_s(P)}{2} \right)^t\,.
\]
In particular, if $s = \frac{\epsilon}{2M}$ we have
\[
t_{mix}(P, \epsilon, M) \leq \frac{\log(2M)-\log(\epsilon)}{-\log \left(1-\frac{\Phi^2_s(P)}{2} \right)}.
\]
\end{lemma}
\begin{remark}
The usual definition of $s$-conductance (see e.g. \cite{LS93,D19}) is slightly different with respect to the one given in \eqref{s_conductance}. 
Appendix \ref{alternative_definition} provides results for the alternative definition.
Definition \eqref{s_conductance} is less tight but leads to neater formulas (see e.g. Theorem \ref{main}). 
\end{remark}
\begin{remark}
While being common in the literature, see e.g. \cite{D17, D19, T22}, the warm start assumption can be quite stringent especially in a high dimensional context. Thus it is often of interest to provide mixing times bounds for \emph{feasible} starts, which can be explicitly sampled from (e.g. \cite{D19}). In Section \ref{sec:hierarchical} we provide a feasible start for hierarchical models. 
\end{remark}

\section{Coordinate-wise MCMC}\label{sec:main}
Let $\sX = \sX_1 \times \dots \times \sX_d$ be a product space. Given $\x = (x_1, \dots, x_d)\in\sX$, we denote by $\x_{-i}=(x_j)_{j\neq i}\in \sX_{-i}=\times_{j\neq i}\sX_j $ the vector $\x$ without the $i$-th element and by $\pi_i (\d x_i \mid \x_{-i} )$ the conditional distribution of $x_i$ given $\x_{-i}$ under $\x\sim \pi$.

In this paper we focus on \emph{coordinate-wise kernels}, defined as
\begin{equation}\label{eq:coordinate_wise}
P=\frac{1}{d} \sum_{i=1}^dP_i, \quad P_i(\x, \d \y) = P_i^{\x_{-i}}(x_i, \d y_i)\delta_{\x_{-i}}\left( \d\y_{-i}\right),
\end{equation}
where $P_i^{\x_{-i}}$ is a $\pi_i (\cdot | \x_{-i} )$-invariant Markov kernel on $\sX_i$ and thus, by construction, $P_i$ is invariant with respect to $\pi$. 
At each iteration, a Markov chain evolving according to $P$ picks  uniformly at random a coordinate $i$ in $\{1,\dots,d\}$ and then updates $x_i$ sampling from $P_i^{\x_{-i}}(x_i,\cdot)$, while leaving $\x_{-i}$ unchanged. 
Lemma \ref{lemma:rev_pos_supp} in the Appendix shows that reversibility and positive semi-definiteness of $P$ are implied by the ones of $P_i$, which are in turn equivalent to $\pi(\cdot \mid \x_{-i})$-reversibility and positive semi-definiteness of $P_i^{\x_{-i}}$ for almost every $\x_{-i}$. Thus, in the following we will always silently assume that $P_i^{\x_{-i}}$ is reversible and positive semi-definite, for every $i$ and $\x_{-i}$. The latter are common assumptions which are often satisfied: for example they hold for the operators $P^{\x_{-i}}$ defined in Section \ref{section_log_concave}. Moreover, every reversible operator $Q$ can be made positive semi-definite by considering its \emph{lazy} version, i.e. $\tilde{Q} = \frac{1}{2}I+\frac{1}{2}Q$ with $I$ being the identity operator \citep{LS93}.

An important special case of \eqref{eq:coordinate_wise} is the so-called random scan Gibbs sampler, whose kernel is defined as $G = \frac{1}{d}\sum_{i = 1}^dG_{i}$ where $G_i$ is the kernel that performs exact sampling from $\pi_i (\cdot | \x_{-i} )$, i.e.\
\begin{equation}\label{eq:exact_updates}
G_{i}\left(\x, \d \y\right) = G_i^{\x_{-i}}(x_i, \d y_i)\delta_{\x_{-i}}\left( \y_{-i}\right), \quad G_i^{\x_{-i}}(x_i, \d y_i) = \pi_i(\d y_i \mid \x_{-i}).
\end{equation}
\begin{remark}
Another popular version of coordinate-wise methods, called \emph{deterministic scan}, updates the coordinates according to a pre-specified order. The relationship between different deterministic and random orders is not univocal: even if we expect them to behave in a similar way in most cases, for some specific problems the associated mixing times can significantly vary (see \cite{roberts2015surprising, he2016scan, gaitonde2024comparison}). In this work we restrict to random scan schemes, which somewhat simplify the mathematical treatment: for example, random scan Gibbs as in \eqref{eq:exact_updates} is always positive semi-definite.
\end{remark}
In order to relate $\Phi_s(P)$ to $\Phi_s(G)$ for general coordinate-wise kernels $P$, we introduce the following notion of \emph{conditional conductance} of $P$.
We define
\begin{align}
&
\kappa(P, K) = \min_{i \in \{1, \dots, d\}} \, \kappa_i(P_i, K)
\,,\quad 
\kappa_i(P_i, K) = \inf_{\x \in K}\, \kappa\left(P_i^{\x_{-i}}\right) 
&K\subseteq \sX
\label{minimum_conditional_conductance}
\end{align}
where $\kappa\left(P_i^{\x_{-i}}\right)$ is the conductance of the kernel $P_i^{\x_{-i}}$ on $\sX_{i}$, defined as
\begin{align}
&\kappa\left(P_i^{\x_{-i}}\right) =  
\inf \left\{ \frac{\int_{B} P^{\x_{-i}}_i(x_i, B^c) \pi_i (dx_i|\x_{-i})
}{\pi_i(B \mid \x_{-i} )}; \,  \pi_i(B \mid \x_{-i}) \in\left(0,\frac{1}{2}\right], B \subset \sX_i \right\}
&\x_{-i}\in\sX_{-i}\,.
\label{conditional_conductance}
\end{align}
The latter measures how much the invariant update of $P_i^{\x_{-i}}$ is close to exact sampling as in $G_i^{\x_{-i}}$.

 \begin{remark}
$\kappa_i(P_i,\sX)$ should not be confused with the conductance of $P_i$ on $\sX$. Indeed the latter is always equal to $0$, since $P_i$ leaves $\x_{-i}$ unchanged.
Similarly, $\kappa(P,\sX)$ is not the conductance of $P$ on $\sX$, since it only measures the quality of the conditional updates, not directly the convergence speed of $P^t$ to $\pi$.
\end{remark}

The next theorem provides a connection between the flux of $G$ and the flux of $P$, for an arbitrary coordinate-wise operator $P$, in terms of the conditional conductance. 
\begin{theorem}\label{main}
Let $P$ be as in \eqref{eq:coordinate_wise}.
For every $A,K \subset \sX$ and  $i = 1, \dots, d$ we have
\[
P_{i}(\partial A) \geq \kappa_i(P_i, K)\left(G_{i}(\partial A)-\pi(A \cap K^c)\right)\,.
\]
Moreover, $G_i(\partial A) \geq P_{i}(\partial A)$.
\end{theorem}
\begin{remark}
Theorem \ref{main} with $K = \sX$ implies
$$G_i(\partial A) \geq P_{i}(\partial A) \geq \kappa_i(P_i, \sX)G_{i}(\partial A)\,,$$ so that $\kappa(P_i)$ controls 
how much flux is lost when passing from exact to invariant updates on the $i$-th coordinate. The extension to a generic $K\subseteq \sX$ allows to ignore ``bad'' sets which have low probability under $\pi$ (see also Remark \ref{asymptotic_cond_Gibbs}). This turns out to be crucial in some statistical applications, see e.g.\ Section \ref{sec:hierarchical} for the case of hierarchical models.
\end{remark}
We can use Theorem \ref{main} to bound $\Phi_s(P)$ in terms of $\Phi_s(G)$, as detailed in the next corollary.
\begin{corollary}\label{cor:bound_conductance}
Let $P$ be as in \eqref{eq:coordinate_wise}. Then for every $s \in [0, 1/2)$ and $K\subseteq \sX$ we have
\begin{align}\label{eq:cond_bound}
\Phi_s(P)& \geq 
\kappa(P, \sX)\Phi_{s}(G) \quad \text{and} \quad \Phi_s(P) \geq 
\kappa(P,K)\Phi_{s}(G)-\frac{\pi(K^c)}{s}\left(\frac{1}{d}\sum_{i = 1}^d\kappa_i(P_i,K)\right)\,,
\end{align}
with $\kappa(P, K)$ as in \eqref{minimum_conditional_conductance}.
Moreover, $\Phi_s(G) \geq \Phi_s(P)$.
\end{corollary}
The inequalities in \eqref{eq:cond_bound} quantify the loss of efficiency incurred by substituting an exact Gibbs update with a $\pi_i (\cdot | \x_{-i} )$-invariant one, provided the conditional conductance is uniformly bounded away from zero. Crucially, the bound is informative even if the conditional conductance is controlled only on a set $K$, provided $\pi(K^c)/s$ is small. 
\begin{remark}
The first inequality in \eqref{eq:cond_bound} is tight, in the sense that for every $c \in [0,1]$ there exists an operator $P$ such that $\kappa(P, \sX)=c$ and $\Phi_s(P) = c\,\Phi_{s}(G)$. A simple example is given by $P_i = cG_i+(1-c)I$ for all $i=1,\dots,d$. See Section \ref{tightness_example} in the Appendix for details.
\end{remark}
\begin{remark}
The results of Corollary \ref{cor:bound_conductance} extend to the case where the update probabilities in \eqref{eq:coordinate_wise} are not uniform. See Corollary \ref{cor:bound_conductance_nonuniform} in the Appendix for a precise statement.
\end{remark}
The bound in \eqref{eq:cond_bound} shares some similarity with the ones in \cite{madras2002markov}, where they decompose a general Markov chain in different, easier to analyze, pieces (see in particular their Theorem $1.1$). Their context, though, is very different, since they are motivated by multi-modal problems and do not consider coordinate-wise schemes.

\section{Applications to Metropolis-within-Gibbs schemes}\label{sec:mwg}
As a first application of the theory developed above, and in particular of Corollary \ref{cor:bound_conductance}, we consider Metropolis-within-Gibbs (MwG) schemes. 
MwG are popular MCMC algorithms that replace the exact conditional updates of Gibbs schemes with $\pi_i (\cdot | \x_{-i} )$-invariant Metropolis-Hastings (MH) updates, thus making coordinate-wise MCMC algorithms applicable to general models as opposed to only conditionally conjugate ones. 
Random-scan MwG kernels take the form of \eqref{eq:coordinate_wise}, with $P^{\x_{-i}}_i$ defined as
\begin{align}\label{eq:operator_MH}
P^{\x_{-i}}_i(x_i, A) 
&=
 \int_A\alpha_i^{\x_{-i}}(x_i, y_i)Q_i^{\x_{-i}}(x_i,\d y_i)+\delta_{x_i}(A)\int \left[1-\alpha_i^{\x_{-i}}(x_i, y_i)\right]Q_i^{\x_{-i}}(x_i,\d y_i)
& A \subset \sX_i\,,
\end{align}
where $\alpha_i^{\x_{-i}}$ is the MH acceptance probability with target $\pi_i (\cdot | \x_{-i} )$ and proposal $Q_i^{\x_{-i}}$, i.e.\ 
\begin{align*}
\alpha^{\x_{-i}}_i(x_i, y_i) 
&=
 \min \left\{ 1, \frac{ 
 \pi_i(\d y_i \mid \x_{-i}) Q_i^{\x_{-i}}(y_i, dx_i)
 }{
 \pi_i(\d x_i \mid \x_{-i}) Q_i^{\x_{-i}}(x_i, dy_i)
 }\right\}
 &x_i,y_i\in\sX_i\,.
\end{align*}
MwG schemes are an instance of coordinate-wise kernels and thus \eqref{eq:cond_bound} applies to them. Moreover, it is well-known that any operator as in \eqref{eq:operator_MH} is reversible (e.g.\ Theorem $7.2$ in \cite{robert1999monte}). We now consider two instances of MwG schemes where the conditional conductance can be controlled, namely independent MH conditional updates and conditionally log-concave target distributions.

\subsection{Conditional updates with independent Metropolis-Hastings}\label{sec:ind}
We first consider MwG schemes with conditional updates performed using Independent Metropolis-Hastings (IMH), meaning that the MH proposal kernel does not depend on $x_i$. Thus we assume (with a slight abuse of notation)
\begin{align}\label{eq:IMH}
Q_i^{\x_{-i}}(x_i,\d y_i)&=Q_i^{\x_{-i}}(\d y_i)&\x\in\sX,\,i=1,\dots,d
\end{align}
for some $Q_i^{\x_{-i}} \in \mathcal{P}(\sX_i)$.
Despite its simplicity, MwG with IMH proposals is routinely used in various contexts, e.g.\ Bayesian inference for diffusions example of Section \ref{sec:diffusions} \citep{roberts2001inference, beskos2006retrospective, papaspiliopoulos2013data}. Moreover, a MwG operator as in \eqref{eq:IMH} is always positive semi-definite, see Lemma $3.1$ in \cite{baxendale2005renewal}.
The next proposition shows that an upper bound on the Radon-Nykodym derivative between $\pi_i(\cdot \mid \x_{-i})$ and $Q_i^{\x_{-i}}$, i.e.\
\begin{align}\label{eq:bounded_deriv}
\frac{ \pi_i(\d x_i \mid \x_{-i})}{Q_i^{\x_{-i}}(\d x_i)} &\leq M
&\x\in\sX,\,i=1,\dots,d\,,
\end{align}
directly implies a lower bound on the conductance of $P = \frac{1}{d}\sum_{i = 1}^dP_i$.
\begin{proposition}\label{prop:conductance_ind}
Define $P_i^{\x_{-i}}$ and $Q_i^{\x_{-i}}$ as in \eqref{eq:operator_MH} and \eqref{eq:IMH}. If \eqref{eq:bounded_deriv} holds, then $\Phi(P) \geq \frac{1}{M}\Phi(G)$.
\end{proposition}
Proposition \ref{prop:conductance_ind} is a consequence of Corollary \ref{cor:bound_conductance} and the fact that \eqref{eq:bounded_deriv} implies a bound on the conditional conductance of the independent Metropolis-Hastings updates.

While it is well-known that IMH suffers from the curse of dimensionality (i.e.\ performances often deteriorate exponentially with $d$ if the update is jointly applied to all coordinates, see e.g.\ Figure \ref{fig:binom_increasingcovariates}, or \cite[Prop.2]{deligiannidis2018ergodic} for lower bounds on IMH asymptotic variances), Proposition \ref{prop:conductance_ind} implies that MwG schemes with IMH proposals only pay a constant slowdown relative to exact Gibbs if the dimensionality of each $\sX_i$ is fixed. This result will be used in Section \ref{sec:diffusions} to provide a lower bound on the conductance of a data augmentation scheme for discretely observed diffusions.

\subsection{Conditionally log-concave distributions}\label{section_log_concave}
In this section we consider the case where the conditional distributions $\pi_i(\cdot \mid \x_{-i})$ are 
strongly log-concave and smooth. 
More specifically, we take $\sX=\R^d$ and assume that the target distribution admits a density $\pi(\x)$ with respect to the Lebesgue measure such that 
$f(\x) := -\log \pi(\x)$ satisfies
\begin{align}\label{eq:m_convex_L_smooth}
\frac{m_i(\x_{-i})}{2}\lE x_i - y_i \rE^2&\leq f(\y)-f(\x)-\nabla f(\x)^\top(\y-\x) \leq \frac{L_i(\x_{-i})}{2}\lE x_i - y_i \rE^2
\end{align}
for all $\x,\y \in \R^d$ such that $\x_{-i}=\y_{-i}$.
The ratio $c_i(\x_{-i}) = L_i(\x_{-i})/m_i(\x_{-i})$ is the condition number of the conditional distribution $\pi_i(\cdot \mid \x_{-i})$.

Log-concavity and smoothness are common assumptions in the MCMC literature \cite{D17, DM17}, under which bounds on the conductance of various MCMC algorithms are available, see e.g.\ \citep{D19, AL22}. 
Combining bounds available in the literature with Corollary \ref{cor:bound_conductance} one can bound the conductance of MwG in terms of the Gibbs one. 

Consider for example targeting $\pi_i(\cdot \mid \x_{-i})$ using the IMH algorithm 
with proposal distribution
\begin{equation}\label{proposal_ind_logconcave}
Q_i^{\x_{-i}}(\d y_i) = N\left(y_i ; x_i^*(\x_{-i}), m_i(\x_{-i})^{-1}\I_{d_i}\right) \d y_i,
\end{equation}
where $N(y;\mu,\Sigma)$ denotes the density of a Gaussian with mean $\mu$ and covariance $\Sigma$ evaluated at $y$, 
$\I_d$ denotes the $d\times d$ identity matrix and $x_i^*(\x_{-i})$ denotes the mode of $\pi_i(\cdot \mid \x_{-i})$. The existence and uniqueness of $x_i^*(\x_{-i})$ follows by the log-concavity of $\pi_i(\cdot \mid \x_{-i})$.
\begin{proposition}\label{prop:IMH_logconcave}
Assume \eqref{eq:m_convex_L_smooth}. 
For  
$P_i^{\x_{-i}}$ and $Q_i^{\x_{-i}}$ defined as in 
\eqref{eq:operator_MH} and \eqref{proposal_ind_logconcave},
we have
\begin{align}\label{cond_indMH_logconcave}
\kappa \left(P_i^{\x_{-i}} \right) &\geq c_i(\x_{-i})^{-\frac{d_i}{2}}
&\x\in\R^d\,,
\end{align}
which implies
\[
\Phi(P) \geq \left[\min_{i}\inf_{\x_{-i}}\, c_i(\x_{-i})^{-\frac{d_i}{2}}\right]\Phi(G).
\]
\end{proposition}
\begin{remark}
The dependence of the conductance on the dimensionality of the single block is the one we expect from the usual theory on Independent Metropolis-Hastings schemes, i.e. the bound decreases exponentially fast in $d_i$. Thus, the best blocking of the variables in terms of conductance is a trade-off between the conditional conductance in \eqref{cond_indMH_logconcave}, which  deteriorates by increasing $d_i$, and the conductance $\Phi(G)$ of the associated Gibbs sampler, which typically increases by increasing $d_i$ (see e.g.\ \cite[Thm.1]{liu1994collapsed} or \cite[Sec.2.4]{R97}).
\end{remark}

We now consider the case of MwG with conditional updates performed through Random-Walk Metropolis (RWM).
\begin{proposition}\label{prop:RWM_logconcave}
Assume \eqref{eq:m_convex_L_smooth} for all $i=1,\dots,d$.
Define $P_i^{\x_{-i}}$ as in \eqref{eq:operator_MH} with 
\begin{equation}\label{eq:rwm_kernel}
    Q_i^{\x_{-i}}(x_i, \d y_i) = N\left(y_i ; x_i, \sigma_i^2(\x_{-i})\I_{d_i}\right)\d y_i, \quad \sigma_i^2(\x_{-i}) = \frac{\eta}{d_iL_i(\x_{-i})},
\end{equation}
where $\eta > 0$ is a fixed constant. Then there exists a constant $M = M(\eta)$ depending only on $\eta$ such that
\begin{align}\label{eq:lower_bound_AL22}
\kappa(P_i^{\x_{-i}}) &\geq M\sqrt{\frac{1}{d_ic_i(\x_{-i})}},&\x\in\R^d\,,
\end{align}
which implies
\[
\Phi(P) \geq M\sqrt{\min_i \inf_{\x_{-i}}\,\frac{1}{d_ic_i(\x_{-i})}}\Phi(G).
\]
\end{proposition}

Also in this case the best blocking scheme, in terms of conductance, is given by a trade-off between the conditional conductance and the behaviour of the Gibbs sampler. However, the dependence on $d_i$ is polynomial, rather than exponential as in the case of Proposition \ref{prop:IMH_logconcave}, which makes RWM more robust to the dimensionality of each coordinate. Positive semi-definiteness of MwG operators as in \eqref{eq:rwm_kernel} follows again by Lemma $3.1$ in \cite{baxendale2005renewal}.

Note that \eqref{eq:m_convex_L_smooth} is implied by global log-concavity and smoothness, but it is a weaker requirement than that.
Thus, Proposition \ref{prop:IMH_logconcave} and \ref{prop:RWM_logconcave} can in principle apply to cases where standard MCMC results based on global log-concavity and smoothness do not hold. 
One interesting example is given by models whose density is log-concave and smooth conditional on some low-dimensional hyperparameters but not jointly.  
In those cases Proposition \ref{prop:IMH_logconcave} and \ref{prop:RWM_logconcave} allow to relate MwG to exact Gibbs, which can then be analyzed with different techniques.
Section \ref{sec:log_regr_hyper} below considers one such example.

\section{Auxiliary results for statistical applications}\label{sec:auxiliary}
In order to successfully apply Theorem \ref{main} and Corollary \ref{cor:bound_conductance} to Markov chains arising in common statistical applications, such as Bayesian hierarchical models considered in Section \ref{sec:hierarchical}, we need some auxiliary results dealing with approximate conductances and perturbation of Markov operators, which can be of independent interest. These are described in Section \ref{sec:conductance_distance}.
We also recall some facts about the conductance of product operators in Section \ref{section:conductance_product}.

\subsection{Approximate conductance and perturbations of Markov operators}\label{sec:conductance_distance}
Consider two generic Markov kernels $P_1$ and $P_2$ with invariant distributions $\pi_1$ and $\pi_2$. We define a notion of discrepancy between $P_1$ and $P_2$ as follows
\begin{equation}\label{delta_operators}
\Delta(P_1, P_2, M) = \sup_{\mu \in \sN\left(\pi_1, M \right)} \, \lTV \mu P_1-\mu P_2 \rTV,
\end{equation}
with $\sN\left(\pi_1, M \right)$ as in \eqref{N_class}.  
Note that the above discrepancy is not symmetric, i.e.\ $\Delta(P_1, P_2, M)\neq\Delta(P_2, P_1, M)$.
The next theorem shows how to relate $\Delta(P_1, P_2, M)$ with the difference between $\Phi_s(P_1)$ and $\Phi_{s'}(P_2)$ for $s'$ close to $s$.
\begin{theorem}\label{theorem:distance_operators}
Let $P_1$ and $P_2$ be transition kernels with invariant distributions $\pi_1$ and $\pi_2$, and $\delta = \lTV \pi_1-\pi_2 \rTV$. 
For $s \geq \delta$, we have
\[
\Phi_s(P_1) \geq \Phi_{s-\delta}(P_2)-\Delta\left(P_1, P_2, 1/s \right)-\frac{2\delta}{s}\,.
\]
\end{theorem}
Therefore,  knowledge on $\Phi_{s-\delta}(P_2)$ can be used to bound $\Phi_s(P_1)$, with two additional terms measuring respectively the distance between the operators and between the associated invariant distributions. 
Theorem \ref{theorem:distance_operators} significantly simplifies for the case of the Gibbs sampler, as we now show.
Let $G_1$ and $G_2$ be (random scan) Gibbs sampler kernels targeting $\pi_1$ and $\pi_2$, meaning that 
$$G_1=\frac{1}{d}\sum_{i=1}^dG_{1,i}\hbox{ and  }G_2=\frac{1}{d}\sum_{i=1}^dG_{2,i}\,,$$ with $G_{1,i}$ and $G_{2,i}$ defined as $G_i$ in \eqref{eq:exact_updates} with $\pi$ replaced by, respectively, $\pi_1$ and $\pi_2$. 
One can control the discrepancy between $G_{1,i}$ and $G_{2,i}$ in terms of $\delta = \lTV \pi_1-\pi_2 \rTV$ as follows. The proof is similar to the one of Proposition $2.2$ in \cite{AZ23}, where deterministic-scan Gibbs sampler is considered (see also \cite{CJ23} for a general approach to bound the total variation distance between Markov kernels in terms of the stationary distributions).
\begin{lemma}\label{lemma:Delta_GS}
For every $M \geq 1$ and $i = 1, \dots, d$ we have
\begin{align}\label{eq:bound_Delta_GS}
\Delta\left(G_{1,i}, G_{2,i}, M \right) &\leq 2M\delta\,.
\end{align}
\end{lemma}
Combining Lemma \ref{lemma:Delta_GS} and Theorem \ref{theorem:distance_operators} we obtain the following result.
\begin{theorem}\label{theorem:delta_Gibbs_operators}
For $s \geq \delta$, we have
\[
\Phi_s(G_1) \geq \Phi_{s-\delta}(G_2)-\frac{4\delta}{s}.
\]
\end{theorem}

\begin{remark}\label{asymptotic_cond_Gibbs}
An interesting byproduct of Theorem \ref{theorem:delta_Gibbs_operators} is that, if $G_n$ and $G$ are Gibbs samplers targeting $\pi_n$ and $\pi$ respectively, then $\lTV \pi_n-\pi \rTV \to 0$  as $n\to\infty$ implies $$\lim \inf_n \Phi_s(G_n) \geq \Phi_{s-\delta}(G)$$
for every $\delta>0$. This crucially relies on $s>0$, that is on using the approximate version of the conductance: for $s=0$ it is not true in general that $\lTV \pi_n-\pi \rTV \to 0$ implies $\lim \inf_n \Phi(G_n)$ greater than $\Phi_{s-\delta}(G)$.
Section \ref{example_asymp_conduct} in the Appendix provides a simple example where $\lTV \pi_n-\pi \rTV \to 0$ and $\Phi(G) = 1$, but $\Phi(G_n) = 0$ for every $n$.
\end{remark}
In Section \ref{sec:hierarchical} we will use Theorem \ref{theorem:delta_Gibbs_operators} to prove that, under suitable assumptions, the approximate conductance of Gibbs samplers targeting Bayesian hierarchical models remains bounded away from zero as the number of parameters and datapoints diverge (Theorem \ref{theorem: gibbs_one_level_nested}).

\subsection{Conductance and independent products}\label{section:conductance_product}
The next lemma, which is a simple consequence of Cheeger inequality, provides a lower bound to the conductance of a product of independent kernels. This will be useful in Section \ref{sec:hierarchical} to bound the conditional conductance of MwG schemes for models with a high degree of conditional independence.
\begin{lemma}\label{lemma_product_conductance}
Let $P_j$ be a $\pi_j$-invariant kernel with $\pi_j\in\sP(\sX_j)$, for $j=1,\dots,J$; 
and $P = \otimes_{j = 1}^JP_j$ be the corresponding product kernel on $\sX=\times_{j = 1}^J\sX_j$. 
Then
\[
\Phi(P) \geq \min_{j \in \{1, \dots, J\}} \, \frac{\Phi^2(P_j)}{4}\,.
\]
\end{lemma}
\begin{remark}
With the result above we lose a factor of $2$, passing from $\Phi(P_j)$ to $\Phi^2(P_j)$. 
While finer results are available under additional assumptions \cite{B97}, the bound in Lemma \ref{lemma_product_conductance} suffices for our purposes.
\end{remark}

\section{High-dimensional hierarchical models}\label{sec:hierarchical}
In this section we analyze the performances of coordinate-wise MCMC targeting Bayesian hierarchical models, in a high-dimensional regime where both the number of datapoints and parameters increase.

We consider a general class of hierarchical models, with data divided in $J$ groups, each having a set of group-specific parameters $\theta_j$. The latter share a common prior with hyper-parameters $\psi$.
Thus, we assume the following model:
\begin{equation}\label{one_level_nested}
\begin{aligned}
Y_j \mid \theta_j & \sim f(\cdot \mid \theta_j) & j = 1, \dots, J,\\
 \theta_j\mid \psi &\overset{\text{iid}}{\sim} p(\cdot \mid \psi)& j = 1, \dots, J,\\
 \psi &\sim p_0(\cdot).&
\end{aligned}
\end{equation}
We assume that the prior for
$\theta_j \in \R^{\ell}$ belongs to the exponential family, that is
\begin{equation}\label{exponential_family}
p(\theta \mid \psi) = h(\theta)\text{exp}\left\{ \sum_{s = 1}^S\eta_s(\psi)T_s(\theta)-A(\psi)\right\},
\end{equation}
where $\psi \in \R^D$, $h \, : \, \R^{\ell} \, \to \, \R_+$ is a non-negative function and $\eta_s(\psi)$, $T_s(\theta)$ and $A(\psi)$ are known real-valued functions with domains $\R^D$, $\R^\ell$ and $\R^D$ respectively.
We let $f(y \mid \theta)$ be an arbitrary likelihood function with data $y\in\sY\subseteq\R^m$ and parameters $\theta \in \R^{\ell}$, dominated by a suitable $\sigma$-finite measure (usually Lebesgue or counting one). 

\subsection{Gibbs and MwG kernels}
When sampling from the posterior distribution of model \eqref{one_level_nested}, it is natural to consider coordinate-wise MCMC schemes that alternate update from the full conditionals of local and global parameters. 
Denoting $\bm{\theta} = (\theta_1, \dots, \theta_J)$, $Y_{1:J} = \left(Y_1, \dots, Y_J\right)$ and $\pi_J(\d\psi, \d\bm{\theta}):=\L\left(\d\psi, \d\bm{\theta} \mid Y_{1:J}\right)$, the transition kernel of the exact two-block GS targeting $\pi_J(\d\psi, \d\bm{\theta})$ is defined as
\begin{equation}\label{two_blocks_gibbs_nested}
G_J = \frac{1}{2}G_{1,J}+\frac{1}{2}G_{2, J}
\end{equation}
with
\[
 G_{1,J}\left((\psi, \bm{\theta}), \left(\d \psi', \d \bm{\theta}' \right)\right) 
 =
 \pi_J\left(\d \psi' \mid \bm{\theta}\right)\delta_{\bm{\theta}}(\d\bm{\theta}'),\;
 G_{2,J}\left((\psi, \bm{\theta}), \left(\d \psi', \d \bm{\theta}' \right)\right) 
 = 
 \pi_J\left(\d \bm{\theta}' \mid \psi\right)\delta_{\psi}(\d\psi').
 \]
Sampling from $G_J$, however, requires  $\pi_J(\d \theta_j \mid \psi)$ to be available in closed form and amenable to exact sampling, which is typically feasible only for conditionally conjugate models. 
 For non-conjugate models, more broadly applicable coordinate-wise MCMC methods (such as MwG schemes) are typically used. 
 The corresponding kernel is
\begin{equation}\label{two_blocks_MH_nested}
P_J = \frac{1}{2}P_{1,J}+\frac{1}{2}P_{2, J},
\end{equation}
where
\begin{align*}
&P_{1,J}\left((\psi, \bm{\theta}), \left(\d \psi', \d \bm{\theta}' \right)\right)) 
= P^{\bm{\theta}}\left(\psi, \d \psi' \right)\delta_{\bm{\theta}}\left(\d\bm{\theta}'\right),\\
&P_{2,J}\left((\psi, \bm{\theta}), \left(\d \psi', \d \bm{\theta}' \right)\right)) =
P^{\psi, \bm{Y}}\left(\bm{\theta}, \d \bm{\theta}' \right)\delta_{\psi}\left(\d\psi'\right),
\end{align*}
with $P^{\bm{\theta}}$ 
and 
$P^{\psi, \bm{Y}}$ 
 being, respectively, $\pi_J(\d \psi \mid \bm{\theta})$ and $\pi_J(\d \bm{\theta} \mid \psi)$-invariant
 transition kernel.

In order to relate the convergence properties of $P_J$ to the ones of $G_J$ we need to control the conditional conductance of $P_J$. 
To do that, we can leverage the conditional independence of $(\theta_1,\dots,\theta_J)$ given $\psi$ under $\pi_J$, which implies that we can take a factorized $P^{\psi, \bm{Y}}$ defined as
\begin{equation}\label{eq:prod_Y_kernels}
P^{\psi, \bm{Y}}\left(\bm{\theta}, \d \bm{\theta}' \right)
=
\prod_{j = 1}^JP^{\psi, Y_j}\left(\theta_j, \d \theta_j' \right)
\end{equation}
with $P^{\psi, Y_j}$
 being a $\pi_J(\d \theta_j \mid \psi)$-invariant kernel. 
Thus, by Lemma \ref{lemma_product_conductance}
\begin{equation}\label{eq:product_cond_conductance}
\kappa(P^{\psi, \bm{Y}})
\geq
\min_{j \in \{1, \dots, J\}} 
\frac{\kappa^2(P^{\psi, Y_j})}{4}\,.
\end{equation}
Condition (C) below imposes a lower bound on $\kappa(P^{\psi, Y_j})$ for $\psi$ belonging to some appropriate set. 
We also require a lower bound on $\kappa(P^{\bm{\theta}})$. This is usually less critical since $\psi$ is low-dimensional (i.e.\ of fixed dimensionality not growing with $J$) and $\pi_J(\d \psi \mid \bm{\theta})$ is often available in closed form due to \eqref{exponential_family}, in which case one can perform exact Gibbs updates (i.e.\ take $P^{\bm{\theta}}\left(\psi, \d \psi' \right)=\pi_J\left(\d \psi' \mid \bm{\theta}\right))$.

\begin{remark}
Beyond being interesting per se, relating the convergence properties of $P_J$ to the ones of $G_J$ is theoretically appealing because the latter is potentially much easier to analyse, using for example the dimensionality reduction approach discussed in \citep[Lemma 4.2]{AZ23}. 
On the contrary $P_J$ does not enjoy such dimensionality reduction property and thus performing a direct analysis of it is in principle much harder.
\end{remark}

\subsection{Statistical assumptions}
We now describe the statistical assumptions we require on the data-generation process and likelihood of model \eqref{one_level_nested}, in order to an analyze the performances of $P_J$ as $J\to\infty$. 

We denote the marginal likelihood of the model, obtained by integrating out the group specific parameter $\theta$, as
\begin{equation}\label{likelihood_data}
g(y \mid \psi) = \int_{\R^\ell}f(y \mid \theta)p(\theta \mid \psi) \, \d \theta
\end{equation}
and its Fisher Information matrix as
\[
\left[\Fisher(\psi)\right]_{d, d'} = E \biggl[\left\{\partial_ {\psi_d} \log g(Y \mid \psi)\right\} \, \left\{\partial_{\psi_{d'}} \log g(Y \mid \psi)\right\} \biggr], \quad d, d' = 1, \dots, D.
\]
We denote by $Q_{\psi} \in \mathcal{P}(\R^m)$ the probability measure with density $g(y \mid \psi)$, and the associated product measures as $Q_{\psi}^{(J)}$ or $Q_{\psi}^{(\infty)}$. 
Following \cite{AZ23} we consider the following assumptions:
\begin{enumerate}
\item[$(B1)$] There exists $\psi^* \in \R^D$ such that
$
Y_j \simiid Q_{\psi^*}$ for $j = 1,2, \dots$.
Moreover the map $\psi \to g(\cdot \mid \psi)$ is one-to-one and the map $\psi \to \sqrt{g(x \mid \psi)}$ is continuously differentiable for every $x$.  Finally, the prior density $p_0$ is continuous and strictly positive in a neighborhood of $\psi^*$.
\item[$(B2)$] There exist a compact neighborhood $\Psi$ of $\psi^*$ and a sequence of tests $u_j \, : \, \R^{mJ} \, \to \, [0,1]$ such that
$\int_{\R^{mJ}}u_j\left(y_1, \dots, y_J\right) \prod_{j = 1}^Jg(y_j \mid \psi^*) \, \d y_{1:J} \to 0$
and\\
$
\sup_{\psi \not\in \Psi} \, \int_{\R^{mJ}}\left[ 1-u_j\left(y_1, \dots, y_J\right)\right] \prod_{j = 1}^Jg(y_j \mid \psi) \, \d y_{1:J} \to 0
$, as $J \to \infty$. 
\item[$(B3)$] The Fisher Information matrix $\Fisher(\psi)$ is non-singular and continuous w.r.t.\ $\psi$.
\end{enumerate}
Assumptions $(B1)$-$(B3)$ require model \eqref{one_level_nested} to be well-specified, in the sense that the marginal likelihood \eqref{likelihood_data} corresponds to the true data-generating mechanism. Moreover, the global parameters $\psi$ need to be identifiable, as formalized in $(B2)$ and $(B3)$: this guarantees that the posterior distribution on $\psi$ contracts to the true value $\psi^*$ at an appropriate rate (thanks to the Bernstein-von Mises Theorem , e.g. \cite{V00}).  
We also consider a set of more technical regularity assumptions on the likelihood $f(y\mid\theta_j)$, $(B4)$-$(B6)$, which are stated for completeness in Appendix \ref{regularity_assumptions}. Notice that assumptions $(B1)$-$(B6)$ are satisfied by common formulations of model \eqref{one_level_nested}, such as the hierarchical normal model or models for binary and categorical data (see e.g.\ Section $5$ in \cite{AZ23}).

\subsection{Dimension-free mixing of MwG for hierarchical models}\label{sec:dim_free}
The following condition requires the conditional conductance of $P_J$ to be bounded away from $0$ in a neighbourhood $S^*$ of $\psi^*$:
\begin{itemize}
\item[$(C)$] There exist a neighborhood $S^*$ of $\psi^*$ and $\kappa > 0$ such that $\kappa \left(P_J, S^*\times \R^{J\ell}\right)\geq \kappa$.
\end{itemize}
Note that the constant $\kappa>0$ above should not depend on $J$ or $Y_{1:J}$. 
By \eqref{eq:product_cond_conductance} condition (C) amounts to requiring
\begin{align}\label{eq:bound_cond_cond_hier}
\hbox{$\inf_{(\psi,Y_j)\in S^*\times\sY} \kappa \left(P^{\psi, Y_j} \right) \geq \kappa$ and 
$\inf_{\bm{\theta}\in \R^{J\ell}} \kappa \left(P^{\bm{\theta}} \right) \geq \kappa$.}
\end{align}
The key requirement in \eqref{eq:bound_cond_cond_hier} is that the kernel $P^{\psi, Y_j}$ is well-behaved in a neighbourhood of $\psi^*$, uniformly in $Y_j\in\sY$. Indeed, as mentioned after \eqref{eq:product_cond_conductance}, the update on the (low dimensional vector) $\psi$ can be often performed in closed form, so that the second requirement of \eqref{eq:bound_cond_cond_hier} is automatically satisfied. Section \ref{sec:discrete_data} provides some examples where these conditions are verified.

The following theorem bounds $\Phi_s(P_J)$ in terms of $\kappa$ and $\Phi_s(G_J)$.
\begin{theorem}\label{theorem_one_level_nested}
Consider model \eqref{one_level_nested}, kernels $G_J$ and $P_J$ defined in \eqref{two_blocks_gibbs_nested}-\eqref{two_blocks_MH_nested} and $s\in(0,1/2)$. Then, under assumptions $(B1)$-$(B3)$ and $(C)$, we have
\begin{align*}
&Q_{\psi^*}^{(J)}\left(\Phi_s(P_J) \geq \frac{\kappa^2}{8}\Phi_s(G_J) \right) \to 1
&\hbox{ as }J \to \infty\,.
\end{align*}
\end{theorem}
At an intuitive level, Theorem \ref{theorem_one_level_nested} states that two things are sufficient for $P_J$ to mix fast: first that $G_J$ mixes fast and second that the conditional conductance of $P_J$ around $\psi^*$ is good.
Motivated by Theorem \ref{theorem_one_level_nested}, the next theorem studies the behaviour of $\Phi_s(G_J)$ as $J\to\infty$. 
\begin{theorem}\label{theorem: gibbs_one_level_nested}
Consider model \eqref{one_level_nested} and $G_J$ defined in \eqref{two_blocks_gibbs_nested}. Then, under assumptions $(B1)$-$(B6)$, for every $s \in(0,1/2)$ there exists a constant $g(s) > 0$ such that
\begin{equation}\label{cond_gibbs_hier}
Q_{\psi^*}^{(J)}\left(\Phi_s(G_J)  \geq g(s) \right) \to 1,
\qquad \hbox{ as }J \to \infty\,.
\end{equation}
\end{theorem}
The inequality in \eqref{cond_gibbs_hier} holds with probability converging to one as $J\to\infty$, under the true data-generating mechanism. Equivalently, \eqref{cond_gibbs_hier} states that $\Phi_s(G_J)$, which is a random variable depending on $Y_{1:J}$, is bounded away from zero in $Q_{\psi^*}^{(\infty)}$-probability as $J \to \infty$.

\begin{remark}\label{rmk:approx_2}
Theorem \ref{theorem: gibbs_one_level_nested} crucially relies on the fact that the approximate version of the conductance is considered, i.e.\ that $s>0$. Indeed, the proof of Theorem \ref{theorem: gibbs_one_level_nested} exploits asymptotic characterizations on the posterior of $(\psi, \bm{\theta})$, which in general do not provide meaningful bounds for the case $s = 0$. This is connected with Remark \ref{asymptotic_cond_Gibbs}: the exact conductance with $s=0$ is not directly controlled by the convergence of the invariant distributions ($\pi_n\to\pi$), since it is sensitive to the behaviour of kernels over sets with arbitrarily small probability under $\pi$.
\end{remark}

 We can combine Theorem \ref{theorem_one_level_nested} with Theorem \ref{theorem: gibbs_one_level_nested} to deduce that the mixing times of $P_J$ are uniformly bounded in probability as $J\to\infty$.
\begin{corollary}\label{mixing_times_warm_hier}
Consider model \eqref{one_level_nested} and $P_J$ defined in \eqref{two_blocks_MH_nested}. Under assumptions $(B1)$-$(B6)$ and $(C)$, for every $M \geq 1$ and $\epsilon > 0$ there exists $T\left(\psi^*, \epsilon, M \right) < \infty$ such that
\begin{align*}
&Q_{\psi^*}^{(J)}\left(t_{mix}(P_J, \epsilon, M)
\leq T\left(\psi^*, \epsilon, M \right)\right) \to 1
&\hbox{as }J \to \infty.
\end{align*}
\end{corollary}
Corollary \ref{mixing_times_warm_hier} implies that 
\begin{align*}
t_{mix}(P_J, \epsilon, M)&=\sO_P(1)
&\hbox{as }J\to\infty\,,
\end{align*}
i.e.\ that the sequence of random variables $\{t_{mix}(P_J, \epsilon, M)\}_{J=1,2,\dots}$ is 
bounded in probability as $J\to\infty$.
This is in accordance with numerical evidences (see e.g.\ Figure \ref{fig:binom_simulations}).

\subsection{Computational complexity implications}\label{sec:complexity}
Sampling from $P_J$ defined in \eqref{two_blocks_MH_nested} requires $\mathcal{O}(J)$ computational cost. 
This follows from \eqref{eq:prod_Y_kernels} and the fact that sampling from each $P^{\psi, Y_j}\left(\theta_j, \d \theta_j' \right)$ involves an $\mathcal{O}(1)$ cost, since the conditional distributions $\pi_J(\d \theta_j \mid\psi)=\L(\d\theta_j|\psi, Y_j,)$ depend only on the local observations $Y_j$ and not on the whole dataset $Y_{1:J}$.
The cost of sampling from $P^{\bm{\theta}}\left(\psi, \cdot\right)$ depends on the specification of $P^{\bm{\theta}}$, but it is typically dominated by the $\mathcal{O}(J)$ cost of computing the conditional density $\pi_J(\d \psi \mid \bm{\theta})=\L(\d \psi \mid \bm{\theta}, Y_{1:J})$ once. Under the assumptions of Corollary \ref{mixing_times_warm_hier}, we obtain therefore that, with high probability with respect to the data generation process, the MwG algorithm with kernel $P_J$ produces a sample with $\epsilon$-accuracy in TV distance with $\mathcal{O}(J)$ cost when initialized from a warm start. 
Section \ref{sec:feasible_start} extends Corollary \ref{mixing_times_warm_hier} to the case where the algorithm is implemented starting from a specific feasible distribution, sampling from which requires $\mathcal{O}(J\log (J))$ cost (given access to the maximum marginal likelihood estimator). Thus, under the above assumption, the total cost required by MwG for each $\epsilon$-accurate sample, including the initialization, scales as $\mathcal{O}(J\log (J))$.

It is interesting to compare this cost with the one of standard gradient-based MCMC methods, such as Metropolis-Adjusted Langevin Algorithm (MALA) and Hamiltonian Monte Carlo (HMC). 
The cost required by a single full likelihood or gradient evaluation for model \eqref{eq:one_level_nested_intro} is $\mathcal{O}(J)$.
Available results suggest that the number of gradient evaluations required by MALA and HMC (as well as other $\pi_J$-invariant or Metropolis-adjusted gradient-based MCMC schemes) to converge to stationarity scales as $\mathcal{O}(J^{\alpha})$ as $J\to\infty$, for some $\alpha>0$ that depends on the setup and type of algorithm \citep{R98,B13,WSC22}.
Combining these two facts, we can expect the cost required by gradient-based MCMC schemes for each posterior sample to scale as $\mathcal{O}(J^{1+\alpha})$ for $\alpha>0$. 
These brief theoretical considerations are in agreement with the numerical results observed in Sections \ref{sec:motivating} and \ref{sec:numerics}.

\subsection{Models with discrete data}\label{sec:discrete_data}
The main requirement in Corollary \ref{mixing_times_warm_hier} is Assumption $(C)$, which imposes a bound on the conditional conductance, uniformly over $Y_j\in\sY$. A setting where the latter is easily satisfied is given by models with discrete data, where $\sY$ is finite, provided exact sampling for $\psi$ is available (as will be the case in all the numerical examples of this paper).

More formally, let $f(y \mid \theta)$ in \eqref{one_level_nested} be a probability mass function with support $\{y_0, \dots, y_m\}$ with $m < \infty$, i.e.\ for every $\theta \in \mathbb{R}^K$ we require the support of $f(y \mid \theta)$ not to depend on $\theta$, i.e.\
\begin{equation}\label{eq:ass_discrete}
\sum_{r = 0}^mf(y_r \mid \theta) = 1, \quad f(y_r \mid \theta) > 0, \quad  r = 0, \dots, m.
\end{equation}
The assumption in \eqref{eq:ass_discrete} is mild and holds for most likelihoods usually employed with binary or categorical data, e.g. multinomial logit and probit. For simplicity we consider the case of a single local parameter with Gaussian prior, i.e.\,
\begin{equation}\label{finite_model}
Y_{j} \mid\theta_j \sim f(y \mid \theta_j)\,,\quad
\theta_1, \dots, \theta_J \mid\mu, \tau \simiid N(\mu, \tau^{-1})\,,\quad
 (\mu,  \tau) \sim p_0(\cdot)\,.
\end{equation}
For example the case $f(y \mid \theta) = \binom{m}{y}\frac{e^{y\theta}}{(1+e^\theta)^m}$, with $y = 0, \dots, m$, corresponds to the logistic hierarchical model with Gaussian random effects. The next proposition shows that, under mild conditions, MwG schemes targeting model \eqref{finite_model} lead to dimension-free mixing times.
\begin{proposition}\label{prop:binary_MH}
Consider model \eqref{finite_model} with likelihood as in \eqref{eq:ass_discrete} and let assumptions $(B1)$-$(B3)$ be satisfied. Let the operator $P_J$ be as in \eqref{two_blocks_MH_nested}, where $\kappa\left(P^{\psi, y_r}\right)$ is continuous with respect to $\psi$ and $\kappa\left(P^{\psi, y_r}\right) > 0$ for every $r = 0, \dots, m$ and $\psi \in S^*$, with $S^*$ neighborhood of $\psi^*$. Then for every $M \geq 1$ and $\epsilon > 0$ there exists $T\left(\psi^*, \epsilon, M \right) < \infty$ such that
\begin{align*}
Q_{\psi^*}^{(J)}\left(t_{mix}(P_J, \epsilon, M)\leq T\left(\psi^*, \epsilon, M \right)\right) &\to 1
&\hbox{as }J \to \infty. 
\end{align*}
\end{proposition}
Requiring $\kappa\left(P^{\psi, y_r}\right) > 0$ for fixed $\psi$ and $y_r$ is arguably a mild condition, e.g.\ it is implied by geometric ergodicity \citep{roberts1996geometric, jarner2000geometric}.

\subsubsection{Numerical simulations}\label{sec:numerics}
First, we provide details on the numerical results illustrated in Figure \ref{fig:binom_simulations}. The model specification is as in \eqref{eq:one_level_nested_intro}, which is a special case of \eqref{finite_model} with the logistic likelihood. The prior on the hyperparameters $\psi = (\mu, \tau)$ is set to $\mu \mid \tau \sim N(0, 10^3/\tau)$ and $\tau \sim \text{Gamma}(1,1)$, while the data are generated by the same model with true parameters $\mu^* = \tau^* = 1$ and $m = 10$ observations per group. Four different algorithms are employed to sample from the associated posterior distribution: a two-block GS defined as in \eqref{two_blocks_gibbs_nested}, where the conditional update of $\mathcal{L}(\theta_j \mid \psi, Y_j)$ is performed through adaptive rejection sampling \citep{GW92}; a MwG algorithm defined as in \eqref{two_blocks_MH_nested}, where the conditional update is done using MH with Barker proposal as in Algorithm 2 of \citep{livingstone2022barker}; MALA with optimal tuning, in the sense that samples obtained from a long run of GS were used to optimally tune a diagonal preconditioning matrix \citep{R98}; the No-U-Turn-Sampler (NUTS) \citep{hoffman2014no} implemented in the software Stan \citep{Rstan}.

Since mixing times of high-dimensional Markov chains are computationally intractable and hard to approximate (see e.g.\ the discussion in \citet{biswas2019estimating}), 
we consider Integrated Autocorrelation Times (IATs) as an easier-to-compute empirical measure of convergence time. Given a $\pi$-invariant Markov chain  $(X^{(t)})_{t\geq 1}$, the IAT associated to a square-integrable test function $f$ is defined as
\[
\textsc{IAT}(f) = 1+2\sum_{t = 2}^\infty\text{Corr}\left(f(X^{(1)}), f(X^{(t)}) \right)\,,
\]
and it measures the number of dependent samples that is equivalent to an independent draw from $\pi$ in terms of estimation of $\int f(x)\pi(dx)$. Intuitively, the higher the IAT the slower the convergence. 
In our numerical experiments, we estimate the IAT with the ratio of the number of iterations and the effective sample size, as described in \cite{GF15}, with the effective sample size computed with the R package \emph{mcmcse} \citep{F21}.

Figure \ref{fig:binom_simulations} depicts the median (over random data set replications) of the maximum IAT over all the parameters (both global and group specific) for the four algorithms, with the number of groups ranging from $128$ to $4096$. Since each iteration of the NUTS kernel involves multiple steps (i.e.\ multiple gradient evaluations), in order to provide a fair comparison the associated IAT is multiplied by the average number of gradient evaluations per iteration. In accordance with Corollary \ref{mixing_times_warm_hier}, the IATs of coordinate-wise samplers do not increase as $J\to\infty$. On the contrary, the IATs of gradient-based samplers seem to grow as $J\to\infty$. The optimally tuned MALA shows better performances than a black-box implementation of NUTS, illustrating the importance of carefully tuning step-sizes.

In our second set of experiments, we test coordinate-wise samplers in situations where the dimensionality of $\theta_j$ in model \eqref{finite_model} is increased by including the presence of $\ell>1$ covariates per group, leading to $J\ell$ latent parameters. In particular, we consider the following specification of the hierarchical logistic regression model
\begin{equation}\label{logistic_model_with_covariates}
Y_{ij} \mid \bm{\theta}_j, X_i \overset{\text{ind.}}{\sim} \, \text{Bernoulli}\left(\frac{e^{X_i^\top\bm{\theta}_j}}{1+e^{X_i^\top\bm{\theta}_j}} \right),\quad
\bm{\theta}_{j}\mid \bm{\mu}, \bm{\tau} \overset{\text{iid}}{\sim} \otimes_{k = 1}^\ell N(\mu_k, \tau_k^{-1})\,,\quad (\bm{\mu}, \bm{\tau}) \sim p_0(\cdot),
\end{equation}
where $i = 1, \dots, m$, $\theta_j = (\theta_{1j}, \dots, \theta_{\ell j})$ for every $j = 1, \dots, J$, $\bm{\mu} = (\mu_1, \dots, \mu_\ell)$ and $\bm{\tau} = (\tau_1, \dots, \tau_\ell)$. The usual conjugate prior is employed for $(\bm{\mu}, \bm{\tau})$, i.e.
\begin{equation}\label{prior_logistic_model_with_covariates}
\mu_k \mid \tau_k \overset{\text{ind.}}{\sim}N(0, 10^3/\tau_k)\,, \quad \tau_k \simiid \text{Gamma}(1,1), \quad k = 1, \dots, \ell.
\end{equation}
In all the simulations to follows, the covariates used for the data-generation process (excluding the intercept) are independently sampled from a uniform distribution between $-5$ and $5$.

 \begin{figure}
\centering
\includegraphics[width=.43\textwidth]{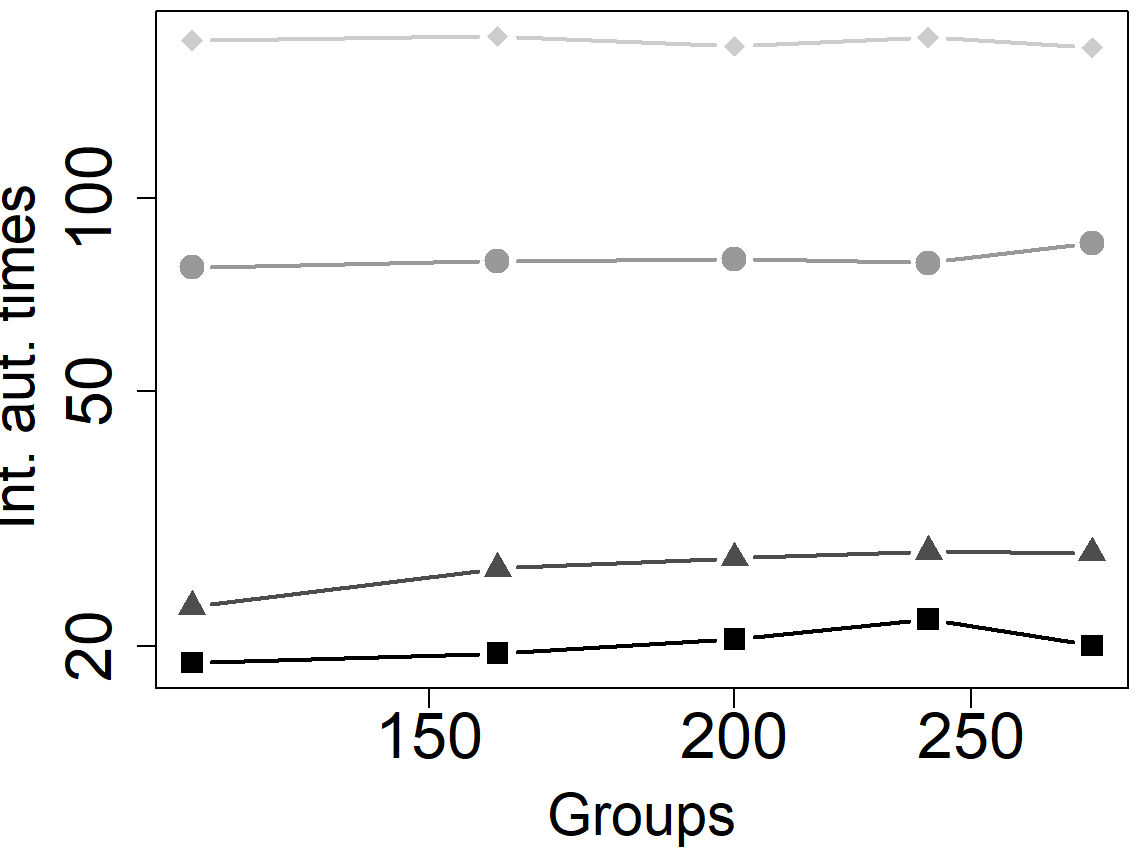} \quad
\includegraphics[width=.47\textwidth]{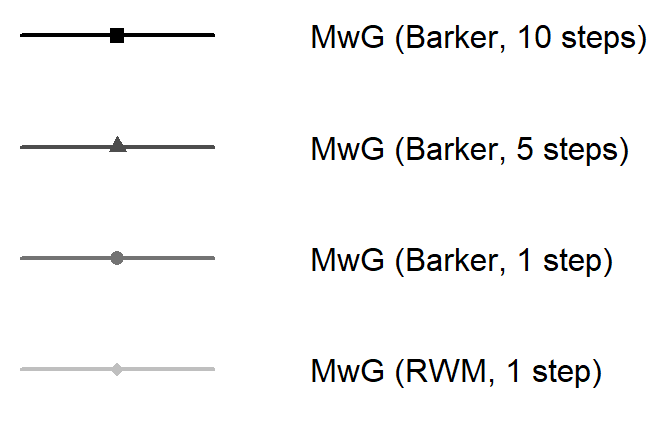}
 \caption{\small{
Median IATs (on the log scale) of four MCMC schemes targeting the posterior distribution of model \eqref{logistic_model_with_covariates} with $\ell = 5$ and $m=30$, as a function of the number of groups.
The median refers to repetitions over datasets randomly generated according to the model with $(\tau^*_1,\dots,\tau^*_5)=(2,1,1,3,2)$ and $\mu_k^*\stackrel{iid}\sim \hbox{Unif}([-1,1])$ for every $k = 1,\dots, 5$.}}
 \label{fig:binom_hier5covariates}
\end{figure}
Figure \ref{fig:binom_hier5covariates}, which is based on model \eqref{logistic_model_with_covariates} with $\ell = 5$ and $m=30$, suggests that the dimension-free convergence of coordinate-wise methods holds even in the presence of covariates. Indeed all the four schemes, whose kernels are as in \eqref{two_blocks_MH_nested} and differ only for the choice of kernel $P^{\psi, Y_j}$ for the conditional updates of the local parameters, exhibit roughly constant IATs as the number of groups grows. Their difference in performance highlights the impact of the conditional conductance, which increases both going from RWM to Barker as well as running multiple steps of the invariant update at each iteration (see Figure \ref{fig:binom_hier5covariates}).

 \begin{figure}
\centering
\includegraphics[width=.45\textwidth]{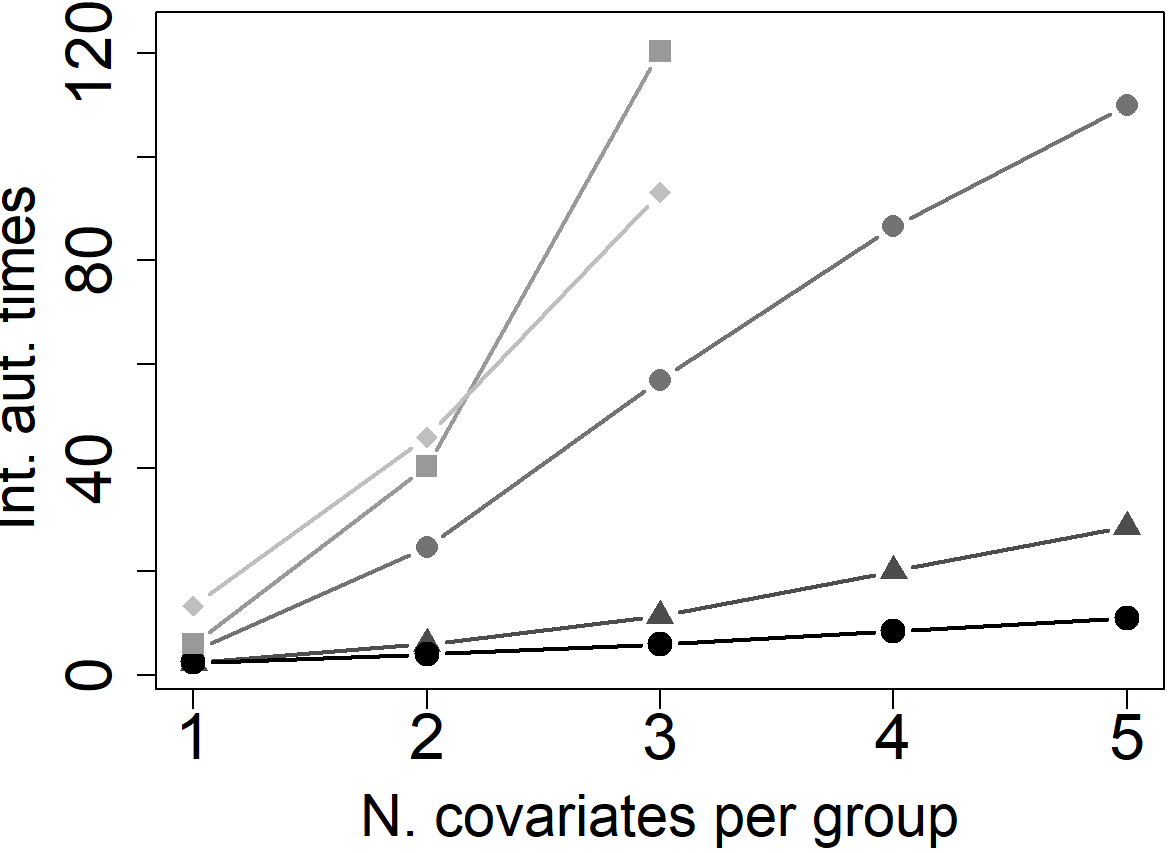} \quad
\includegraphics[width=.45\textwidth]{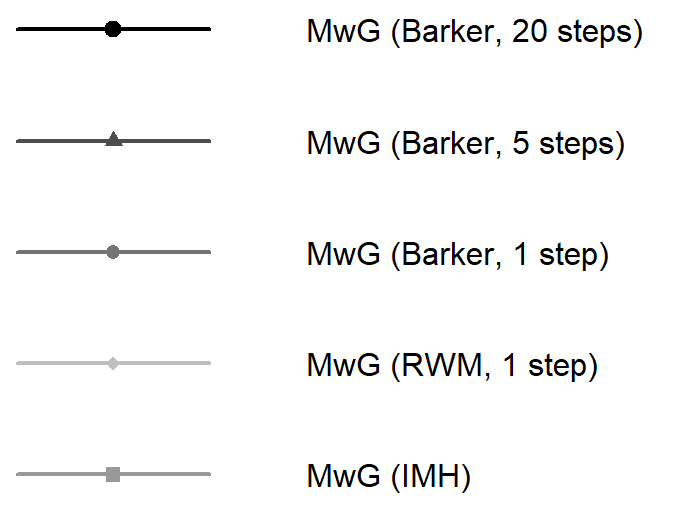}
 \caption{\small{
Median IATs of five MCMC schemes targeting the posterior distribution of model \eqref{eq:one_level_nested_intro}, as a function of the number of covariates (intercept included).
The median refers to repetitions over datasets randomly generated according to model \eqref{logistic_model_with_covariates} with $\tau_k^* = 0.5$ and $\mu_k^*\stackrel{iid}\sim \hbox{Unif}([-1,1])$ for every $k$. Some points for IMH and RWM are omitted due to the difficulty of appropriately estimating high values of IATs.}}
 \label{fig:binom_increasingcovariates}
\end{figure}

Figure \ref{fig:binom_increasingcovariates} instead reports the IATs of different coordinate-wise samplers for model \eqref{logistic_model_with_covariates} with $J = m = 30$ and the number of covariates $\ell$ ranging between $1$ and $5$. As above, the samplers' kernels are as in \eqref{two_blocks_MH_nested}, but they differ in the specific invariant update on the local parameters. The observed IATs increase with $\ell$, which is coherent with the theory developed e.g.\ in Propositions \ref{prop:IMH_logconcave} and \ref{prop:RWM_logconcave}, which suggests that the goodness of the conditional update typically deteriorates with the dimensionality (equivalently, referring to \eqref{eq:cond_bound}, the conditional conductance decreases with $\ell$). In particular, in accordance with Proposition \ref{prop:IMH_logconcave}, the IATs associated to IMH grow very quickly with the dimensionality of $\bm{\theta}_j$. This emphasizes that depending on $\ell$ the choice of the conditional update becomes more and more relevant. Notice moreover that IMH requires to compute the mode of $\pi_J(\bm{\theta}_j \mid \bm{\mu}, \bm{\tau}, Y_j)$ for every $j$ at each iteration, which becomes infeasible even with small $\ell$. 
Note that in this context also the IATs of the exact GS are expected to increase with $\ell$ (see Section \ref{sec:log_regr_hyper} for more illustrations of similar phenomena and discussion about connections with \eqref{logistic_model_with_covariates} being a so-called centered parametrization). This suggests that the increase in IATs as $\ell$ increases observed in Figure \ref{fig:binom_increasingcovariates} is due to a combination of the reduction of the conductance of GS and the conditional conductance of the updates.

\subsection{Feasible start}\label{sec:feasible_start}
Corollary \ref{mixing_times_warm_hier} proves that the Markov chain defined in \eqref{two_blocks_MH_nested} yields bounded mixing times, provided it is initialized from a warm distribution as of \eqref{N_class}. In this section we provide a so-called feasible starting distribution which provides similar guarantees. For simplicity, here we assume that the update of the global parameters is exact, i.e. $P_{1,J} = G_{1, J}$ in \eqref{two_blocks_MH_nested}.

Let $S^*$ be the neighborhood of $\psi^*$ satisfying assumption (C) and $(\nu_{j,\psi})_{j=1}^J \subset \sP(\R^{\ell})$ be a collection of $J$ distributions such that
\begin{align*}
\nu_{j,\psi}  &\in \sN\left(\pi_J\left(\d\theta_j\mid \psi \right), M \right),&j=1,\dots,J\,,
\end{align*}
for every $\psi \in S^*$ and a fixed constant $M \geq 1$. 
Under \eqref{eq:ass_discrete}, any distribution $\nu_{j,\psi} \in \sP(\R^{\ell})$ with compact support satisfies the requirement above. Denote now with $\hat{\psi}_J = \arg\max_\psi \prod_{j = 1}^Jg(Y_j \mid \psi)$ the Maximum marginal Likelihood Estimator (MLE), with $g$ as in \eqref{likelihood_data}. Let $\mu_J = \mu_{J,c} \in \mathcal{P}\left(\R^{D+\ell J} \right)$, with $c > 0$ fixed constant, be defined as
\begin{equation}\label{def_feasible}
\mu_J\left(A  \right) = \int_A\text{Unif}_{cJ^{-1/2}}\left(\hat{\psi}_J  \right)(\d \psi)\left(\prod_{j = 1}^J\nu_{j, \psi}(P^{\psi, Y_j})^{t_J}(\d \theta_j)\right)\,, \quad t_J = \frac{\log(J)}{-\log\left(1-\frac{\kappa^2}{2} \right)}
\end{equation}
for every $A \subset \R^{D+\ell J}$, where $\kappa$ is the constant satisfying assumption $(C)$. Moreover $B_c(x)$ is the closed ball of center $x$ and radius $c > 0$, $\text{Unif}\left(B\right)(\d \psi)$ denotes the uniform distribution over $B \subset \R^D$ and $(P^{\psi, Y_j})^t$ denotes the $t$-th power of the kernel defined in \eqref{two_blocks_MH_nested}. Thus, sampling from $\mu_J$ in \eqref{def_feasible} can be performed in two steps: the global parameters $\psi$ follow a random perturbation of the MLE and, conditional on $\psi$, $J$ independent Markov chains with kernels $(P^{\psi, Y_j})_{j=1}^J$ are run for a logarithmic (in $J$) number of iterations. The next proposition shows that mixing times starting from $\mu_J$ enjoy a dimension-free asymptotic behaviour.
\begin{proposition}\label{prop:feasible}
Consider the same setting and assumptions of Corollary \ref{mixing_times_warm_hier} and $\mu_J$ as in \eqref{def_feasible}.
Then, for every $\epsilon > 0$ there exists $T\left(\psi^*, \epsilon \right) < \infty$ such that
\begin{align}
&Q_{\psi^*}^{(J)}\left(t_{mix}(P_J, \epsilon, \mu_J)\leq T\left(\psi^*, \epsilon \right)\right) \to 1
&\hbox{ as }J \to \infty\,.
\end{align}
\end{proposition}
\begin{remark}
Notice that $\mu_J$ defined in \eqref{def_feasible} is not in general a warm start, i.e.\ we cannot guarantee $\mu_J\in \sN\left(\pi_J\left(\d\psi, \d \bm{\theta} \right), M \right)$ for a fixed $M$ that does not grow with $J$. However in the proof of Proposition \ref{prop:feasible} it is shown that $\mu_J$ is ``close enough'' to a warm distribution in total variation distance, which suffices for our purposes.
\end{remark}
Assuming $\hat{\psi}_J$ can be approximately computed with  $\sO(J \log(J))$ computational cost or less, the cost of sampling from $\mu_J$ in \eqref{def_feasible} is $\sO(J \log(J))$. As already mentioned in Section \ref{sec:complexity}, this implies that the overall computational cost of running the algorithm is again $\sO_P(J \log(J))$.

\section{Bayesian binary regression with unknown prior variance}\label{sec:log_regr_hyper}
Consider a Bayesian logistic regression model with unknown prior precision defined as
\begin{equation}\label{logistic_model}
\begin{aligned}
Y_i \mid \bm{\theta},\alpha &\sim \text{Bernoulli}\left(\frac{e^{\bm{\theta}^\top X_i}}{1+e^{\bm{\theta}^\top X_i}} \right),
&i=1,\dots,n
\\
\bm{\theta} \mid \alpha &\sim N\left(\bm{0}, \alpha^{-1}\Sigma\right),
\\
\alpha&\sim\text{Gamma}(a,b),
\end{aligned}
\end{equation}
where $\bm{\theta} = (\theta_1, \dots, \theta_d)$, 
 $X$ is a $n\times d$ matrix with $i$-th row $X_i$,
$\Sigma$ is a $d\times d$ positive definite covariance matrix and 
$\text{Gamma}(a,b)$ denotes the Gamma distribution with parameters $a,b>0$. 
    Let $\pi(\d\alpha,\d\bm{\theta}):=\L\left(\d\alpha ,\d\bm{\theta}\mid Y\right)\in\sP(\R^{1+d})$ be the joint posterior of $\alpha$ and $\bm{\theta}$ given the vector of observations $Y=(Y_1, \dots, Y_n)$, under model \eqref{conditional_logistic}.
The conditional density of $\bm{\theta}$ given $\alpha$ under $\pi$, i.e.\ 
\begin{equation}\label{conditional_logistic}
\pi\left(\bm{\theta} \mid \alpha\right) \propto \text{exp}\left\{Y^\top X\bm{\theta}-\sum_{i = 1}^n\log \left(1+e^{\bm{\theta}^\top X_i}\right)-
\frac{\alpha}{2}\bm{\theta}^\top\Sigma^{-1}\bm{\theta}
\right\}\,,
\end{equation}
is strongly log-concave.
We denote the condition number of $\pi\left(\bm{\theta} \mid \alpha\right)$ by $c(\alpha)$. Explicit bounds on $c(\alpha)$
are available \citep{D19}, which however diverge to $\infty$ as $\alpha$ goes to $0$.
As a result, convergence properties of MCMC algorithms targeting model \eqref{conditional_logistic} with fixed $\alpha$ are well understood \citep{D17,DM17,D19}. However, $\alpha$ is usually unknown in applications and it is typically incorporated as a random variable in the Bayesian model, as in  \eqref{conditional_logistic}. In such cases, the joint posterior $\pi$ is not log-concave on $\R^{d+1}$ and analyzing the convergence of MCMC algorithms performing joint updates of $(\alpha,\bm{\theta})$ can be much harder.

Given model \eqref{logistic_model}, it is natural to consider a coordinate-wise posterior sampling scheme with $\pi$-invariant kernel
\begin{equation}\label{two_blocks_MH_nested_logistic}
P = \frac{1}{2}G_{1}+\frac{1}{2}P_{2},
\end{equation}
where
\begin{align*}
&G_{1}(\left(\alpha, \bm{\theta}), \left(\d \alpha', \d \bm{\theta}' \right)\right)
= \delta_{\bm{\theta}}\left(\d\bm{\theta}'\right)\pi\left(\d \alpha' \mid \bm{\theta} \right),
\\
&P_{2}\left((\alpha, \bm{\theta}), \left(\d \alpha', \d \bm{\theta}' \right)\right) = \delta_{\alpha}\left(\d\alpha'\right)
P^{\alpha, Y}\left(\bm{\theta}, \d \bm{\theta}' \right),
\end{align*}
with $P^{\alpha, Y}$ being a $\pi(\d \bm{\theta} \mid \alpha)$-invariant kernel. Indeed, by strong log-concavity of $\pi\left(\bm{\theta} \mid \alpha\right)$, Propositions  \ref{prop:IMH_logconcave} and \ref{prop:RWM_logconcave} can be applied when choosing $P^{\alpha, Y}$ appropriately. Moreover, sampling from $\pi\left(\d \alpha \mid \bm{\theta} \right)$ is straightforward 
due to the Normal-Gamma conjugacy.
The next proposition states the resulting bound on the conductance, when $P^{\alpha, Y}$ is a RWM update.
\begin{proposition}\label{MwG_logistic_regression}
Let $S \subset \R^+$ and $P$ be as in \eqref{two_blocks_MH_nested_logistic} with $P^{\alpha, Y}$ as in \eqref{eq:rwm_kernel}. Then for every $s \in (0, 1/2)$ we have
\[
\Phi_s(P) \geq M\sqrt{\frac{\inf_{\alpha \in S}\,c(\alpha)^{-1}}{d}}\Phi_s(G)-\frac{\pi(S^c \times \R^d)}{s},
\]
where $G=\frac{1}{2}G_1+\frac{1}{2}G_2$ is the kernel of a two-component GS targeting $\pi$ and $M=M(\eta)>0$ is the same constant of Proposition \ref{prop:RWM_logconcave}, which does not depend on $\pi$.
\end{proposition}
Intuitively, Proposition \ref{MwG_logistic_regression} implies that the following two conditions are sufficient for $P$ to mix fast: 
the posterior distribution of $\alpha$ concentrates in a region where the condition number $c(\alpha)$ is not too high and the exact GS has a good approximate conductance (i.e.\ $\Phi_s(G)$ is not too close to $0$). 
While providing a lower bound to $\Phi_s(G)$ may seem equally challenging as doing it for $\Phi_s(P)$, an important simplification (in terms of dimensionality reduction) is available for the exact GS. Let $T = \bm{\theta}^\top\Sigma^{-1}\bm{\theta}$ and let $\tilde{\pi}(\d \alpha,\d T):=\L\left(\d\alpha , \d T\mid Y\right)\in\sP(\R^{2})$ be the two-dimensional marginal posterior distribution of $(\alpha,T)$, with associated GS kernel $\tilde{G}$.
 The next lemma provides a lower bound on  $\Phi_s(G)$ in terms of $\Phi_s(\tilde{G})$.
\begin{lemma}\label{dimensionality_reduction_logistic}
Let $G$ and $\tilde{G}$ be kernels of the Gibbs samplers on $\pi$ and $\tilde{\pi}$ as above. Then for every $s \in (0, 1/2)$ we have
\[
\Phi_s(G) \geq \frac{1 }{4}\frac{-\log \left(1-\frac{\Phi^2_{s/8}(\tilde{G})}{2} \right)}{\log(8)-\log(s)}.
\]
\end{lemma}
\begin{remark}\label{remark:dimensionality_reduction}
The main intuition underlying Lemma \ref{dimensionality_reduction_logistic} is that the update of $\alpha$ given $\bm{\theta}$ depends only on $T$, i.e.\ 
$\pi\left(\d \alpha \mid \bm{\theta} \right) = \pi\left(\d \alpha \mid T\right)$ for $T=\bm{\theta}^\top\Sigma^{-1}\bm{\theta}$, which implies that $G$ and $\tilde{G}$ have closely related convergence properties. See Lemma \ref{sufficient_lemma} and the proof of Lemma \ref{dimensionality_reduction_logistic} in the supplement for more details and \cite{R01,D08,Q19,AZ23} for results similar in spirit. 
\end{remark}
\begin{figure}
\centering
\includegraphics[width=.40\textwidth]{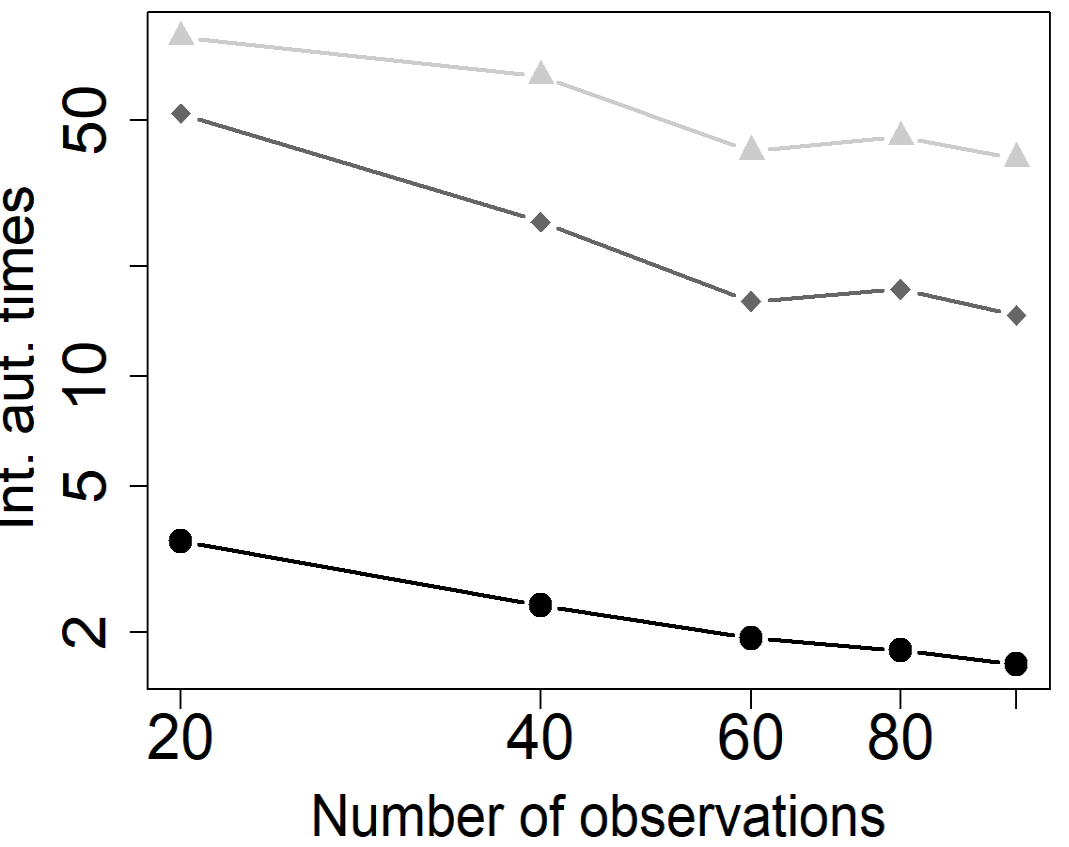} \quad
\includegraphics[width=.56\textwidth]{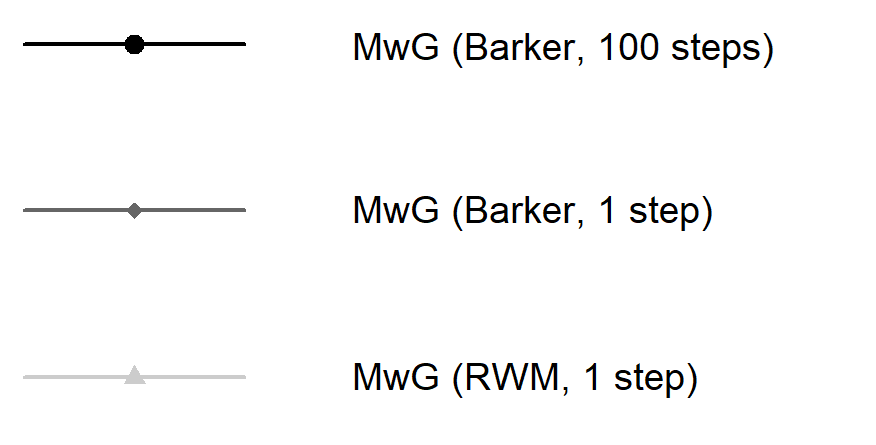}
 \caption{\small{
Median IATs (on the log scale) of three	 MCMC schemes targeting the posterior distribution of model \eqref{logistic_model} with $d = 5$, as a function of the number of observations $n$.
The median refers to repetitions over datasets randomly generated according to model \eqref{logistic_model} with $\alpha^* = 1$ and $\Sigma = \frac{1}{5}\mathbb{I}$.}}
 \label{fig:logisticRegression1}
\end{figure}
We compare numerically three coordinate-wise samplers with kernels as in \eqref{two_blocks_MH_nested_logistic}, with $P^{\alpha, Y}$ given by, respectively RWM, Barker and $100$ repeated steps of Barker (which we take as a proxy of the exact GS in this context due to the small value of $d$ and the high number of steps).
We consider model \eqref{logistic_model} with $\Sigma = d^{-1}\mathbb{I}_d$, $d = 5$, $a=b=1$ and $n$ ranging from $20$ to $100$. 
The data are generated from the same model with $\alpha$ set to $1$.
The results, reported in Figure \ref{fig:logisticRegression1}, illustrate that, for all the samplers under consideration, performances improve as $n$ grows and in particular the IATs decrease as $n$ increases. This is in accordance with the fact that the parametrization of model \eqref{logistic_model} is a so-called `centered' one (in the sense of \cite{gelfand1995efficient,PRS03, papaspiliopoulos2007general}), so that the performances of the associated coordinate-wise samplers improve as $n/d$ increases (i.e.\ as the data became more informative) while they suffer if $d/n$ is large. 
Notice that the reduction in IAT with respect to $n$ is mostly due to the behaviour of the exact GS rather than a change in the conditional conductance: indeed, also the black line in Figure \ref{fig:logisticRegression1} exhibits a decreasing behaviour.

\begin{figure}
\centering
\includegraphics[width=.42\textwidth]{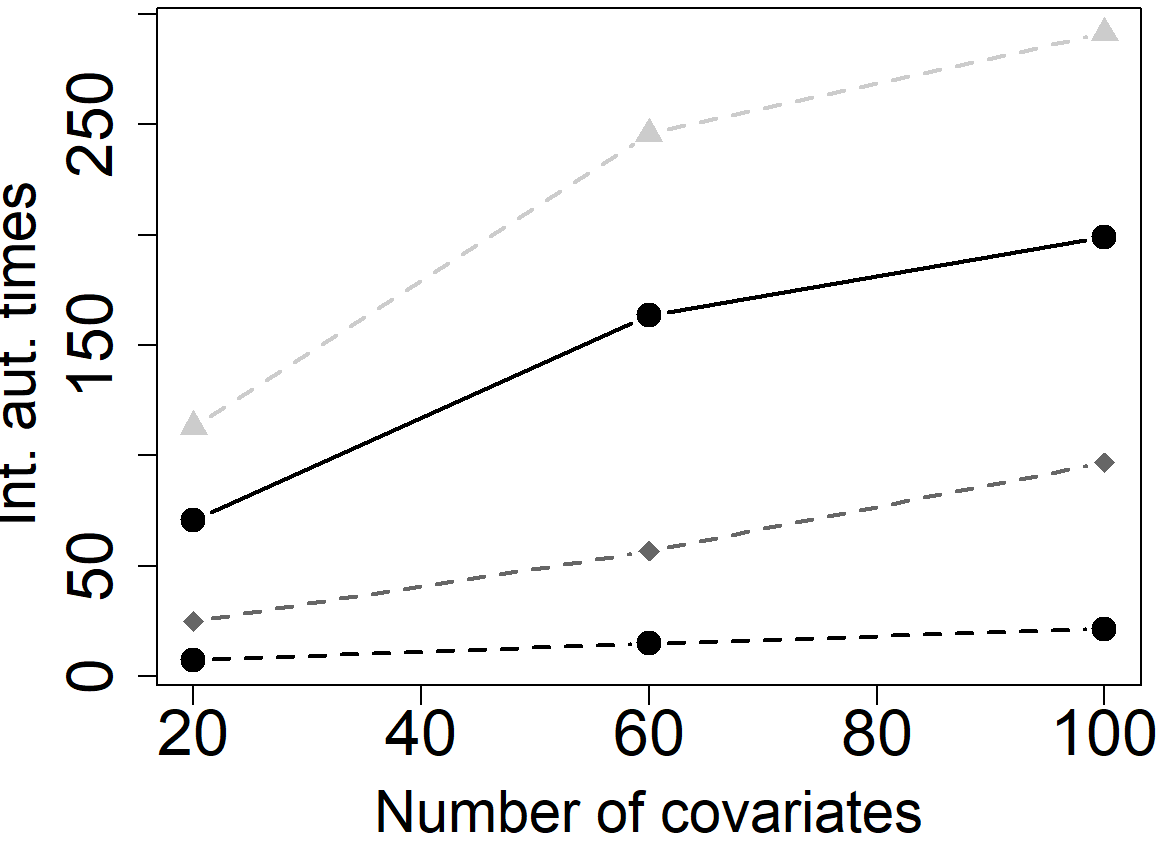} \quad
\includegraphics[width=.53\textwidth]{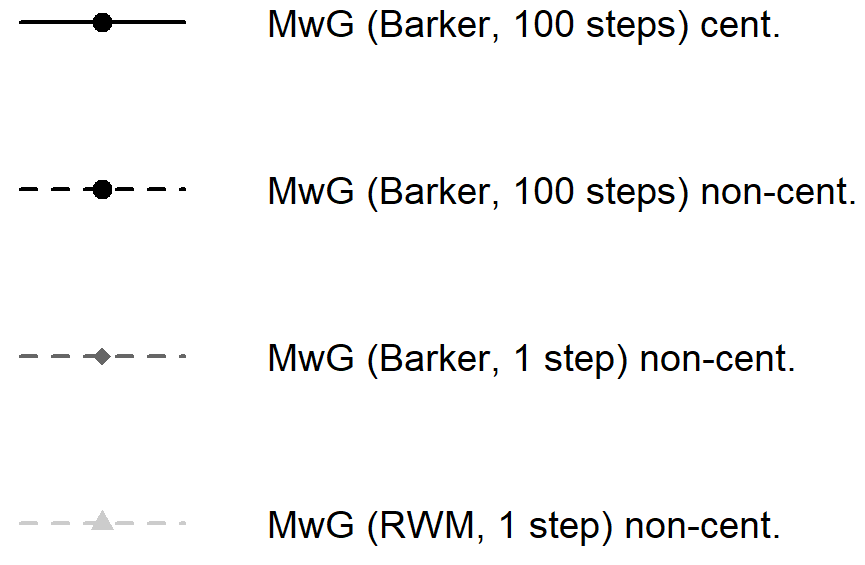}
 \caption{\small{
Median IATs of four MCMC schemes targeting the posterior distribution of model \eqref{logistic_model}, as a function of the number of covariates $d$ and $n = d/2$. 
The median refers to repetitions over datasets randomly generated according to model \eqref{logistic_model} with $\alpha^* = 1$ and $\Sigma = \frac{1}{d}\mathbb{I}$. Full lines: centered parametrization. Dotted lines: non-centered parametrization.}}
 \label{fig:logisticRegression2}
\end{figure}
Next, we consider also the non-centered version of model \eqref{logistic_model}, which can be formulated as
\begin{equation}\label{logistic_model_nonC}
\begin{aligned}
Y_i \mid \bm{\beta},\alpha \sim \text{Bernoulli}\left(\frac{e^{\bm{\beta}^\top X_i/\sqrt{\alpha}}}{1+e^{\bm{\beta}^\top X_i/\sqrt{\alpha}}} \right), \quad \bm{\beta} \mid \alpha \sim N\left(\bm{0}, \Sigma\right),
\quad 
\alpha&\sim\text{Gamma}(a,b).
\end{aligned}
\end{equation}
Similarly to \eqref{two_blocks_MH_nested_logistic}, coordinate-wise samplers can be used to sample from the posterior distribution $\pi(\d\alpha,\d\bm{\beta})=\mathcal{L}(\d\alpha,\d\bm{\beta}\mid Y)$, with the caveat that $\pi(\d\alpha\mid \bm{\beta})$ is not available in closed form and thus we perform a MH update for it instead of an exact Gibbs one. Given that $\alpha$ is a one-dimensional parameter, we found a RWM update of it to be sufficient for our purposes. Figure \ref{fig:logisticRegression2} reports the IAT estimates for the coordinate-wise samplers obtained with the centered (full lines) and non-centered (dotted lines) parametrizations when $d \in \{20,60,100\}$ and $n = d/2$. In accordance with the discussion above, the non-centered parametrization behaves significantly better than the centered one in this regime with $d/n$ relatively large. Note that even the two MwG schemes with $100$ Barker steps per iteration exhibit very different behaviours and in particular the non-centered one is significantly more robust to large values of $d$.

\section{Data augmentation schemes for discretely-observed diffusions}\label{sec:diffusions}

Many real-life phenomena of interest, for example in  biology, physics and finance, can be described as a diffusion over $[0, T]$ defined through a Stochastic Differential Equation (SDE) such as
\begin{equation}\label{eq:sde_general}
\d X_t = b(\theta, X_t)\d t + \sigma(X_t) \d B_t, \quad t \in [0, T],
\end{equation}
where $B_t$ is a Brownian motion on $\R$ and $\theta \in \R^p$ is a vector of parameters on which we want to make inference. See e.g. \cite{roberts2001inference, beskos2006retrospective, papaspiliopoulos2013data} and references therein for a review.

The functions $b$ and $\sigma$ are assumed to satisfy the basic regularity conditions (e.g. locally Lipschitz, $\sigma$ bounded below) which imply the existence of a weakly unique solution. 
For ease of exposition in this section, we shall restrict ourselves to the case where the function $\sigma$ is known, although our results could be readily extended
to the case of unknown $\sigma$ through standard techniques (see e.g.\ \cite{roberts2001inference}). Multivariate extensions are also possible.

Proceeding in the standard way, we apply the Lamperti transformation $x \to \eta(x)$ to $X_t$ as in \eqref{eq:sde_general}, with $\eta(x) = \int_z^x\sigma(u)^{-1}\, \d u$, from which we obtain a diffusion with unit coefficient. Therefore without loss of generality we can consider
\begin{equation}\label{eq:sde}
\d X_t = b(\theta, X_t)\d t + \d B_t, \quad t \in [0, T].
\end{equation}
Assuming hypothetically we observed the entire trajectory $X_{[0, T]} = \{X_t\;:\;t \in [0, T]\}$, the likelihood of $\theta$ is given by the well-known Girsanov formula, which reads
\begin{equation}\label{Girsanov_formula}
G\left(X_{[0, T]}, \theta\right) = \text{exp}\left\{\int_0^Tb(X_t, \theta) \, \d X_t-\frac{1}{2}\int_0^Tb^2(X_t,\theta) \, \d t\right\},
\end{equation}
from which the posterior distribution of $\theta$, given a prior distribution $\pi(\d \theta)$, has density given by
\begin{equation}\label{eq:post_theta}
\pi\left(\theta \mid X_{[0, T]}, \theta\right) \propto  G\left(X_{[0, T]}, \theta\right)\pi(\theta),
\end{equation}
assuming $\pi(\d \theta)$ admits a density $\pi(\theta)$ with respect to the Lebesgue measure.

In practice we observe \eqref{eq:sde} at fixed points $\T = (t_0, t_1, \dots, t_N)$ where we assume for simplicity $t_i-t_{i-1} = \Delta$ for every $i$. If $\Delta$ is small enough typically a discretization of the integrals in \eqref{Girsanov_formula} can often be employed
(as in the Euler-Maruyama method) circumventing the need for data-augmentation. More realistically, when data is sparser, the likelihood is given by
\[
f\left(X_{t_0}, \dots, X_{t_N} \right) = f_\theta(X_{t_0})\prod_{i = 1}^Nf_\theta(X_{t_{i-1}}, X_{t_i}),
\]
where $f_\theta(x, z)$ is the transition probability density associated to \eqref{eq:sde} of passing from $x$ to $z$ in an interval $\Delta$ of time. However $f_\theta$ is in general intractable and not available in closed form. 
To circumvent this problem, a data augmentation scheme has been proposed in the literature (\cite{roberts2001inference, beskos2006retrospective, papaspiliopoulos2013data}), where at each step the missing data $Y = \{X_t\}_{t \not\in \T}$ are imputed. This can be described as a coordinate-wise scheme on $\sX = \R^p \times \prod_{i = 1}^N\R^{(t_{i-1},t_i)}$ with invariant distribution $\pi(\d \theta, \d Y \mid X_{\T})$. The associated operator $P$ can be formally defined as
\begin{equation}\label{eq:operator_diffusion}
P = \frac{1}{2}G_1+\frac{1}{2}P_2,
\end{equation}
where
\[
 G_{1}\left((\theta, Y), \left(\d \theta', \d Y' \right)\right) = \delta_{Y}(\d Y')\pi\left(\d \theta \mid X_{[0, T]}\right)
 \]
is the exact step for $\theta$, which we assume to be feasible thanks to the explicit likelihood in \eqref{Girsanov_formula}, and
\begin{equation}\label{eq:update_diffusion}
P_{2}\left((\theta, Y), \left(\d \theta', \d Y' \right)\right) = \delta_{\theta}(\d \theta')\prod_{i = 1}^NP_i^{\theta}\left(Y_i, \d Y_i' \right),
\end{equation}
with $Y_i$ denoting the evolution of the diffusion over the interval $(t_{i-1}, t_i)$ and $P_i^{\theta}$ a Markov operator on $\R^{(t_{i}-t_{i-1})}$ with invariant distribution $\pi\left(\d Y_i \mid \theta, X_{t_{i-1}}, X_{t_i}\right)$. In words, step \eqref{eq:update_diffusion} requires updating separately the evolution of the diffusion in the $N$ sub-intervals defined by $\T$: this can be done independently thanks to the Markov property of \eqref{eq:sde}. A common choice for $P_i^\theta$, see e.g. \cite{roberts2001inference}, is an Independent Metropolis-Hastings as in Section \ref{sec:ind}, where the new path $Y_i'$ is sampled according to a Brownian bridge on $[t_{i}-t_{i-1}]$ constrained to be equal to $X_{t_{i-1}}$ and $X_{t_i}$ at the endpoints and it is accepted with probability
\begin{equation}\label{eq:acceptance_diffusions}
\alpha(Y, Y') = \min \left\{1, \frac{G(Y', \theta)}{G(Y, \theta)} \right\}.
\end{equation}
For the technical details we refer to \cite{roberts2001inference}.

We make the following assumptions:
\begin{enumerate}
\item[(C1)] The function $b(\theta, x)$ is differentiable in $x$ for every $\theta$ and continuous in $\theta$ for every $x$.
\item[(C2)] The function $b^2(\theta, x)+b'(\theta,x)$ is bounded below for every $\theta$.
\item[(C3)] The prior $\pi(\d \theta)$ is supported on a compact space $S \subset \R$.
\item[(C4)] There exists $c>0$,  not depending on $N$, such that $X_{t_i} \in [-c,c]$ for every $i = 1, \dots, N$.
\end{enumerate}
Assumptions $(C1)-(C2)$ are common in this literature (see e.g. \cite{beskos2006retrospective}) and are satisfied in many interesting cases \citep{polson1994bayes}.
On the other hand, assumptions $(C3)$ and especially $(C4)$ are arguably strong assumptions, which allow for technical simplifications but may not be satisfied in many practical scenarios.
An alternative would be to consider a compact set $K$ which retains most of the posterior mass, as in \eqref{eq:cond_bound}.  The next proposition shows that we can then provide a lower bound on the conductance of $P$ in \eqref{eq:operator_diffusion} in terms of the corresponding Gibbs sampler.
\begin{proposition}\label{prop:conductance_diffusion}
Let $P$ be as in \eqref{eq:operator_diffusion}, with $P_i^{\theta}$ operators of the Independent Metropolis-Hastings as in \eqref{eq:acceptance_diffusions}. Then under assumptions $(C1)-(C4)$ there exists a constant $\kappa = \kappa(c)>0$ such that
\begin{equation}\label{eq:cond_cond_diffusions}
\Phi(P) \geq \kappa \Phi(G),    
\end{equation}
for every $N$, $\T = (t_0, \dots, t_N)$, with $0 \leq t_0 < \dots < t_N \leq T$, and $X_\T = \left( X_{t_0}, \dots, X_{t_N}\right)$.
\end{proposition}
Note that, although not explicitly indicated in our notation, the kernels $P$ and $G$ depend on $N$ (as well as on $\T$ and $X_\T$).
Indeed, the key and non-trivial part in the statement of Theorem \ref{prop:conductance_diffusion} is that the constant $\kappa$ in \eqref{eq:cond_cond_diffusions} does not depend on $N$, $\T$ and $X_\T$.
This, for example, implies that, under assumptions $(C1)-(C4)$, the decrease of conductance passing from $G$ to $P$ remains bounded as $N\to\infty$.
Extensions of Proposition \ref{prop:conductance_diffusion} to the case of approximate conductance, in order to relax assumption $(C3)$, are direct following the approach of the previous sections.

Proposition \ref{prop:conductance_diffusion} illustrates the applicability of the techniques developed in this paper in the context of parameters inference for diffusions: in this case the missing data require the imputation of an infinite-dimensional object, i.e. a diffusion trajectory. Under the above assumptions, Proposition \ref{prop:conductance_diffusion} reduces the problem of determining the performances of $P$ as in \eqref{eq:operator_diffusion} to the task of studying the associated Gibbs sampler. Similarly to the applications discussed in the Sections \ref{sec:hierarchical} and \ref{sec:log_regr_hyper}, the latter can entail major simplifications through dimensionality reduction: for example, in the popular setting where $b(x, \theta)$ is a polynomial of degree $m$ in $\theta$, 
techniques analogous to Lemma \ref{dimensionality_reduction_logistic} and Remark \ref{remark:dimensionality_reduction} allow to reduce the problem to analysing a Gibbs sampler $\tilde{G}$ targeting a $\ell$-dimensional target with $\ell=\mathcal{O}(2m)$.
We leave more detailed examples with specific classes of diffusions to future work.

\section{Discussion}
The results of Section \ref{sec:main} show that performances of a coordinate-wise scheme can be controlled monitoring two factors: (a) whether conditional updates are close enough to exact sampling (i.e. the conditional conductance is far from zero) and (b) whether the GS itself is a good scheme for the specific sampling problem.
For example, the first inequality in \eqref{eq:cond_bound} bounds the slow-down of a general coordinate-wise sampler relative to GS (measured in terms of reduction of conductance) with a multiplicative term of the form $\min_{i \in \{1, \dots, d\}} \, \kappa_i(P_i, \sX)$, defined in \eqref{minimum_conditional_conductance}.  
An interesting aspect of such a `slow-down' factor is that it 
depends on $d$ only through the minimum operation. 
Thus, if the GS updates are replaced by moderately good MH conditional updates (e.g.\ ones with conductance uniformly bounded away from $0$) the resulting MwG sampler will incur in a slow-down that is constant with respect to $d$. 
This observation agrees with the observation that, provided the full-conditional distributions are well-behaved, MwG tends to perform similarly to the corresponding GS (see e.g.\ Figure \ref{fig:binom_simulations}), even when $d$ is large.

Note that the results of Section \ref{sec:main}, in their current form, strongly relies on the random scan architecture, where at each iteration a randomly chosen  block is updated. Extensions to other popular orderings, e.g.\ deterministic-scan ones, are less trivial than one might expect. 
In particular, naive applications of the techniques employed in this paper would lead to conductance bounds depending on multiplicative factors of the form $\prod_{i=1}^d \kappa_i(P_i, K)$, which scale exponentially badly with $d$.
Developing tight bounds for deterministic-scan coordinate-wise samplers 
is an interesting direction for future work.



Our results 
provide simple and intuitive guidance for practitioners using MwG-type schemes, which is to think separately of the two potential sources of slow mixing: (a) bad conditional updates and (b) strong dependence among parameters (which might slow down GS). Also, they justify applying the various techniques developed in the literature to improve GS convergence \citep{gelfand1995efficient,PRS03, papaspiliopoulos2007general,khare2011spectral,yu2011center,kastner2014ancillarity} to the broader class of MwG and general coordinate-wise samplers.

Beyond being interesting per se (in terms of improving our understanding of commonly used coordinate-wise samplers), the general results of Section \ref{sec:main} are motivated by and applied to the analysis of MwG schemes targeting high-dimensional non-conjugate Bayesian hierarchical models, where we manage to derive dimension-free bounds on total variation mixing times as $J\to\infty$. As illustrated in Section \ref{sec:motivating}, these are computationally challenging (as well as widely used) models, where competing sampling algorithms (including popular gradient-based MCMC schemes, such as NUTS) exhibit, either empirically, theoretically or both, a total computational cost scaling super-linearly with the number of groups $J$. On the contrary, our results imply a total computational cost of coordinate-wise samplers that scales linearly with $J$, thus providing state-of-the-art performances for high-dimensional non-conjugate Bayesian hierarchical models. 


Our use of the approximate conductance and of perturbation results (see e.g.\ Section \ref{sec:conductance_distance}) is motivated by statistical applications, where we seek to combine MCMC convergence analysis with Bayesian posterior asymptotic results (such as the celebrated Bernstein-von Mises theorem). 
Note that, despite being a powerful and potentially useful approach in various contexts, rigorous combinations of MCMC theory and Bayesian asymptotics are relatively scarce in the literature and, beyond few exceptions \cite{roberts2001approximate,belloni2009computational,kamatani2014local}, mostly recent \cite{nickl2022polynomial,negrea2022statistical,tang2022computational,AZ23}. 
We hope that the techniques developed in this paper might serve as starting point to develop a better quantitative understanding of popular coordinate-wise samplers for other classes of high-dimensional structured Bayesian models (e.g.\ times series, factor models, Gaussian processes, etc), thus reducing the significant gap between theory and practice still present in this area.





\medskip
\medskip

\textbf{Funding. }
GZ acknowledges support from the European Research Council (ERC), through StG ``PrSc-HDBayLe'' grant ID 101076564.
GOR was supported by EPSRC grants Bayes for Health (R018561) CoSInES (R034710) PINCODE (EP/X028119/1), EP/V009478/1 and by the UKRI grant, OCEAN, EP/Y014650/1.

\bibliographystyle{chicago}
\bibliography{Bib_MwG}

\begin{appendices}

\section{Additional examples}\label{examples}

\subsection{Tightness of the bound in Corollary \ref{cor:bound_conductance}}\label{tightness_example}
Define
\begin{equation}\label{defin:P}
P = \frac{1}{d}\sum_{i = 1}^dP_i, \quad P_i = cG_i+(1-c)I,
\end{equation}
for $i = 1, \dots, d$ and $c\in(0,1)$. 
It follows that $P_i(\partial A) = cG_i(\partial A)$ for every $A\subseteq \sX$, so that 
$\Phi(P) = c\Phi(G)$. Also, by \eqref{defin:P}, we have
\[
P_i^{\x_{-i}}(x_i, \d y_i) = cG_i^{\x_{-i}}(x_i, \d y_i)+(1-c)\delta_{x_i}(y_i),
\]
so that 
\[
\int_{B} P^{\x_{-i}}_i(x_i, B^c) \pi_i (dx_i|\x_{-i}) = c\int_{B} G^{\x_{-i}}_i(x_i, B^c) \pi_i (dx_i|\x_{-i}) = c\pi_i(B \mid \x_{-i})\pi_i(B^c \mid \x_{-i})
\]
for every $B\subseteq \sX_i$
By \eqref{conditional_conductance} we conclude $\kappa(P_i^{\x_{-i}}) = c$ for every $\x\in\sX$, from which it follows $\Phi(P) = \kappa(P)\Phi(G)$ as desired.

\subsection{Convergence of the stationary distributions does not imply convergence of the conductances}\label{example_asymp_conduct}

Let $\sX = \R^2$ and $\pi$ 
be a bivariate standard normal distribution.
Define $\pi_n$ to be the truncation of $\pi$ on the set $C_n$, where
\[
C_n = \left\{(-\infty, n] \times (-\infty, n]\right\} \, \bigcup \,\left\{ [n, + \infty) \times [n, +\infty)\right\}.
\]
Let $G_n$ and $G$ be the operators of the associated Gibbs samplers. Then, it is not difficult to show that $\lTV \pi_n - \pi \rTV \to 0$ as $n \to \infty$ and $\Phi(G_n) = 0$ for every $n$, since it suffices to take $A = [n, + \infty) \times [n, +\infty)$ in \eqref{s_conductance}. On the other hand $\Phi(G) = 1$, since $G$ is a Gibbs Sampler (GS) on a distribution with independent components.

\section{Reversible and positive semi-definite coordinate-wise Markov chains}
Let $P$ as in \eqref{eq:coordinate_wise}. The next lemma shows that reversibility and positive semi-definiteness of $P$ follow by the ones of $P_i$ and in turn of $P_i^{\x_{-i}}$.
\begin{lemma}\label{lemma:rev_pos_supp}
Let $P=\frac{1}{d} \sum_{i=1}^dP_i$ be a Markov kernel. Then, if $P_i$ is $\pi$-reversible and positive semi-definite for every $i$, then the same holds for $P$.

Moreover, let $P_i$ be as in \eqref{eq:coordinate_wise}. Then $P_i$ is $\pi$-reversible and positive semi-definite if and only if $P_i^{\x_{-i}}$ is $\pi(\cdot \mid \x_{-i})$-reversible and positive semi-definite for $\pi(\d \x_{-i})$-almost every $\x_{-i}$.
\end{lemma}
\begin{proof}
If $A, B \subset \sX$, then
\[
\begin{aligned}
\int_A\int_B P(\x, \d \y)\pi(\d \x) &= \frac{1}{d} \sum_{i=1}^d\int_A\int_B P_i(\x, \d \y)\pi(\d \x)\\
& = \frac{1}{d} \sum_{i=1}^d\int_B\int_A P_i(\x, \d \y)\pi(\d \x)  = \int_B\int_A P(\x, \d \y)\pi(\d \x),
\end{aligned}
\]
where the second equality follows by $\pi$-reversibility of $P_i$. Moreover, let $f \, : \, \sX \, \to \, \sX$ be such that $\int f^2(\x) \pi(\d \x) < \infty$ and notice that
\[
\begin{aligned}
\int \left(Pf(\x)\right)f(\x) \pi(\d \x) &= \frac{1}{d} \sum_{i=1}^d \int \left(P_if(\x)\right)f(\x) \pi(\d \x) \geq 0,
\end{aligned}
\]
by positive semi-definiteness of $P_i$. Thus, $P$ is $\pi$-reversible and positive semi-definite.

As regards the second point, assume that $P_i^{\x_{-i}}$ is $\pi(\cdot \mid \x_{-i})$-reversible and positive semi-definite for $\pi(\d \x_{-i})$-almost every $\x_{-i}$. Then, letting $A(\x_{-i}) = \left\{x_i \, ; \, (x_i, \x_{-i}) \in A \right\}$ and similarly for $B(\x_{-i})$, we have that
\[
\begin{aligned}
\int_A\int_B P_i(\x, \d \y)\pi(\d \x) &= \int \left[\int_{A(\x_{-i})}\int_{B(\x_{-i})} P_i^{\x_{-i}}(x_i,\d y_i) \pi_i(\d x_i \mid \x_{-i}) \right]\pi(\d \x_{-i}) \\
& = \int \left[\int_{B(\x_{-i})}\int_{A(\x_{-i})} P_i^{\x_{-i}}(x_i,\d y_i)\pi_i(\d x_i \mid \x_{-i}) \right]\pi(\d \x_{-i}) \\
&= \int_B\int_A P_i(\x, \d \y)\pi(\d \x).
\end{aligned}
\]
Moreover
\[
\int \left(P_if(\x)\right)f(\x) \pi(\d \x) = \int \left[\int \left(P_i^{\x_{-i}}f(\x)\right)f(\x)\pi_i(\d x_i \mid \x_{-i}) \right]\pi(\d \x_{-i}) \geq 0,
\]
so that $P_i$ is $\pi$-reversible and positive semi-definite.

Conversely, assume that $C = \left\{\x_{-i} \, ; \, P_i^{\x_{-i}} \text{ not reversible} \right\} \subset \sX_{-i}$ is such that $\pi(C) > 0$. Then there exist classes $\{ A(\x_{-i})\}_{\x_{-i} \in C}$ and $\{ B(\x_{-i})\}_{\x_{-i} \in C}$ such that $A(\x_{-i}) \subset \sX_i$, $B(\x_{-i}) \subset \sX_i$ and
\[
\int_{A(\x_{-i})}\int_{B(\x_{-i})} P_i^{\x_{-i}}(x_i,\d y_i) \pi_i(\d x_i \mid \x_{-i}) > \int_{B(\x_{-i})}\int_{A(\x_{-i})} P_i^{\x_{-i}}(x_i,\d y_i) \pi_i(\d x_i \mid \x_{-i})
\]
for every $\x_{-i} \in C$. Then, let
\[
A = \left\{\x \, ; \, \x_{-i} \in C, x_i \in A(\x_{-i}) \right\}, \quad \text{and} \quad B = \left\{\x \, ; \, \x_{-i} \in C, x_i \in B(\x_{-i}) \right\}
\]
and notice that
\[
\begin{aligned}
\int_A\int_B P_i(\x, \d \y)\pi(\d \x) &= \int_C \left[\int_{A(\x_{-i})}\int_{B(\x_{-i})} P_i^{\x_{-i}}(x_i,\d y_i) \pi_i(\d x_i \mid \x_{-i}) \right]\pi(\d \x_{-i}) \\
& > \int_C \left[\int_{B(\x_{-i})}\int_{A(\x_{-i})} P_i^{\x_{-i}}(x_i,\d y_i)\pi_i(\d x_i \mid \x_{-i}) \right]\pi(\d \x_{-i}) \\
&= \int_A\int_B P_i(\x, \d \y)\pi(\d \x),
\end{aligned}
\]
so that $P_i$ cannot be $\pi$-reversible. Similarly, if $C = \left\{\x_{-i} \, ; \, P_i^{\x_{-i}} \text{ not positive semi-definite} \right\} \subset \sX_{-i}$ is such that $\pi(C) > 0$, then there exists a class $\{ g^{\x_{-i}}\}_{\x_{-i} \in C}$ such that $g^{\x_{-i}} \, : \, \sX_i \, \to \, \sX_i$ and
\[
\int \left(P_i^{\x_{-i}}g^{\x_{-i}}(x_i)\right)g^{\x_{-i}}(x_i)\pi_i(\d x_i \mid \x_{-i}) < 0,
\]
for every $\x_{-i} \in C$. Then, letting $f(\x) = \mathbbm{1}_{C}(\x_{-i})g^{\x_{-i}}(x_i)$ we have that
\[
\int \left(P_if(\x)\right)f(\x) \pi(\d \x) = \int_C \left[\int \left(P_i^{\x_{-i}}g^{\x_{-i}}(x_i)\right)g^{\x_{-i}}(x_i)\pi_i(\d x_i \mid \x_{-i}) \right]\pi(\d \x_{-i}) < 0,
\]
so that $P_i$ cannot be positive semi-definite.
\end{proof}

\section{Alternative definition of $s$-conductance}\label{alternative_definition}
An alternative definition relative to \eqref{s_conductance} is given by
\begin{equation}\label{modified_conductance}
\tilde{\Phi}_s(P) = \inf \left\{\frac{P(\partial A)}{\pi(A)-s}; \, s < \pi(A) \leq \frac{1}{2}, A \subset \sX \right\}.
\end{equation}
Note that $\Phi_s(P) \leq \tilde{\Phi}_s(P)$ for every $s \geq 0$ and $\tilde{\Phi}_0(P) = \Phi_0(P) = \Phi(P)$.

The next theorem shows that the conclusions of Corollary \ref{cor:bound_conductance} hold also for the conductance defined as in \eqref{modified_conductance}.

\begin{theorem}
Let $P$ be a coordinate-wise operator as in \eqref{eq:coordinate_wise}, with target distribution $\pi \in \sP(\sX)$, and $\kappa(P, K)$ as in \eqref{minimum_conditional_conductance}. Then for every $0\leq s \leq1/2$ and $K\subseteq \sX$ we have
\begin{align}\label{eq:cond_cond_bound_no_K}
\tilde{\Phi}_s(P) \geq 
\kappa(P, \sX)\tilde{\Phi}_{s}(G) 
\end{align}
and
\begin{align}\label{eq:cond_cond_bound}
\tilde{\Phi}_s(P) \geq \kappa(P,K)\frac{2(s-s')}{1-2s}\tilde{\Phi}_{s'}(G)-\frac{2\pi(K^c)}{1-2s} \left(\frac{1}{d}\sum_{i = 1}^d\kappa_i(P_i,K)\right)\,.
\end{align}
\end{theorem}
\begin{proof}
Let $A \subset \sX$. By Theorem \ref{main} we have
\[
\begin{aligned}
P(\partial A) 
&=\frac{1}{d}\sum_{i = 1}^d P_i(\partial A)
\geq \frac{1}{d}\sum_{i = 1}^d\kappa_i(P_i, \sX) G_i(\partial A)
\geq \kappa(P, \sX)\frac{1}{d}\sum_{i = 1}^dG_i(\partial A).
\end{aligned}
\]
The inequality in \eqref{eq:cond_cond_bound_no_K} is obtained from the above by dividing by $\pi(A)-s$ and taking the infimum over $A\subset \sX$ such that $1/2 \geq \pi(A) > s$.

We now consider \eqref{eq:cond_cond_bound}. Let $A \subset \sX$ be such that $1/2 \geq \pi(A) > s$. Since $\pi(A) - s\leq (1-2s)/2$, we have
\[
\begin{aligned}
\frac{P(\partial A)}{\pi(A) - s} &\geq \left(\frac{2}{1-2s} \right)\kappa(P,K)G(\partial A)-\left(\frac{2}{1-2s} \right)\left(\frac{1}{d}\sum_{i = 1}^d\kappa_i(P_i,K)\right)\pi(K^c).
\end{aligned}
\]
Since $\pi(A)>s$ we also have
\[
\begin{aligned}
G(\partial A) &= \frac{\pi(A) - s'}{\pi(A) - s'}G(\partial A) \geq (s-s')\frac{G(\partial A)}{\pi(A) - s'}\,.
\end{aligned}
\]
The above imply
\[
\begin{aligned}
\frac{P(\partial A)}{\pi(A) - s} &\geq \kappa(P,K)\frac{2(s-s')}{1-2s}\frac{G(\partial A)}{\pi(A)-s'}-\frac{1}{d}\sum_{i = 1}^d\kappa_i(P_i,K)\left(\frac{2}{1-2s} \right)\pi(K^c).
\end{aligned}
\]
The desired inequality follows by taking the infimum over $A\subset \sX$ such that $1/2 \geq \pi(A) > s$.
\end{proof}

\section{Extension of Corollary \ref{cor:bound_conductance} to non-uniform updating probabilities}

Given probability weights $\w = (w_1, \dots, w_d)$ such that $w_i > 0$ and $\sum_{i = 1}^dw_i = 1$, we define
\begin{equation}\label{eq:coordinate_wise_nonuniform}
P^\w=\sum_{i=1}^dw_iP_i, \quad P_i(\x, \d \y) = P_i^{\x_{-i}}(x_i, \d y_i)\delta_{\x_{-i}}\left( \d\y_{-i}\right).
\end{equation}
The choice $w_i = 1/d$ leads to the uniform case in \eqref{eq:coordinate_wise}. Similarly we define $G^\w = \sum_{i=1}^dw_iG_i$. The next corollary shows that the same results of Corollary \ref{cor:bound_conductance} hold for arbitrary weights.

\begin{corollary}\label{cor:bound_conductance_nonuniform}
Let $P^\w$ be as in \eqref{eq:coordinate_wise_nonuniform}. Then for every $s \in [0, 1/2)$ and $K\subseteq \sX$ we have
\begin{align}\label{eq:cond_bound}
\Phi_s(P^\w)& \geq 
\kappa(P^\w, \sX)\Phi_{s}(G^\w) \quad \text{and} \quad \Phi_s(P^\w) \geq 
\kappa(P^\w,K)\Phi_{s}(G^\w)-\frac{\pi(K^c)}{s}\left(\sum_{i = 1}^dw_i\kappa_i(P_i,K)\right)\,,
\end{align}
with $\kappa(P^\w, K)$ as in \eqref{minimum_conditional_conductance}.
Moreover, $\Phi_s(G^\w) \geq \Phi_s(P^\w)$.
\end{corollary}
\begin{proof}
Let $A \subset \sX$. By Theorem \ref{main} we have
\[
\begin{aligned}
P^\w(\partial A) 
&= \sum_{i = 1}^d w_iP_i(\partial A)
\geq \sum_{i = 1}^dw_i\kappa_i(P_i, \sX) G_i(\partial A)
\geq \kappa(P^\w, \sX)\sum_{i = 1}^dw_iG_i(\partial A).
\end{aligned}
\]
The first inequality in \eqref{eq:cond_bound} is obtained dividing by $\pi(A)$ and taking the infimum over $A\subset \sX$ such that $1/2 \geq \pi(A) > s$.

As regards the other inequality, again by Theorem \ref{main} we have
\begin{align*}
P^\w(\partial A) 
&\geq \sum_{i = 1}^dw_i\kappa_i(P_i,K)G_i(\partial A)-\left(\sum_{i = 1}^dw_i\kappa_i(P_i,K)\right)\pi(A \cap K^c)\\
& \geq \kappa(P^\w,K)G^\w(\partial A)-\left(\sum_{i = 1}^dw_i\kappa_i(P_i,K)\right)\pi(K^c).
\end{align*}
Let now $A \subset \sX$ be such that $1/2 \geq \pi(A) > s$. Therefore by the above we have
\[
\frac{P^\w(\partial A)}{\pi(A)} \geq \kappa(P^\w,K)\frac{G^\w(\partial A)}{\pi(A)}-\left(\sum_{i = 1}^dw_i\kappa_i(P_i,K)\right)\frac{\pi(K^c)}{s},
\]
from which the right inequality in \eqref{eq:cond_bound} follows.

Finally, since $P^{\x_{-i}}$ is reversible and positive semi-definite, again by Theorem \ref{main} we have
\[
G^\w(\partial A) = \sum_{i = 1}^dw_i G_i(\partial A) \geq \sum_{i = 1}^dw_i P_i(\partial A) = P(\partial A),
\]
from which we immediately deduce $\Phi_s(G^\w) \geq \Phi_s(P^\w)$.
\end{proof}

\section{Regularity assumptions (B4)-(B6) for model \eqref{one_level_nested}}\label{regularity_assumptions}
Define $\bm{T} = \left(\sum_{j = 1}^JT_1(\theta_j), \dots, \sum_{j = 1}^JT_S(\theta_j) \right)$, with $T_s$ as in \eqref{exponential_family}, and let
\begin{align}\label{posterior_moments}
M^{(p)}_{s}(\psi \mid y) &= E \left[T^p_s(\theta_j) \mid Y_j = y, \psi \right]\,,\\
M^{(p)}_{s,s'}(\psi \mid y) &= E \left[T^p_s(\theta_j)T^p_{s'}(\theta_j) \mid Y_j = y, \psi \right],
\end{align}
be the posterior moments of $\bm{T}$ given $\psi$, denote $M^{(p)}(\psi \mid y) = \left(M_1^{(p)}(\psi \mid y), \dots, M_S^{(p)}(\psi \mid y) \right) \in \R^S$  and
\begin{equation}\label{def_C_V}
\left[C(\psi)\right]_{s,d} = E_{Y_j} \left[\partial_{\psi_d}M_s^{(1)}\left(\psi \mid Y_j \right) \right], \quad \left[V(\psi) \right]_{s,s'} = E_{Y_j} \left[\text{Cov}\left(T_s(\theta_j), T_{s'}(\theta_j) \mid Y_j, \psi \right) \right],
\end{equation}
with $s, s' = 1, \dots S$ and $d = 1, \dots, D$. Moreover we write $B_\delta$ for the ball of center $\psi^*$ and radius $\delta$, and denote expectations with respect to the law of $Y_j$ as defined in $(B1)$ by $E_{Y_j}[\cdot]$. Finally, we define the posterior characteristic function of $T(\theta_j) = \left(T_1(\theta_j), \dots, T_S(\theta_j)   \right)$ and $\sum_{j = 1}^kT(\theta_j)$, given $\psi$, as
$
\varphi\left(t \mid Y_j, \psi \right)= E \left[e^{it^\top T(\theta_j)} \mid Y_j, \psi \right]$ for $t\in\R^S$.
and
$
\varphi^{(k)}\left(t \mid Y_{1:k}, \psi \right) = \prod_{j = 1}^k\varphi\left(t \mid Y_j, \psi \right)
$, 
respectively.
Assumptions $(B4)-(B6)$ now read:
\begin{enumerate}
\item[$(B4)$] The expectation $M^{(p)}_{s}(\psi \mid y)$ is well defined for every $y$ and $p = 1, \dots, 6$. Moreover, there exist $\delta_4 > 0$ and $C$ finite constant such that for every $\psi \in B_{\delta_4}$ it holds
$E_{Y_j}\left[\left \lvert  \partial_{\psi_d}M^{(6)}_s(\psi \mid Y_j) \right\rvert \right] < C$, 
$E_{Y_j}\left[\left \lvert \partial_{\psi_d}\partial_{\psi_{d'}}M^{(1)}_s(\psi \mid Y_j) \right\rvert \right] < C$, \newline $
E_{Y_j}\left[\left \lvert  \partial_{\psi_d}M^{(1)}_{s, s'}(\psi \mid Y_j) \right\rvert \right] < C$ and $E_{Y_j}\left[\left \lvert  \partial_{\psi_d}\left\{M^{(1)}_{s}(\psi \mid Y_j)M^{(1)}_{s'}(\psi \mid Y_j)\right\} \right\rvert \right] < C$
for $s, s' = 1, \dots, S$ and  $d, d' = 1, \dots, D$.
 Finally, the matrix $V(\psi^*)$ defined in \eqref{def_C_V} is non singular.
\item[$(B5)$] There exist $k \geq 1$ and $\delta_5 >0$ such that
\[
\sup_{\psi \in B_{\delta_5}} \, \int_{\R^S} \left \lvert\varphi^{(k)}\left(t \mid Y_{1:k}, \psi \right) \right\rvert^2 \, \d t < \infty,
\]
for almost every $Y_1,\dots,Y_k \simiid Q_{\psi^*}$.
\item[$(B6)$] There exist $k' \geq 1$ and $\delta_6 >0$ such that
\[
\sup_{\psi \in B_{\delta_6}} \, \sup_{|t| > \epsilon} \left \lvert \varphi^{(k')}\left(t \mid Y_{1:k'}, \psi \right) \right \rvert < \phi(\epsilon),
\]
for almost every $Y_1,\dots,Y_k \simiid Q_{\psi^*}$, with $\phi(\epsilon) < 1$ for every $\epsilon > 0$.
\end{enumerate}
Some discussion on the interpretation and applicability of assumptions $(B4)$-$(B6)$ can be found in Appendix B of \cite{AZ23}.

\section{Proofs}
\subsection{Proof of Lemma \ref{conductance_convergence}}
\begin{proof}
Using the notation $\tilde{\Phi}_s(P)$ in \eqref{modified_conductance}, by Corollary $1.5$ (point (b)) in \cite{LS93} we have that
\[
\lTV \mu P^t - \pi \rTV \leq Ms + M\left(1-\frac{\tilde{\Phi}^2_s(P)}{2} \right)^t\,.
\]
The result now follows by noticing that $\Phi_s(P) \leq \tilde{\Phi}_s(P)$ for every $s \geq 0$.
\end{proof}
\subsection{Proof of Theorem \ref{main}}
For every $A \subseteq \sX$, $i = 1, \dots, d$ and $\x_{-i} \in \sX_{-i}$  we denote
\[
A_i(\x_{-i}) = \{x_i\in\sX_i\,:\, \x \in A\}, \quad A_i^c(\x_{-i}) = \{x_i \in\sX_i\,:\, \x \not\in A\}.
\]
Notice that $A_i(\x_{-i}) \subseteq \sX_i$ and $A_i^c(\x_{-i}) \subseteq \sX_i$.

\begin{proof}[Proof of Theorem \ref{main}]
By \eqref{eq:coordinate_wise} we have
\begin{equation}\label{key_equality_coord}
\begin{aligned}
&P_i(\partial A)
 = \int_A P_i(\x, A^c) \pi(\d \x)= \int_{\x_{-i} }  \int_{A_i(\x_i)}P_i^{\x_{-i}}\left(x_i, A_i^c(\x_{-i})\right)\pi_i(\d x_i \mid \x_{-i}) \pi (\d \x_{-i})\\
& = \int_{\x_{-i} }  \left[\frac{\int_{A_i(\x_i)}P^{\x_{-i}}_i\left(x_i,A_i^c(\x_{-i})\right)\pi_i(\d x_i \mid \x_{-i})}{\pi_{i}\left(A_i \left(\x_{-i}\right) \mid \x_{-i} \right) \pi_{i} \left(A_i^c \left(\x_{-i} \right) \mid \x_{i}\right)}\right]\pi_{i}\left(A_i \left(\x_{-i}\right) \mid \x_{-i} \right) \pi_{i} \left(A_i^c \left(\x_{-i}\right) \mid \x_{-i} \right) \pi (\d \x_{-i})\,,
\end{aligned}
\end{equation}
where, with a slight abuse of notation, $\pi (\d \x_{-i})$ denotes the marginal distribution of $X_{-i}$ under $X\sim\pi$.  

If $\pi_{i}\left(A_i \left(\x_{-i}\right) \mid \x_{-i} \right) \leq 1/2$, by \eqref{conditional_conductance} we have
\[
\frac{\int_{A_i(\x_i)}P_i^{\x_{-i}}\left(x_i, A_i^c(\x_{-i})\right)\pi_i(\d x_i \mid \x_{-i})}{\pi_{i}\left(A_i \left(\x_{-i}\right) \mid \x_{-i} \right) \pi_{i} \left(A_i^c \left(\x_{-i}\right) \mid \x_{-i} \right)}\geq \frac{\int_{A_i(\x_i)}P_i^{\x_{-i}}\left(x_i, A_i^c(\x_{-i})\right)\pi_i(\d x_i \mid \x_{-i})}{\pi_{i}\left(A_i \left(\x_{-i}\right) \mid \x_{-i} \right)} \geq \kappa\left(P_i^{\x_{-i}}\right).
\]
If instead $\pi_{i}\left(A_i^c \left(\x_{-i}\right) \mid \x_{-i} \right) \leq 1/2$, since $P^{\x_{-i}}$ is reversible we have
\[
\begin{aligned}
\frac{\int_{A_i(\x_i)}P_i^{\x_{-i}}\left(x_i, A_i^c(\x_{-i})\right)\pi_i(\d x_i \mid \x_{-i})}{\pi_{i}\left(A_i \left(\x_{-i}\right) \mid \x_{-i} \right) \pi_{i} \left(A_i^c \left(\x_{-i}\right) \mid \x_{-i} \right)} & = \frac{\int_{A_i^c(\x_i)}P_i^{\x_{-i}}\left(x_i, A_i(\x_{-i})\right)\pi_i(\d x_i \mid \x_{-i})}{\pi_{i}\left(A_i \left(\x_{-i}\right) \mid \x_{-i} \right) \pi_{i} \left(A_i^c \left(\x_{-i}\right) \mid \x_{-i} \right)}\\
&\geq \frac{\int_{A_i^c(\x_i)}P_i^{\x_{-i}}\left(x_i, A_i(\x_{-i})\right)\pi_i(\d x_i \mid \x_{-i})}{\pi_{i}\left(A_i^c \left(\x_{-i}\right) \mid \x_{-i} \right)} \geq \kappa\left(P_i^{\x_{-i}}\right).
\end{aligned}
\]
In conclusion
\[
\frac{\int_{A_i(\x_i)}P_i^{\x_{-i}}\left(x_i, A_i^c(\x_{-i})\right)\pi_i(\d x_i \mid \x_{-i})}{\pi_{i}\left(A_i \left(\x_{-i}\right) \mid \x_{-i} \right) \pi_{i} \left(A_i^c \left(\x_{-i}\right) \mid \x_{-i} \right)}\geq \kappa\left(P_i^{\x_{-i}}\right),
\]
and thus
\[
\begin{aligned}
P_i(\partial A)
&\geq
\int_{\x_{-i} }\kappa\left(P_i^{\x_{-i}}\right)\pi_{i}\left(A_i \left(\x_{-i}\right) \mid \x_{-i} \right) \pi_{i} \left(A_i^c \left(\x_{-i}\right) \mid \x_{-i} \right)  \pi(\d \x_{-i}) 
  \\
&=  \int_{\x\in A}\kappa\left(P_i^{\x_{-i}}\right)\pi_{i} (A_i^c (\x_{-i})\mid \x_{-i} )  \pi(\d \x)
  \\
&\geq \int_{\x \in A \cap K }\kappa\left(P_i^{\x_{-i}}\right) \pi_{i} \left(A_i^c \left(\x_{-i}\right) \mid \x_{-i} \right)  \pi(\d \x)    \\
&\geq \kappa_i(P_i,K)\,  \int_{\x \in A \cap K } \pi_{i} \left(A_i^c \left(\x_{-i}\right) \mid \x_{-i} \right)  \pi(\d \x).
\end{aligned}
\]
Moreover
\[
\begin{aligned}
\int_{\x \in A \cap K } \pi_{i} \left(A_i^c \left(\x_{-i}\right) \mid \x_{-i}\right)  \pi(\d \x) &= \int_{\x \in A } \pi_{i} \left(A_i^c \left(\x_{-i}\right) \mid \x_{-i} \right)  \pi(\d \x)-\int_{\x \in A \cap K^c } \pi_{i} \left(A_i^c \left(\x_{-i}\right) \mid \x_{-i} \right)  \pi(\d \x)\\
&\geq \int_{\x \in A } \pi_{i} \left(A_i^c \left(\x_{-i}\right) \mid \x_{-i} \right)  \pi(\d \x)-\pi\left(A \cap K^c\right).
\end{aligned}
\]
Combining the two previous inequalities and using $\pi_{i} (A_i^c (\x_{-i} ) \mid \x_{-i} )=G_{i}^{\x_{-i}}(x_i, A_i^c \left(\x_{-i}\right)) $ we get
\[
\begin{aligned}
P_i(\partial A)&\geq \kappa_i(P_i,K)\left(\int_{\x \in A } \pi_{i} (A_i^c (\x_{-i} ) \mid \x_{-i} )  \pi(\d \x)-\pi\left(A \cap K^c\right) \right)\\
& = \kappa_i(P_i,K)\left(\int_{\x \in A} \pi (\d\x ) G_{i}(\x, A^c)-\pi(A \cap K^c)\right)
\end{aligned}
\]
as desired. 

As regards the second point, since $P_i^{\x_{-i}}$ is reversible and positive semi-definite, it is possible to show that
\begin{equation}\label{cons_positive_def}
\frac{\int_{A_i(\x_i)}P_i^{\x_{-i}}\left(x_i, A_i^c(\x_{-i})\right)\pi(\d x_i \mid \x_{-i})}{\pi_{i}\left(A_i \left(\x_{-i}\right) \mid \x_{-i} \right) \pi_{i} \left(A_i^c \left(\x_{-i}\right) \mid \x_{-i} \right)} \leq 1.
\end{equation}
Indeed, since $P_i^{\x_{-i}}$ is invariant with respect to $\pi(\cdot \mid \x_{-i})$ it holds
\[
\int_{A_i(\x_i)}P_i^{\x_{-i}}\left(x_i, A_i^c(\x_{-i})\right)\pi(\d x_i \mid \x_{-i})+\int_{A_i^c(\x_i)}P_i^{\x_{-i}}\left(x_i, A_i^c(\x_{-i})\right)\pi(\d x_i \mid \x_{-i}) = \pi_i\left(A_i^c(\x_{-i}) \mid \x_{-i} \right).
\]
Moreover, since $P_i^{\x_{-i}}$ is positive semi-definite, by e.g. Lemma $1.1$ in \cite{LS93} we have
\[
\int_{A_i^c(\x_i)}P_i^{\x_{-i}}\left(x_i, A_i^c(\x_{-i})\right)\pi(\d x_i \mid \x_{-i}) \geq \pi_i^2\left(A_i^c(\x_{-i}) \mid \x_{-i}\right),
\]
from which \eqref{cons_positive_def} follows. Combining \eqref{key_equality_coord} with \eqref{cons_positive_def} we obtain
\[
P_i(\partial A) \leq \int_{\x_{-i} } \pi_{i}\left(A_i \left(\x_{-i}\right) \mid \x_{-i} \right) \pi_{i} \left(A_i^c \left(\x_{-i}\right) \mid \x_{-i}\right) \pi (\d \x_{-i}) = G_i(\partial A).
\]
\end{proof}
\subsection{Proof of Corollary \ref{cor:bound_conductance}}
\begin{proof}
Let $A \subset \sX$. By Theorem \ref{main} we have
\[
\begin{aligned}
P(\partial A) 
&=\frac{1}{d}\sum_{i = 1}^d P_i(\partial A)
\geq \frac{1}{d}\sum_{i = 1}^d\kappa_i(P_i, \sX) G_i(\partial A)
\geq \kappa(P, \sX)\frac{1}{d}\sum_{i = 1}^dG_i(\partial A).
\end{aligned}
\]
The first inequality in \eqref{eq:cond_bound} is obtained dividing by $\pi(A)$ and taking the infimum over $A\subset \sX$ such that $1/2 \geq \pi(A) > s$.

As regards the other inequality, again by Theorem \ref{main} we have
\begin{align*}
P(\partial A) 
&\geq \frac{1}{d}\sum_{i = 1}^d\kappa_i(P_i,K)G_i(\partial A)-\left(\frac{1}{d}\sum_{i = 1}^d\kappa_i(P_i,K)\right)\pi(A \cap K^c)\\
& \geq \kappa(P,K)G(\partial A)-\left(\frac{1}{d}\sum_{i = 1}^d\kappa_i(P_i,K)\right)\pi(K^c).
\end{align*}
Let now $A \subset \sX$ be such that $1/2 \geq \pi(A) > s$. Therefore by the above we have
\[
\frac{P(\partial A)}{\pi(A)} \geq \kappa(P,K)\frac{G(\partial A)}{\pi(A)}-\left(\frac{1}{d}\sum_{i = 1}^d\kappa_i(P_i,K)\right)\frac{\pi(K^c)}{s},
\]
from which the right inequality in \eqref{eq:cond_bound} follows.

Finally, since $P^{\x_{-i}}$ is reversible and positive semi-definite, again by Theorem \ref{main} we have
\[
G(\partial A) = \frac{1}{d}\sum_{i = 1}^d G_i(\partial A) \geq \frac{1}{d}\sum_{i = 1}^d P_i(\partial A) = P(\partial A),
\]
from which we immediately deduce $\Phi_s(G) \geq \Phi_s(P)$.
\end{proof}
\subsection{Proof of Proposition \ref{prop:conductance_ind}}
\begin{proof}
By definition of $M$, for every $x_i$ and $y_i$ we have
\[
\alpha(x_i, y_i) \geq \min \left\{ 1, \frac{1}{M}\frac{ \d \pi_i(\cdot \mid \x_{-i})}{\d Q_i}(y_i)\right\} = \frac{1}{M}\frac{ \d \pi_i(\cdot \mid \x_{-i})}{\d Q_i}(y_i),
\]
so that by \eqref{eq:operator_MH} it holds
\[
P_i^{\x_{-i}}(x_i, B^c) \geq \frac{1}{M}\int_{B^c}\frac{ \d \pi_i(\cdot \mid \x_{-i})}{\d Q_i}(y_i)Q_i(\d y_i) = \frac{1}{M}\pi_i(B^c \mid \x_{-i}),
\]
for every $B \subset \sX_i$. This implies that $\kappa(P_i^{\x_{-i}}) \geq 1/M$ for every $i = 1, \dots, d$ and the result follows by Corollary \ref{cor:bound_conductance}.
\end{proof}
\subsection{Proof of Proposition \ref{prop:IMH_logconcave}}
We need a preliminary lemma.
\begin{lemma}\label{lemma:prop_logconcave}
Let $\pi$ be a log-concave distribution on $\R^d$ with parameters $m$ and $l$ and mode $\x^*$. Then for very $\x \in \R^d$ it holds
\[
\frac{\pi(\x)}{N\left(\x \mid x^*, \frac{1}{m}\I_d\right)} \leq \left(\frac{L}{m}\right)^{\frac{d}{2}}.
\]
\end{lemma}
\begin{proof}
Denoting $f(\x) = -\log \pi(\x)$ we have
\begin{equation}\label{bound_logratio_logconcave}
\begin{aligned}
\log \frac{\pi(\x)}{N\left(\x \mid x^*, \frac{1}{m}\I_d\right)} &= -f(\x)+\frac{m}{2}\lE\x - \x^* \rE^2-\frac{d}{2}\log \left(\frac{m}{2\pi}\right)\\
& = -f(\x^*)-\left[f(\x) -f(\x^*)-\frac{m}{2}\lE\x - \x^* \rE^2\right]-\frac{d}{2}\log \left(\frac{m}{2\pi}\right)\\
& \geq -f(\x^*)-\frac{d}{2}\log \left(\frac{m}{2\pi}\right),
\end{aligned}
\end{equation}
by \eqref{eq:m_convex_L_smooth}. Moreover, notice that
\[
1 = \int_{\R^d}e^{-f(\x)}\, \d \x = e^{-f(\x^*)}\int_{\R^d}e^{-\left[f(\x)-f(\x^*)\right]}\, \d \x.
\]
By \eqref{eq:m_convex_L_smooth} it holds $f(\x)-f(\x^*) \leq \frac{L}{2}\lE\x - \x^* \rE^2$ which implies
\[
\int_{\R^d}e^{-\left[f(\x)-f(\x^*)\right]}\, \d \x \geq \int_{\R^d}e^{-\frac{L}{2}\lE\x - \x^* \rE^2}\, \d \x = \left(\frac{2\pi}{L} \right)^{\frac{d}{2}},
\]
so that
\[
e^{-f(\x^*)} \leq \left(\frac{L}{2\pi} \right)^{\frac{d}{2}}.
\]
Combining this with \eqref{bound_logratio_logconcave} the result follows.
\end{proof}
\begin{proof}[Proof of Proposition \ref{prop:IMH_logconcave}]
By Lemma \ref{lemma:prop_logconcave} it is possible to apply Proposition \ref{prop:conductance_ind} with $M = \min_i \inf_{\x_{-i}}\, \left(\frac{m(\x_{-i})}{L(\x_{-i})}\right)^{\frac{d_i}{2}}$, from which the result follows.
\end{proof}
\subsection{Proof of Proposition \ref{prop:RWM_logconcave}}
\begin{proof}
The first part of the statement follows directly from Corollary $35$ in \cite{AL22}, while the second part is a consequence of Corollary \ref{cor:bound_conductance}.
\end{proof}
\subsection{Proof of Theorem \ref{theorem:distance_operators}}
\begin{proof}
For every $i = 1, 2$ and $A \subset \sX$ denote with $\pi_i^A(\cdot) = \pi_i(\cdot \cap A)/\pi_i(A)$ the restriction of $\pi_i$ to $A$. It is clear that $\pi_i^A \in \sN\left(\pi_i, 1/\pi_i(A)\right)$. Fix now $A \subset \sX$ such that $s < \pi_1(A) \leq A$ and notice that by the triangular inequality we have
\[
\begin{aligned}
\left \lvert \frac{P_1(\partial A)}{\pi_1(A)}-\frac{P_2(\partial A)}{\pi_2(A)}  \right \rvert &= \left \lvert \pi_1^AP_1(A^c)-\pi_2^AP_2(A^c)  \right \rvert\\
& \leq \left \lvert \pi_1^AP_1(A^c)-\pi_1^AP_2(A^c)  \right \rvert+\left \lvert \pi_1^AP_2(A^c)-\pi_2^AP_2(A^c)  \right \rvert\\
& \leq \Delta(P_1, P_2, 1/s) + \lTV \pi_1^A-\pi_2^A \rTV.
\end{aligned}
\]
Moreover, for every $B \subset \sX$ we have
\[
\begin{aligned}
\left \lvert \pi_1^A(B)-\pi_2^A(B) \right \rvert &= \left \lvert \frac{\pi_1(B \cap A)}{\pi_1(A)}-\frac{\pi_2(B \cap A)}{\pi_2(A)} \right \rvert\\
& \leq \left \lvert \frac{\pi_1(B \cap A)}{\pi_1(A)}-\frac{\pi_2(B \cap A)}{\pi_1(A)} \right \rvert+\left \lvert \frac{\pi_2(B \cap A)}{\pi_1(A)}-\frac{\pi_2(B \cap A)}{\pi_2(A)} \right \rvert\\
& \leq \frac{\delta}{s}+\left \lvert \frac{1}{\pi_1(A)}-\frac{1}{\pi_2(A)} \right \rvert \pi_2(A) = \frac{\delta}{s}+\left \lvert\frac{\pi_2(A)-\pi_1(A)}{\pi_1(A)} \right \rvert\\
& \leq \frac{2\delta}{s}.
\end{aligned}
\]
Therefore we have
\[
\left \lvert \frac{P_1(\partial A)}{\pi_1(A)}-\frac{P_2(\partial A)}{\pi_2(A)}  \right \rvert \leq \Delta(P_1, P_2, 1/s)+\frac{2\delta}{s},
\]
for every $A \subset \sX$ such that $s < \pi_1(A) \leq A$, so that it holds
\[
\begin{aligned}
\Phi_s(P_1) &= \inf \left\{\frac{P_1(\partial A)}{\pi_1(A)}, \, s < \pi_1(A) \leq \frac{1}{2} \right\}\\
& \geq \inf \left\{\frac{P_2(\partial A)}{\pi_2(A)}, \, s < \pi_1(A) \leq \frac{1}{2} \right\}+\Delta(P_1, P_2, 1/s)+\frac{2\delta}{s}\\
& \geq \inf \left\{\frac{P_2(\partial A)}{\pi_2(A)}, \, s-\delta < \pi_2(A) \leq \frac{1}{2} \right\}+\Delta(P_1, P_2, 1/s)+\frac{2\delta}{s}\\
& = \Phi_{s-\delta}(P_2)-\Delta\left(P_1, P_2, 1/s \right)-\frac{2\delta}{s},
\end{aligned}
\]
as desired.
\end{proof}

\subsection{Proof of Lemma \ref{lemma:Delta_GS}}
It is easy to prove that an equivalent way to define $\sN$ as in \eqref{N_class} is given by
\begin{align}\label{equivalent_N_class}
\sN\left(\pi, M \right)=&
\left\{\mu\in\mathcal{P}(\sX)\,:\,\frac{\d \mu}{\d \pi}(\x) \leq M \hbox{ for all }A\subseteq \sX\right\}, &M \geq 1,\,\pi \in \mathcal{P}(\sX) \,,
\end{align}
where $\frac{\d \mu}{\d \pi}$ denotes the Radon-Nikodym derivative of $\mu$ with respect to $\pi$.
\begin{proof}[Proof of Lemma \ref{lemma:Delta_GS}]
Recall that, for $\pi_1 \in \sP(\sX)$ and $\pi_2 \in \sP(\sX)$, the total variation distance is defined as
\[
\lTV \pi_1 - \pi_2 \rTV = \sup_{f \in \sH(\sX)} \, \biggl\lvert \int_\sX f(\x) \pi_1(\d \x) - \int_\sX f(\x) \pi_2(\d \x) \biggr\rvert,
\]
where
\[
\sH(\sX) = \left\{ f \mid 
 f \, : \, \sX \, \to \, [0,1] \text{ measurable} \right\}.
\]
For notational convenience, in the following we denote
\[
f(\x_{-i}, y_i) = f(x_1, \dots, x_{i-1}, y_i, x_{i+1}, \dots, x_d),
\]
for every $f \in \sH(\sX)$ and $ i = 1, \dots, d$. Moreover, we write $\pi_{j, i}(\d x_i \mid \d \x_{_i})$, with $j = 1,2$ and $i = 1, \dots, d$, to denote the conditional distribution of the $i$-th coordinate induced by $\pi_j$.

Let now $f \in \sH(\sX)$ and $\mu \in \sN(\pi_1, M)$. Then with $j = 1,2$ we have
\[
\begin{aligned}
\int_{\sX} f(\y) \mu G_{j,i}(\d \y) &= \int_{\sX} \int_{\sX_i}f(\x_{-i}, y_i) \pi_{j,i}(\d y_i \mid \x_{-i}) \mu(\d \x)\\
& = \int_{\sX} \int_{\sX_i}\frac{\d \mu}{\d \pi_{1}}(\x)f( \x_{-i}, y_i) \pi_{j,i}(\d y_i \mid \x_{-i}) \pi_1(\d \x).
\end{aligned}
\]
Then we have
\[
\begin{aligned}
\biggl \lvert \int_{\sX}& f(\d \y) \mu G_{1,i}(\d \y) - \int_{\sX} f(\d \y) \mu G_{2,i}(\d \y) \biggr \rvert \\
& = M \left \lvert\int_{\sX} \int_{\sX_i}\frac{\d \mu}{\d \pi_{1}}(\x)\frac{f(\x_{-i}, y_i)}{M} \pi_{1,i}(\d y_i \mid \x_{-i}) \pi_1(\d \x)-\int_{\sX} \int_{\sX_i}\frac{\d \mu}{\d \pi_{1}}(\x)\frac{f(\x_{-i}, y_i)}{M} \pi_{2,i}(\d y_i \mid \x_{-i}) \pi_1(\d \x) \right \rvert\\
& = M \left \lvert\int_{\sX} \int_{\sX_i}g(\x, y_i) \pi_{1,i}(\d y_i \mid \x_{-i})\pi_1(\d \x)-\int_{\sX} \int_{\sX_i}g(\x, y) \pi_{2,i}(\d y_i \mid \x_{-i}) \pi_1(\d \x) \right \rvert,
\end{aligned}
\]
where $g(\x, y_i) = \frac{\d \mu}{\d \pi_{1}}(\x)\frac{f(\d \x_{-i}, y)}{M} \in \sH(\sX \times \sX_i)$, by \eqref{equivalent_N_class}. Moreover, notice that we can disintegrate as follows
\[
\pi_{j,i}(\d y_i \mid \x_{-i})\pi_1(\d \x) = \pi_{1,i}(\d x_i \mid \x_{-i})\pi_{j,i}(\d y_i \mid \x_{-i})\pi_1(\d \x_{-i}),
\]
with $j = 1, 2$. Therefore
\[
\begin{aligned}
\biggl \lvert &\int_{\sX} f(\y) \mu G_{1,i}(\d \y) - \int_{\sX} f(\y) \mu G_{2,i}(\d \y) \biggr \rvert \\
& = 
\begin{aligned}
    M \biggl\lvert\int_{\sX}\int_{\sX_i}g(\x, y_i)\pi_{1,i}(\d x_i \mid \x_{-i}) &\pi_{1,i}(\d y_i \mid \x_{-i})\pi_1(\d \x_{-i})\\
    &-\int_{\sX}\int_{\sX_i}g(\x, y_i)\pi_{1,i}(\d x_i \mid \x_{-i}) \pi_{2,i}(\d y_i \mid \x_{-i}) \pi_1(\d \x_{-i}) \biggr \rvert
    \end{aligned}\\
& =M \left \lvert\int_{\sX}h(\x_{-i}, y_i) \pi_{1,i}(\d y_i \mid \x_{-i})\pi_1(\d \x_{-i})-\int_{\sX}h(\x_{-i}, y_i) \pi_{2,i}(\d y_i \mid \x_{-i}) \pi_1(\d \x_{-i}) \right \rvert,
\end{aligned}
\]
where $h(\x_{-i}, y_i) = \int_{\sX_i}g(\x, y)\pi_{1,i}(\d x_i \mid \x_{-i}) \, \in \sH(\sX)$. Thus we get that
\[
\lTV \mu G_{1,i}-\mu G_{2,i} \rTV \leq M \sup_{h \in \sH(\sX)} \, \left \lvert\int_{\sX}h(\x_{-i}, y_i) \pi_{1,i}(\d y_i \mid \x_{-i})\pi_1(\d \x_{-i})-\int_{\sX}h(\x_{-i}, y_i) \pi_{2,i}(\d y_i \mid \x_{-i}) \pi_1(\d \x_{-i}) \right \rvert.
\] 
By the triangular inequality, for every $h \in \sH(\sX)$ we have
\[
\begin{aligned}
\biggl \lvert \int_{\sX}h(\x_{-i}, y_i)& \pi_{1,i}(\d y_i \mid \x_{-i})\pi_1(\d \x_{-i})-\int_{\sX}h(\x_{-i}, y_i) \pi_{2,i}(\d y_i \mid \x_{-i}) \pi_1(\d \x_{-i}) \biggr \rvert \\
&\leq \biggl \lvert\int_{\sX}h(\x_{-i}, y_i) \pi_{1,i}(\d y_i \mid \x_{-i})\pi_1(\d \x_{-i})-\int_{\sX}h(\x_{-i}, y_i) \pi_{2,i}(\d y_i \mid \x_{-i}) \pi_2(\d \x_{-i}) \biggr \rvert\\
&+\biggl \lvert\int_{\sX}h(\x_{-i}, y_i) \pi_{2,i}(\d y_i \mid \x_{-i})\pi_2(\d \x_{-i})-\int_{\sX}h(\x_{-i}, y_i) \pi_{2,i}(\d y_i \mid \x_{-i}) \pi_1(\d \x_{-i}) \biggr \rvert\\
& \leq 2\lTV \pi_1-\pi_2 \rTV,
\end{aligned}
\]
as desired.
\end{proof}

\subsection{Proof of Theorem \ref{theorem:delta_Gibbs_operators}}
\begin{proof}
Let $M=1/s$.
By $G_1=d^{-1}\sum_{i=1}^dG_{1,i}$,  $G_2=d^{-1}\sum_{i=1}^dG_{2,i}$, and the definition of $\Delta$ in \eqref{delta_operators}, we have
\begin{align*}
\Delta(G_1, G_2, 1/s) 
&\leq 
\frac{1}{d}
\sum_{i=1}^d
\Delta(G_{1,i}, G_{2,i}, 1/s)\,,
\end{align*}
where we used the triangle inequality for the total variation norm and the fact that $G_1$ and $G_{1,i}$ for $i=1,\dots,d$ share the same invariant distribution $\pi_1$. 
Combining the latter with Lemma \ref{lemma:Delta_GS} we obtain 
\begin{align*}
\Delta(G_1, G_2, 1/s) 
&\leq 
\frac{2\delta}{s}\,.
\end{align*}
The desired statement follows by combining the above inequality with 
Theorem \ref{theorem:distance_operators}.
\end{proof}

\subsection{Proof of Lemma \ref{lemma_product_conductance}}
\begin{proof}
Denote with Gap$(P)$ the spectral gap of the operator $P$ (see e.g. \cite{AL22} for the definition). It is well-known (see e.g.\ Lemma $2$ in \cite{PZ20}) that Gap$(P) = \min_{j \in \{1, \dots, J\}} \, \text{Gap}(P_j)$. By Cheeger's bounds (e.g. Lemma 5 in \cite{AL22}) we have
\[
\Phi(P) \geq \frac{\text{Gap}(P)}{2} = \min_{j \in \{1, \dots, J\}} \, \frac{\text{Gap}(P_j)}{2} \geq \min_{j \in \{1, \dots, J\}} \, \frac{\Phi^2(P_j)}{4},
\]
as desired.
\end{proof}
\subsection{Proof of Theorem \ref{theorem_one_level_nested}}
\begin{proof}
Fix $s> 0$ and let $K^* = S^*\times \R^{\ell J}$. By Corollary \ref{cor:bound_conductance} we have
\[
\Phi_s(P_j) \geq \kappa\left(P_J, K^*\right)\Phi_{s}(G_J)-\frac{1-\pi_J(K^*)}{s}.
\]
By assumption $(C)$ and \eqref{eq:product_cond_conductance} we have $\kappa\left(P, K^*\right) \geq \kappa^2/4$, so that
\[
\Phi_s(P_j) \geq \frac{\kappa^2}{4}\Phi_{s}(G_J)-\frac{1-\pi_J(K^*)}{s}.
\]
Moreover, by the Bernstein-von Mises Theorem (e.g. Theorem $10.1$ in \cite{V00}), whose assumptions are met thanks to $(B1)-(B3)$, we can deduce
\[
Q_{\psi^*}^{(J)}\left(\pi_J(K^*) \to 1 \right) \to 1,
\]
as $J \to \infty$. Thus we conclude
\[
\begin{aligned}
Q_{\psi^*}^{(J)}\left(\Phi_s(P_J) \geq \frac{\kappa^2}{8}\Phi_s(G_J) \right) \leq Q_{\psi^*}^{(J)}\left(\frac{1-\pi_J(K^*)}{s} \leq  \frac{\kappa^2}{8}\Phi_s(G_J) \right) \to 1,
\end{aligned}
\]
as $J \to \infty$.
\end{proof}

\subsection{Proof of Theorem \ref{theorem: gibbs_one_level_nested}}
This proof requires multiple steps, which we summarize here:
\begin{enumerate}
\item We show that the $s$-conductance of an arbitrary operator $P$ can be lower bounded by the associated mixing times (Lemma \ref{lemma:s_conductance_mix_time}).
\item We prove that the mixing times of the Gibbs sampler $G$ are uniformly bounded in probability (Lemma \ref{lemma_mixing_times_gibbs_hierarchical}). In order to do this, we exploit a sufficiency structure which is available for hierarchical models (Lemma \ref{sufficient_lemma}).
\item Finally, we combine points $1.$ and $2.$ to provide a lower bound on the $s$ conductance of $G$.
\end{enumerate}
We need a preliminary lemma, which is well-known and whose proof (inspired by the one of Theorem $7.3$ in \cite{LP17} for discrete Markov chains) is included for completeness. 
\begin{lemma}\label{lower_bound_flux}
Let $P$ be a Markov kernel which is reversible with respect to $\pi$. Then it holds
\[
\lTV \mu_A P^t-\pi \rTV \geq \frac{1}{2}-t\,\frac{P(\partial A)}{\pi(A)}
\]
for every $A \subset \sX$ such that $\pi(A) \leq 1/2$ and $t \geq 1$, with $\mu_A(\cdot) = \pi(\cdot \cap A)/\pi(A)$.
\end{lemma}
\begin{proof}
Let $B \subset A$. Then by reversibility of $P$ we have
\[
\begin{aligned}
(\mu_A P)(B)-\mu_A(B) &= \frac{1}{\pi(A)}\left[\int_A P(x, B) \pi(\d x) -\pi(B)\right] \\
& = \frac{1}{\pi(A)}\left[\int_B P(x, A) \pi(\d x) -\pi(B)\right] \leq 0.
\end{aligned}
\]
If instead $B \subset A^c$ we have $(\mu_A P)(B)-\mu_A(B) = (\mu_A P)(B) \geq 0$. Thus
\[
\lTV \mu_A P-\mu_A \rTV = (\mu_A P)(A^c) = \frac{P(\partial A)}{\pi(A)}.
\]
Using repeatedly the triangle inequality and the monotonicity of the total variation distance with respect to transition kernel multiplications, we obtain
\begin{equation}\label{eq:lower_flux}
\lTV \mu_A P^t-\mu_A \rTV \leq t \frac{P(\partial A)}{\pi(A)}
\end{equation}
for every $t \geq 1$. Moreover
\[
\lTV \mu_A - \pi \rTV \geq \pi(A^c) \geq \frac{1}{2},
\]
so that by \eqref{eq:lower_flux} and the triangle inequality
\[
\begin{aligned}
\frac{1}{2} &\leq \lTV \mu_A - \pi \rTV \leq \lTV \mu_A P^t-\pi \rTV+\lTV \mu_A P^t-\mu_A \rTV\\
& \leq \lTV \mu_A P^t-\pi \rTV + t \frac{P(\partial A)}{\pi(A)},
\end{aligned}
\]
from which the result follows.
\end{proof}
We can use Lemma \ref{lower_bound_flux} to provide a lower bound on $\Phi_s(P)$ in terms of the corresponding mixing times.
\begin{lemma}\label{lemma:s_conductance_mix_time}
Let $P$ be a $\pi$-reversible Markov kernel. For every $s,\epsilon\in(0,1/2)$ 
we have 
\[
\Phi_s(P) \geq \frac{1/2-\epsilon}{t_{mix}(P, \epsilon,s^{-1})}
\,.
\]
\end{lemma}
\begin{proof}
Let $M = 1/s$. For any $A\subseteq \sX$ with $s<\pi(A)\leq1/2$ define
\[
\mu_A(B) = \frac{\pi(B \cap A)}{\pi(A)}, \quad B \subseteq\sX\,.
\]
By construction $\mu_A$ is a $M$-warm start.
Moreover, by Lemma \ref{lower_bound_flux} we have
\[
\lTV \mu_A P^t-\pi \rTV \geq \frac{1}{2}-t\,\frac{P(\partial A)}{\pi(A)}.
\]
Taking the infimum with respect to $A$ such that $\pi(A) > s$, we get
\[
\sup_{\mu \in \sN(\pi, s^{-1})}\, \lTV \mu P^t-\pi \rTV \geq 
\sup_{s<\pi(A)<1/2}\, \lTV \mu P^t-\pi \rTV \geq 
\frac{1}{2}-t\Phi_s(P).
\]
It follows 
$$
\Phi_s(P)\geq t^{-1}\left(\frac{1}{2}-\sup_{\mu \in \sN(\pi, s^{-1})}\, \lTV \mu P^t-\pi \rTV\right)\,.
$$
Taking $t=t_{mix}(P,\epsilon,s^{-1})$ completes the proof.
\end{proof}
We need moreover two other preliminary Lemmas. These can be seen as the analogue of Lemma $4.1$ and Theorem $4.2$ in \cite{AZ23} to the setting of random-scan Gibbs sampler, and the proofs follow similar lines. 
Let $\sX = \sX_1 \times \sX_2$ be a product space and $\pi \in \sP(\sX)$, whose associated Gibbs sampler has operator $G$  as in \eqref{two_blocks_gibbs_nested}. Let $T \, : \, \sX_1 \, \to \, \hat{\sX}_1$, with $\hat{\sX}_1$ with $\hat{\sX}_1$ being an arbitrary measurable space and $T$ such that
\begin{equation}\label{sufficient_T}
    \L(X_2 \mid X_1) = \L(X_2 \mid T(X_1))\quad\hbox{ under }(X_1,X_2)\sim \pi.
\end{equation}
For example, in the case of model \eqref{one_level_nested}, $\L(\psi \mid \bm{\theta}, Y_{1:J}) = \L\left(\psi \mid \bm{T}(\bm{\theta}), Y_{1:J}\right)$, with $\bm{T}$ the sufficient statistics defined in Section \ref{regularity_assumptions}.

Let $\left(T^{(t)}, X_2^{(t)} \right)_{t \geq 1}=\left(T(X_1^{(t)}), X_2^{(t)} \right)_{t \geq 1}$ be the stochastic process obtained as a time-wise mapping of the Markov chain $\left(X_1^{(t)}, X_2^{(t)} \right)_{t \geq 1}$, with operator $G$, under $(x_1, x_2)\mapsto (T(x_1), x_2)$. The latter process contains all the information characterising the convergence of $\left(X_1^{(t)}, X_2^{(t)} \right)_{t \geq 1}$, in the sense made precise in the following lemma. 
Below we denote by 
$\hat{\pi}=\L(T(X_1),X_2)$ under $(X_1,X_2)\sim \pi$, i.e.\ the push-forward of $\pi$ under $(x_1, x_2)\mapsto (T(x_1), x_2)$, by $\hat{\pi}_1(\d t\mid x_2)$ and $\hat{\pi}_2(\d x_2\mid t)$ its conditional distributions and by 
$\hat{G}$ the kernel of the two-block Gibbs sampler targeting $\hat{\pi}$.
Under this notation \eqref{sufficient_T} can be written as $\pi_2(\d x_2\mid x_1)=\hat{\pi}_2(\d x_2\mid T(x_1))$.
\begin{lemma}\label{sufficient_lemma}
Assume \eqref{sufficient_T} holds. Then, the process $\left(T^{(t)}, X_2^{(t)} \right)_{t \geq 1}$ is a Markov chain, its transition kernel coincides with $\hat{G}$, and 
\begin{align*}
t_{mix}(G, \epsilon, M)
&= t_{mix}(\hat{G}, \epsilon, M)
&(M, \epsilon) \in [1, \infty) \times (0,1)\,.
\end{align*}
\end{lemma}
\begin{proof}
The Markovianity of the sequence $\left(T^{(t)}, X_2^{(t)} \right)_{t \geq 1}$ follows by the one of $\left(X_2^{(t)}\right)_{t \geq 1}$, which is well known \citep{D08, R01}. We now show that $\left(T^{(t)}, X_2^{(t)} \right)_{t \geq 1}$ admits $\hat{P}$ as kernel. 
Using the definition of $G$ and the law of total probability conditioning on $X_1^{(t)}$, the conditional distribution of $\left(T^{(t)}, X_2^{(t)} \right)$ given $\left(T^{(t-1)}, X_2^{(t-1)} \right)$ is given by
\begin{align*}
&\L\biggl(\d T^{(t)}, \d X_2^{(t)} \mid T^{(t-1)}, X_2^{(t-1)} \biggr) 
\\&= 
\frac{1}{2}
\delta_{X_2^{(t-1)}}(\d X_2^{(t)})
\int
\delta_{T(x_1)}(\d T^{(t)})
\pi_1\left(\d x_1 \mid X_2^{(t-1)} \right)
\\
&+\frac{1}{2}
\delta_{T^{(t-1)}}(\d T^{(t)})
\int \pi_2(\d X_2^{(t)}\mid x_1)
\alpha\left(\d x_1 \mid T^{(t-1)}, X_2^{(t-1)}\right)
\end{align*}
where 
$\alpha\left(\d x_1 \mid t, x_2\right)$ denotes the conditional distribution of $X_1$ given $T(X_1)=t$ and $X_2=x_2$ when $(X_1,X_2)\sim \pi$.
Note that, by construction, $\int \mathbbm{1}(T(x_1)=t)\alpha\left(\d x_1 \mid t, x_2 \right)=1$ where $\mathbbm{1}$ denotes the indicator function.
Combining the latter with \eqref{sufficient_T} we have
\begin{equation}\label{key_equality}
\begin{aligned}
&\int \pi_2(\d X_2^{(t)}\mid x_1)
\alpha\left(\d x_1 \mid T^{(t-1)}, X_2^{(t-1)}\right)
\\
& = \int
\hat{\pi}_2\left(\d X_2^{(t)} \mid T(x_1) \right)
\alpha\left(\d x_1 \mid T^{(t-1)}, X_2^{(t-1)}\right)
= 
\hat{\pi}_2\left(\d X_2^{(t)} \mid T^{(t-1)} \right)\,.
\end{aligned}
\end{equation}
Also, by definition of $\hat{\pi}$ we have
$$\int
\delta_{T(x_1)}(\d T^{(t)})
\pi_1\left(\d x_1 \mid X_2^{(t-1)} \right)
=
\hat{\pi}_1\left(\d T^{(t)} \mid X_2^{(t-1)} \right)\,.
$$
Combining the above we obtain
\begin{align*}
\L\biggl(\d T^{(t)}, \d X_2^{(t)}& \mid T^{(t-1)}, X_2^{(t-1)} \biggr) \\
&=\frac{1}{2}\hat{\pi}_1\left(\d T^{(t)} \mid X_2^{(t-1)} \right)\delta_{X_2^{(t-1)}}(X_2^{(t)})
+\frac{1}{2}\hat{\pi}_2\left(\d X_2^{(t)} \mid T^{(t-1)}\right)\delta_{T^{(t-1)}}(T^{(t)})\\&
= \hat{G}\left(\left(T^{(t-1)}, X_2^{(t-1)} \right),\left(\d T^{(t)}, \d X_2^{(t)} \right)\right) \,,
\end{align*}
as desired. From the above one can easily deduce that $\left(X_1^{(t)}, X_2^{(t)} \right)_{t \geq 1}$ and $\left(T^{(t)}, X_2^{(t)} \right)_{t \geq 1}$ are \emph{co-deinitializing} as in \cite{R01} and thus, by 
Corollary 2 therein, for every $\mu \in\sP \left(\sX_1\times \sX_2 \right)$ we have
\begin{equation}\label{equality_suff}
\lTV \mu G^t-\pi \rTV = \lTV \nu \hat{G}^t-\hat{\pi}\rTV,
\end{equation}
where $\nu\in \mathcal{P} \left(\hat{\sX}_1\times \sX_2 \right)$ is the push forward of $\mu$ under $(x_1, x_2)\mapsto (T(x_1), x_2)$. Moreover, by \eqref{N_class}  we have that $\nu \in \sN \left(\hat{\pi}, M \right)$ whenever $\mu \in \sN \left(\pi, M \right)$. 
It follows that $t_{mix}(G, \epsilon, M)\leq t_{mix}(\hat{G},\epsilon, M)$. For the reverse inequality, fix $\nu \in \sN \left(\hat{\pi}, M \right)$ and take
\[
\mu(\d x_1, \d x_2) = \int \alpha\left(\d x_1 \mid t,x_2\right)\nu(\d t, \d x_2).
\]
By \eqref{N_class} we have $\mu \in \sN \left(\pi, M \right)$ and thus \eqref{equality_suff}. It follows $t_{mix}(\hat{G}, \epsilon, M)\leq t_{mix}(G,\epsilon, M)$ as desired.
\end{proof}

\begin{lemma}\label{lemma_mixing_times_gibbs_hierarchical}
Consider model \eqref{one_level_nested} and the Gibbs sampler defined in \eqref{two_blocks_gibbs_nested}. Then, under assumptions $(B1)$-$(B6)$, for every $\psi^* \in \R^D$, $M\geq 1$ and $\epsilon > 0$, there exists $T(\psi^*, \epsilon, M)$ such that
\begin{align}
Q_{\psi^*}^{(J)}\left(t_{mix}(G_J, \epsilon, M)  \leq T(\psi^*, \epsilon, M) \right) &\to 1,
\qquad \hbox{ as }J \to \infty\,.
\end{align}
\end{lemma}
\begin{proof}
Let $\bm{T}$ be as in Section \ref{regularity_assumptions},  
$\hat{\pi}_J=\L(\bm{T}(\bm{\theta}), \psi \mid Y_{1:J})$ under $(\bm{\theta}, \psi)\sim \pi_J$ and $\hat{G}_J$ be the kernel of the Gibbs sampler targeting $\hat{\pi}_J$. 
Then, by Lemma \ref{sufficient_lemma} it holds $t_{mix}(G_J, \epsilon, M) = t_{mix}(\hat{G}_J, \epsilon, M)$. 
Let $\tilde{\psi}$ and $\tilde{\bm{T}}$ be the one-to-one transformations defined in $(17)$ and $(19)$ of \cite{AZ23}. Denoting $\tilde{\pi}_J = \L( \tilde{\bm{T}}, \tilde{\psi}\mid Y_{1:J})$, by an analogous version of Lemma $2.1$ in \cite{AZ23} we have $t_{mix}(\hat{G}_J, \epsilon, M) = t_{mix}(\tilde{G}_J, \epsilon, M)$, where $\tilde{G}_J$ is the operator of the Gibbs sampler on $\tilde{\pi}_J$. 
By Lemma $C.18$ in \cite{AZ23} we have that 
\begin{equation}\label{eq:conv_posterior}
\lTV \tilde{\pi}_J-\tilde{\pi}\rTV \to 0
\end{equation}
 as $J \to \infty$ in $Q_{\psi^*}^{(\infty)}$-probability, where $\tilde{\pi}$ is a multivariate Normal distribution with non-singular covariance matrix. Thus, by Theorem $1$ in \cite{amit1996convergence} we have that $\Phi(\tilde{G}) > 0$, where $\tilde{G}$ is the operator of the Gibbs sampler on $\tilde{\pi}$. By Theorem \ref{theorem:delta_Gibbs_operators}, \eqref{eq:conv_posterior} implies that
 \[
 \lim \inf_J \, \Phi_s(\tilde{G}_J) \geq \Phi(\tilde{G}) > 0
 \]
 in $Q_{\psi^*}^{(\infty)}$-probability for every $s > 0$. The result then follows by Lemma \ref{conductance_convergence} with
 \[
 T(\psi^*, \epsilon, M) = \frac{\log(2M)-\log(\epsilon)}{-\log \left(1-\frac{\Phi^2(G)}{2} \right)}.
 \]
\end{proof}
\begin{proof}[Proof of Theorem \ref{theorem: gibbs_one_level_nested}]
Lemma \ref{lemma:s_conductance_mix_time} with $\epsilon = 1/4$ implies
$$\Phi_s(G_J)  \geq \frac{1}{4t_{mix}(G_J, 1/4, s^{-1})}\,.$$
Lemma \ref{lemma_mixing_times_gibbs_hierarchical} with $M = 1/s$ and $\epsilon = 1/4$ implies that there exists $T = T(\psi^*, \epsilon, M) < \infty$ such that $Q_{\psi^*}^{(J)}\left(t_{mix}(G_J, \epsilon, M)\leq T\right) \to 1$ as $J \to \infty$. It follows 
$Q_{\psi^*}^{(J)}\left(\Phi_s(G_J)  \geq g(s) \right)  \to 1$ as $J\to\infty$ with $g(s) = 1/(4T)>0$.
\end{proof}

\subsection{Proof of Corollary \ref{mixing_times_warm_hier}}
\begin{proof}
Combining Theorems \ref{theorem: gibbs_one_level_nested} and \ref{theorem_one_level_nested} we get
\begin{equation}\label{eq:bound_cond_MH}
Q_{\psi^*}^{(J)}\left(\Phi_s(P_J) \geq \frac{\kappa^2}{8}g(s) \right) \to 1,
\end{equation}
as $J \to \infty$, with $g(s) > 0$ for every $s > 0$. Let now
 \[
 T(\psi^*, \epsilon, M) = \frac{\log(2M)-\log(\epsilon)}{-\log \left(1-\frac{\kappa^4g^2(\epsilon/(2M))}{128} \right)},
 \]
then by Lemma \ref{conductance_convergence} and \eqref{eq:bound_cond_MH} we conclude
\[
Q_{\psi^*}^{(J)}\left(t_{mix}(P_J, \epsilon, M)  \leq T(\psi^*, \epsilon, M) \right) \to 1
\]
as $J \to \infty$.
\end{proof}

\subsection{Proof of Proposition \ref{prop:binary_MH}}
\begin{proof}
Without loss of generality we can assume $S^*$ to be compact, which by Weierstrass' Theorem implies that assumption $(C)$ is satisfied. Moreover, assumptions $(B4)-(B6)$ are satisfied by Lemma $5.3$ in \cite{AZ23}, so that the result follows by Corollary \ref{mixing_times_warm_hier}.
\end{proof}

\subsection{Proof of Proposition \ref{prop:feasible}}
Let $\mu, \pi \in \sP(\sX)$ and define
\begin{equation}\label{L2_start}
L_2(\mu, \pi) = \int_{\sX}\left( \frac{\d \mu }{\d \pi}(\x) \right)^2 \pi(\d \x) = \int_{\sX}\frac{\d \mu }{\d \pi}(\x) \mu(\d \x)\,,
\end{equation}
with $L_2(\mu, \pi)=\infty$ if $\mu$ is not absolutely continuous with respect to $\pi$. The next lemma shows that if $L_2(\mu, \pi)$ is small, then $\mu$ is close to a warm start in total variation.
\begin{lemma}\label{non_warm_start}
Assume $\mu \in \sP(\sX)$ and $L_2(\mu, \pi) \leq c$. Then, for every $r \in (0,1)$ there exists $\nu \in \sN\left(\pi, \frac{c}{r(1-r)} \right)$ such that
\[
\lTV \mu - \nu \rTV \leq r\frac{2-r}{1-r}.
\]
\end{lemma}
\begin{proof}
Fix $r \in (0,1)$ and let $A = \left\{\x \in \sX \, : \, \frac{\d \mu }{\d \pi}(\x) \geq c/r \right\}$. 
Then by definition of $L_2(\mu, \pi)$ and $A$ we have
\[
c \geq L_2(\mu, \pi) \geq \int_{A} \frac{\d \mu }{\d \pi}(\x) \mu(\d \x) \geq \frac{c}{r}\mu(A),
\]
which implies that $\mu(A) \leq r$. Define now $\nu \in \mathcal{P}(\sX)$ as $\nu(\cdot) = \mu(\cdot \cap A^c)/\mu(A^c)$. By definition of $A$ and the above, it holds
\begin{equation}\label{warm_nu}
\frac{\d \nu }{\d \pi}(\x) \leq \frac{c}{r\mu(A^c)} \leq \frac{c}{r(1-r)}
\end{equation}
for every $\x\in\sX$. 
Moreover, for every $B \subset \sX$, we have
\[
\begin{aligned}
\left \lvert \nu(B)-\mu(B) \right \rvert &= \left \lvert \mu(B\cap A^c)\left(\frac{1}{\mu(A^c)}-1 \right)-\mu(B \cap A) \right \rvert\\
& = \left \lvert \mu(B\cap A^c)\frac{\mu(A)}{\mu(A^c)}-\mu(B \cap A) \right \rvert 
\\&\leq
\mu(B\cap A^c)\frac{\mu(A)}{\mu(A^c)}+\mu(B \cap A)\\
& \leq \frac{\mu(A)}{\mu(A^c)}+\mu(A) \leq \frac{r}{1-r}+r = r\frac{2-r}{1-r},
\end{aligned}
\]
which implies $\lTV \mu - \nu \rTV \leq r\frac{2-r}{1-r}$, as desired.
\end{proof}
Define $\tilde{\pi}_J \in \sP\left(\R^{D+\ell J} \right)$ as
\begin{equation}\label{eq:pi_tilde_J_def}
\tilde{\pi}_J(\d\psi, \d\bm{\theta}) = N_{(S^*)}\left(\d\psi \mid \hat{\psi}_J, \frac{1}{J}\Fisher^{-1}(\psi^*) \right)\prod_{j = 1}^J\L(\d \theta_j \mid \psi, Y_J),    
\end{equation}
where $N_{(S^*)}$ is the normal distribution truncated on $S^*$ and $\Fisher(\psi^*)$ denotes the Fisher information matrix associated to the marginal likelihood of $\psi$, evaluated at the data-generating value $\psi^*$. We need another preliminary lemma.
\begin{lemma}\label{lemma:bounded_L2}
Under the same notation and assumptions of Proposition \ref{prop:feasible}, there exists a constant $C \geq 1$ such that
\[
Q^{(J)}_{\psi^*}\left(L_2(\mu_J, \tilde{\pi}_J) \leq C\right)\to 1
\]
as $J\to\infty$. The constant $C$ depends only on $c$ and $M$ used in the definition of $\mu_J$. 
\end{lemma}
\begin{proof}
By Lemma $C.45$ in \cite{AZ23} there exist a constant $R=R(c)\in [1,\infty)$ such that 
$$Q^{(J)}_{\psi^*}\left(\mu_J^{(-2)} \in \sN(\tilde{\pi}_J^{(-2)}, R)\right) \to 1\,,$$ 
as $J \to \infty$. 
Thus, with probability converging to $1$ as $J\to\infty$ under $Q^{(J)}_{\psi^*}$, we have
\begin{equation}\label{first_L2}
\begin{aligned}
L_2(\mu_J, \tilde{\pi}_J) 
&=
\int \frac{\d \mu_J}{\d \tilde{\pi}_J}(\psi, \bm{\theta}) \d \mu_J(\d \psi, \d \bm{\theta}) 
=
\int \frac{\d \mu_J^{(-2)}}{\d \tilde{\pi}_J^{(-2)}}(\psi) 
\left(\prod_{j = 1}^J
\frac{\d \mu_J(\theta_j\mid \psi)}{\d \pi_J(\theta_j\mid \psi)} \right)
\d \mu_J(\d \psi, \d \bm{\theta})\\
& 
\leq R\int_{S^*} 
\left(\prod_{j = 1}^J\int \frac{\d \mu_J(\theta_j\mid \psi)}{\d \pi_J(\theta_j\mid \psi)}\, 
\mu_J(\d\theta_j \mid \psi)\right)\mu_J^{(-2)}(\d \psi),
\end{aligned}
\end{equation}
by \eqref{equivalent_N_class} and \eqref{eq:pi_tilde_J_def}. 
Moreover, since $\nu_{j,\psi}  \in \sN\left(\pi_J\left(\theta_j\mid \psi \right), M \right)$, $\mu_J(\d\theta_j \mid \psi)=\nu_{j, \psi}(P^{\psi, Y_j})^{t_J}(\d\theta_j)$, and using assumption (C) and Lemma \ref{conductance_convergence}, we have
\[
\left \lvert \int \frac{\d \mu_J(\theta_j \mid \psi)}{\d \pi_J(\theta_j \mid \psi)}\, \mu_J(\d\theta_j \mid \psi) - 1 \right \rvert \leq M \lTV\nu_{j, \psi}(P^{\psi, Y_j})^{t_J}(\d \theta_j)- \pi_J(\d \theta_j \mid \psi)\rTV \leq M\left(1-\frac{\kappa^2}{2} \right)^{t_J}.
\]
Using the definition of $t_J$ we obtain
\begin{equation}\label{eq:exp_inequality}
\prod_{j = 1}^J\int \frac{\d \mu_J(\theta_j \mid \psi)}{\d \pi_J(\theta_j \mid \psi)}\, \mu_J(\d\theta_j \mid \psi) \leq \left( 1+\frac{M}{J}\right)^J \leq e^{M},    
\end{equation}
for every $\psi \in S^*$ and $J\geq 1$. 
By \eqref{first_L2} and \eqref{eq:exp_inequality} we obtain the desired statement with $C=Re^M$.
\end{proof}
\begin{proof}[Proof of Proposition \ref{prop:feasible}]
Since the conditional distribution of $\bm{\theta}$ given $\psi$ under $\pi_J$ and $\tilde{\pi}_J$ coincide, we have
\begin{equation}\label{distance_pi_first}
\delta_J := \lTV \pi_J - \tilde{\pi}_J \rTV = \lTV \pi_J^{(-2)}-\tilde{\pi}^{(-2)}_J\rTV,
\end{equation}
where $\pi_J^{(-2)}$ and $\tilde{\pi}^{(-2)}_J$ are the marginal distributions of $\psi$ under $\pi_J$ and $\tilde{\pi}_J$, respectively.
Thus, by the Bernstein-von Mises Theorem (e.g. Theorem $10.1$ in \cite{V00}), whose requirements are met thanks to assumptions $(B1)$-$(B3)$, we have
\begin{equation}\label{distance_pi}
\delta_J 
\to 0,
\end{equation}
as $J \to \infty$ in $Q_{\psi^*}^{(\infty)}$-probability. 
Consider now the kernel $\tilde{P}_J$ targeting $\tilde{\pi}_J$, defined as
\[
\tilde{P}_J = \frac{1}{2}\tilde{G}_{1,J}+\frac{1}{2}P_{2, J},
\]
with $P_{2,J}$ as in \eqref{two_blocks_MH_nested} and $\tilde{G}_{1,J}$ Gibbs update on $\tilde{\pi}_J(\psi, \bm{\theta})$. By construction and Lemma \ref{lemma:Delta_GS}, for every $J$ and $N \geq 1$ we have
\begin{equation}\label{distance_P}
\begin{aligned}
\Delta\left( \tilde{P}_J, P_J, N\right) &= \sup_{\mu \in \sN(\tilde{\pi}_J, N)} \lTV \mu \tilde{P}_J-\mu P_J \rTV\\
& = \frac{1}{2}\sup_{\mu \in \sN(\tilde{\pi}_J, N)} \lTV \mu \tilde{G}_{1,J}-\mu G_{1,J} \rTV \leq N\delta_J\,.
\end{aligned}
\end{equation}
Combining \eqref{distance_P} with Theorem \ref{theorem:distance_operators}, for every $s \geq \delta_J$ we get
\begin{equation}\label{eq:cond_bound_with_3}
    \Phi_s(\tilde{P}_J) 
\geq
\Phi_{s-\delta_J}(P_J)-\Delta\left( \tilde{P}_J, P_J, 1/s\right)-\frac{2}{s}\delta_J
\geq
\Phi_{s-\delta_J}(P_J)-\frac{3}{s}\delta_J.
\end{equation}
Therefore, by \eqref{eq:cond_bound_with_3}, \eqref{distance_pi} and Theorems \ref{theorem_one_level_nested} and \ref{theorem: gibbs_one_level_nested},  for every $s>0$ there exists a constant $h(s) > 0$ such that
\begin{equation}\label{conductance_Ptilda}
Q_{\psi^*}^{(J)}\left(\Phi_s(\tilde{P}_J)  \geq h(s) \right) \to 1,
\end{equation}
as $J \to \infty$. 
Note that the value $h(s)$ depends on the model and data generating process under consideration, and thus in particular also on $\psi^*$.
Now, for every $\nu_J \in \sN(\tilde{\pi}_J, N)$ with $N \geq 1$ and $s>0$, applying the triangular inequality followed by  \eqref{distance_P} and Lemma \ref{conductance_convergence}, we have
\[
\begin{aligned}
\lTV \nu_J P_J^t - \pi_J \rTV &\leq 
\lTV \nu_J P_J^t - \nu_J \tilde{P}_J^t \rTV
+
\lTV \nu_J \tilde{P}_J^t - \tilde{\pi}_J \rTV
+
\lTV \tilde{\pi}_J - \pi_J \rTV
\\
&\leq \Delta\left( \tilde{P}_J, P_J, N\right)+\lTV \nu_J \tilde{P}_J^t - \tilde{\pi}_J \rTV+ \delta_J
\\&
\leq
(N+1)\delta_J+Ns + N\left(1-\frac{\Phi^2_s(\tilde{P}_J)}{2} \right)^t\,.
\end{aligned}
\]
By Lemma \ref{lemma:bounded_L2} there exists a constant $C \geq 1$ such that 
$L_2(\mu_J, \tilde{\pi}_J) \leq C$
with $Q_{\psi^*}^{(J)}$-probability going to $1$ as $J\to\infty$.
Thus, for every $r\in(0,1)$, by  Lemma \ref{non_warm_start}, there exist $\nu_J \in \sN\left(\tilde{\pi}_J, \frac{C}{r(1-r)} \right)$ such that $\lTV \mu_J - \nu_J \rTV \leq r\frac{2-r}{1-r}$ as $J \to \infty$ in $Q_{\psi^*}^{(\infty)}$-probability. 
Thus, for every fixed $t\in\mathbb{N}$ and $s,r\in(0,1)$ we have
\begin{align*}
\lTV \mu_J P_J^t - \pi_J \rTV&\leq 
\lTV \mu_J - \nu_J \rTV
+
\lTV \nu_J P_J^t - \pi_J \rTV
\\&\leq
r\frac{2-r}{1-r}
+\left(\frac{C}{r(1-r)}+1\right)\delta_J
+\frac{Cs}{r(1-r)}
+ \frac{C}{r(1-r)}\left(1-\frac{h(s)^2}{2} \right)^t
\end{align*}
as $J \to \infty$ in $Q_{\psi^*}^{(\infty)}$-probability, with $C$ being the constant of Lemma \ref{lemma:bounded_L2}.
Fix $\epsilon>0$.
Then one can that $r=r(\epsilon)\in (0,1)$, $s=s(\epsilon)\in(0,1)$ and $t=t(\epsilon)\in\mathbb{N}$, all depending on $\epsilon$, such that 
$r\frac{2-r}{1-r}< \epsilon/4$,
$\frac{Cs}{r(1-r)}< \epsilon/4$ and 
$\frac{C}{r(1-r)}\left(1-\frac{h(s)^2}{2} \right)^t< \epsilon/4$.
Combining with $\delta_J\to 0$ as $J \to \infty$ in $Q_{\psi^*}^{(\infty)}$-probability it follows that 
$\lTV \mu_J P_J^{t(\epsilon)} - \pi_J \rTV<\epsilon$
as $J \to \infty$ in $Q_{\psi^*}^{(\infty)}$-probability and thus 
$Q_{\psi^*}^{(J)}\left(t_{mix}(P_J, \epsilon, \mu_J)\leq t(\epsilon)\right) \to 1$ as $J \to \infty$ as desired. 
The results follows with $T\left(\psi^*, \epsilon \right)=t(\epsilon)$ where the dependence on $\psi^* $ comes from $h(s)$.
\end{proof}

\subsection{Proof of Proposition \ref{MwG_logistic_regression}}
\begin{proof}
The result follows immediately by combining Proposition \ref{prop:RWM_logconcave} with Corollary \ref{cor:bound_conductance}, choosing $K = S\times \R^d$.
\end{proof}
\subsection{Proof of Lemma \ref{dimensionality_reduction_logistic}}
\begin{proof}
Notice that by definition of $T=\bm{\theta}^\top\Sigma^{-1}\bm{\theta}$ we have
\[
\pi(\d \alpha \mid \bm{\theta}) = \tilde{\pi}(\d \alpha \mid T).
\]
Thus, for every $\epsilon > 0$ and $M>0$ by Lemma \ref{sufficient_lemma} it holds $t_{mix}(\tilde{G}, \epsilon, M) = t_{mix}(G, \epsilon, M)$. Thus, Lemma \ref{lemma:s_conductance_mix_time} with $\epsilon = 1/4$ implies
\begin{equation}\label{eq:first_bound_mixing}
\Phi_s(G)  \geq \frac{1}{4t_{mix}(G, 1/4, s^{-1})} = \frac{1}{4t_{mix}(\tilde{G}, 1/4, s^{-1})}\,.
\end{equation}
Moreover, by Lemma \ref{conductance_convergence} we have
\begin{equation}\label{eq:second_bound_mixing}
t_{mix}(\tilde{G}, 1/4, s^{-1}) \leq \frac{\log(8)-\log(s)}{-\log \left(1-\frac{\Phi^2_{\frac{s}{8}}(\tilde{G})}{2} \right)}.
\end{equation}
The result follows by combining \eqref{eq:first_bound_mixing} and \eqref{eq:second_bound_mixing}.
\end{proof}
\subsection{Proof of Proposition \ref{prop:conductance_diffusion}}
\begin{proof}
Denote with $\P_\theta$ the law of $Y_i$ induced by \eqref{eq:sde} for a fixed $\theta$ and with $\W(X_{t_{i-1}}, X_{t_i})$ the law of the corresponding Brownian bridge. By equation $(10)$ in \cite{roberts2001inference} we have
\begin{equation}\label{eq:RN_derivative}
\frac{\d \P_\theta}{\d \W(X_{t_{i-1}}, X_{t_i})}(Y_i) = G(Y_i, \theta)\frac{g(X_{t_{i-1}}, X_{t_i})}{f_\theta(X_{t_{i-1}}, X_{t_i})}, \quad g(x, y) = \frac{1}{\sqrt{2\pi}}e^{-\frac{1}{2\Delta}(y-x)^2}.
\end{equation}
By equation $(4)$ in \cite{beskos2006retrospective} thanks to $(C1)$ we have
\begin{equation}\label{eq:prop_RN}
G(Y_i, \theta) = \text{exp}\left\{-\frac{1}{2}\int_{t_{i-1}}^{t_i}\left[b^2(\theta, Y_{i,t})-b'(\theta, Y_{i,t})  \right] \, \d t \right\}.
\end{equation}
Combining \eqref{eq:RN_derivative} and \eqref{eq:prop_RN}, then for every $\theta$ by $(C2)$ there exists $M(\theta) > 0$ such that $G(Y_i, \theta) \leq M(\theta)$ for every $Y_i$. Thanks to $(C3)$ and continuity of $b$ with respect to $\theta$ we can find a constant $M > 0$ such that $M(\theta) \leq M$ for every $\theta \in K$. Finally, thanks to $(C3)$ and $(C4)$ we have
\[
\frac{g(X_{t_{i-1}}, X_{t_i})}{f_\theta(X_{t_{i-1}}, X_{t_i})} \leq R,
\]
for a suitable $R = R(c) >0$, for every $\theta \in S$ and $i = 1, \dots, N$. Thus, in the end we get
\[
\frac{\d \P_\theta}{\d \W(X_{t_{i-1}}, X_{t_i})}(Y_i) \leq MR
\]
Thus, by Proposition \ref{prop:conductance_ind} we have $\kappa(P_i^\theta) \geq m/M$ for every $\theta \in K$, which by  Lemma \ref{lemma_product_conductance} implies
\[
\kappa(P_2) \geq \frac{1}{4M^2R^2}.
\]
The result then follows by Corollary \ref{cor:bound_conductance} with $\kappa = \frac{1}{4M^2R^2}$.
\end{proof}

\end{appendices}

\end{document}

%% file: notation.tex
\usepackage[a4paper,top=2.5cm,bottom=2.5cm,left=2.5cm,right=2.5cm,marginparwidth=1.75cm]{geometry}

\usepackage[numbers]{natbib}
\usepackage{natbib} 
\usepackage[title]{appendix}
\usepackage{hyperref} 
\usepackage{url}

\usepackage{titling}
\usepackage{comment}
\usepackage{xr}

\usepackage{amsmath,amsfonts,amssymb}
\usepackage{bbm} 
\usepackage{cancel}
\usepackage{bm}
\usepackage{amsthm}
\usepackage{tabularx} 
\usepackage{geometry}

\usepackage{mathtools}

\usepackage{mathtools}
\usepackage[x11names]{xcolor}
\newtagform{red}{\color{red}(}{)}

\newcommand{\rTV}{\right \|_{TV}}
\newcommand{\lTV}{\left \|}
\newcommand{\rE}{\right \|_{2}}
\newcommand{\lE}{\left \|}

\newcommand\gzcom[1]{\textcolor{red}{#1}}

\newtheorem{theorem}{Theorem}
\newtheorem{proposition}{Proposition}
\newtheorem{corollary}{Corollary}
\newtheorem{lemma}{Lemma}

\newtheorem{remark}{Remark}
\usepackage{xcolor,colortbl}
\definecolor{Gray}{gray}{0.85}
\newcolumntype{a}{>{\columncolor{Gray}}X}

\def\sN {\mathcal{N}}

\def\sH {\mathcal{H}}
\def\sP {\mathcal{P}}
\def\sO {\mathcal{O}}

\def \T {\mathbb{T}}
\def \W {\mathbb{W}}
\def \L {\mathcal{L}}
\def \sX {\mathcal{X}}
\def \sY {\mathcal{Y}}

\def \R {\mathbb{R}}
\def \P {\mathbb{P}}
\def\d {\text{d}}
\def\x {\textbf{x}}

\def\y {\textbf{y}}

\def\z {\textbf{z}}
\def\w {\textbf{w}}
\def \I {\mathbb{I}}
\def \Fisher {\mathcal{I}}

\def \W {\textbf{W}}

\def \simiid {\overset{\text{iid}}{\sim}}

%% file: Bib_MwG.bib
@article{ascolani2024entropy,
  title={{Entropy contraction of the Gibbs sampler under log-concavity}},
  author={Ascolani, Filippo and Lavenant, Hugo and Zanella, Giacomo},
  journal={arXiv preprint arXiv:2410.00858},
  year={2024}
}

@article{liu1994collapsed,
  title={{The collapsed Gibbs sampler in Bayesian computations with applications to a gene regulation problem}},
  author={Liu, Jun S},
  journal={Journal of the American Statistical Association},
  volume={89},
  number={427},
  pages={958--966},
  year={1994},
  publisher={Taylor \& Francis}
}

@article{deligiannidis2018ergodic,
  title={Which ergodic averages have finite asymptotic variance?},
  author={Deligiannidis, George and Lee, Anthony},
  journal={The Annals of Applied Probability},
  volume={28},
  number={4},
  pages={2309--2334},
  year={2018},
  publisher={JSTOR}
}

@article{luu2024gibbs,
  title={{Is Gibbs sampling faster than Hamiltonian Monte Carlo on GLMs?}},
  author={Luu, Son and Xu, Zuheng and Surjanovic, Nikola and Biron-Lattes, Miguel and Campbell, Trevor and Bouchard-C{\^o}t{\'e}, Alexandre},
  journal={Proceedings of the 28th International Conference on Artificial Intelligence and Statistics (AISTATS)},
  year={2025},
Volume = {258}
}

@misc{robert1999monte,
  title={Monte Carlo Statistical Methods},
  author={Robert, CP},
  year={1999},
  publisher={Springer-Verlag New York}
}

@article{roberts1996geometric,
  title={Geometric convergence and central limit theorems for multidimensional Hastings and Metropolis algorithms},
  author={Roberts, Gareth O and Tweedie, Richard L},
  journal={Biometrika},
  volume={83},
  number={1},
  pages={95--110},
  year={1996},
  publisher={Oxford University Press}
}

@article{baxendale2005renewal,
  title={{Renewal theory and computable convergence rates for geometrically ergodic Markov chains}},
  journal={The Annals of Applied Probability},
  volume={15},
  number={1},
  pages={700-738},
  author={Baxendale, Peter H},
  year={2005}
}

@article{he2016scan,
  title={{Scan order in Gibbs sampling: Models in which it matters and bounds on how much}},
  author={He, Bryan D and De Sa, Christopher M and Mitliagkas, Ioannis and R{\'e}, Christopher},
  journal={Advances in neural information processing systems},
  volume={29},
  year={2016}
}

@article{gaitonde2024comparison,
  title={{Comparison Theorems for the Mixing Times of Systematic and Random Scan Dynamics}},
  author={Gaitonde, Jason and Mossel, Elchanan},
  journal={arXiv preprint arXiv:2410.11136},
  year={2024}
}

@article{roberts2015surprising,
  title={{Surprising convergence properties of some simple Gibbs samplers under various scans}},
  author={Roberts, Gareth O and Rosenthal, Jeffrey S},
  journal={International Journal of Statistics and Probability},
  volume={5},
  number={1},
  pages={51--60},
  year={2015}
}

@article{diaconis2010stochastic,
  title={{Stochastic alternating projections}},
  author={Diaconis, Persi and Khare, Kshitij and Saloff-Coste, Laurent},
  journal={Illinois Journal of Mathematics},
  volume={54},
  number={3},
  pages={963--979},
  year={2010},
  publisher={Duke University Press}
}

@article{khare2011spectral,
  title={{A spectral analytic comparison of trace-class data augmentation algorithms and their sandwich variants}},
  author={Khare, Kshitij and Hobert, James P},
  journal={The Annals of Statistics},
volume = {39},
number = {5},
pages = {2585-2606},
  year={2011}
}

@article{kamatani2014local,
  title={{Local consistency of Markov chain Monte Carlo methods}},
  author={Kamatani, Kengo},
  journal={Annals of the Institute of Statistical Mathematics},
  volume={66},
  number={1},
  pages={63--74},
  year={2014},
  publisher={Springer}
}

@article{tang2022computational,
  title={{On the computational complexity of Metropolis-adjusted Langevin algorithms for Bayesian posterior sampling}},
  author={Tang, Rong and Yang, Yun},
  journal={Journal of Machine Learning Research},
  volume={25},
  number={157},
  pages={1--79},
  year={2024}
}

@article{negrea2022statistical,
  title={{Statistical inference with stochastic gradient algorithms}},
  author={Negrea, Jeffrey and Yang, Jun and Feng, Haoyue and Roy, Daniel M and Huggins, Jonathan H},
  journal={arXiv preprint arXiv},
  volume={2207},
  year={2022}
}

@article{nickl2022polynomial,
  title={{On polynomial-time computation of high-dimensional posterior measures by Langevin-type algorithms}},
  author={Nickl, Richard and Wang, Sven},
  journal={Journal of the European Mathematical Society},
  volume={26},
  number={3},
  pages={1031--1112},
  year={2022}
}

@article{roberts2001approximate,
  title={{Approximate predetermined convergence properties of the Gibbs sampler}},
  author={Roberts, Gareth O and Sahu, Sujit K},
  journal={Journal of Computational and Graphical Statistics},
  volume={10},
  number={2},
  pages={216--229},
  year={2001},
  publisher={Taylor \& Francis}
}

@article{belloni2009computational,
  title={{On the computational complexity of MCMC-based estimators in large samples}},
  author={Belloni, Alexandre and Chernozhukov, Victor},
  year={2009}
}

@article{kastner2014ancillarity,
  title={{Ancillarity-sufficiency interweaving strategy (ASIS) for boosting MCMC estimation of stochastic volatility models}},
  author={Kastner, Gregor and Fr{\"u}hwirth-Schnatter, Sylvia},
  journal={Computational Statistics \& Data Analysis},
  volume={76},
  pages={408--423},
  year={2014},
  publisher={Elsevier}
}

@article{yu2011center,
  title={{To center or not to center: That is not the question: an Ancillarity--Sufficiency Interweaving Strategy (ASIS) for boosting MCMC efficiency}},
  author={Yu, Yaming and Meng, Xiao Li},
  journal={Journal of Computational and Graphical Statistics},
  volume={20},
  number={3},
  pages={531--570},
  year={2011},
  publisher={Taylor \& Francis}
}

@article{gelfand1995efficient,
  title={{Efficient parametrisations for normal linear mixed models}},
  author={Gelfand, Alan E and Sahu, Sujit K and Carlin, Bradley P},
  journal={Biometrika},
  volume={82},
  number={3},
  pages={479--488},
  year={1995},
  publisher={Oxford University Press}
}

@article{neath2009variable,
  title={{Variable-at-a-time implementations of Metropolis-Hastings}},
  author={Neath, Ronald C and Jones, Galin L},
  journal={arXiv preprint arXiv:0903.0664},
  year={2009}
}

@article{johnson2013component,
  title={{Component-wise Markov chain Monte Carlo: Uniform and geometric ergodicity under mixing and composition}},
  author={Johnson, Alicia A and Jones, Galin L and Neath, Ronald C},
  year={2013}
}

@article{jarner2000geometric,
  title={{Geometric ergodicity of Metropolis algorithms}},
  author={Jarner, S{\o}ren Fiig and Hansen, Ernst},
  journal={Stochastic processes and their applications},
  volume={85},
  number={2},
  pages={341--361},
  year={2000},
  publisher={Elsevier}
}

@article{biswas2019estimating,
  title={{Estimating convergence of Markov chains with L-lag couplings}},
  author={Biswas, Niloy and Jacob, Pierre E and Vanetti, Paul},
  journal={Advances in Neural Information Processing Systems},
  volume={32},
  year={2019}
}

@article{hoffman2014no,
  title={{The No-U-Turn sampler: adaptively setting path lengths in Hamiltonian Monte Carlo}},
  author={Hoffman, Matthew D and Gelman, Andrew},
  journal={J. Mach. Learn. Res.},
  volume={15},
  number={1},
  pages={1593--1623},
  year={2014}
}

@Misc{RStan,
    title = {{RStan}: the {R} interface to {Stan}},
    author = {{Stan Development Team}},
    note = {R package version 2.32.5},
    year = {2024},
    url = {https://mc-stan.org/},
  }

@article{livingstone2022barker,
  title={{The Barker proposal: combining robustness and efficiency in gradient-based MCMC}},
  author={Livingstone, Samuel and Zanella, Giacomo},
  journal={Journal of the Royal Statistical Society Series B: Statistical Methodology},
  volume={84},
  number={2},
  pages={496--523},
  year={2022},
  publisher={Oxford University Press}
}

@article{jones2014convergence,
  title={{Convergence of conditional Metropolis-Hastings samplers}},
  author={Jones, Galin L and Roberts, Gareth O and Rosenthal, Jeffrey S},
  journal={Advances in Applied Probability},
  volume={46},
  number={2},
  pages={422--445},
  year={2014},
  publisher={Cambridge University Press}
}

@article{qin2022telescope,
  title={{Spectral telescope: Convergence rate bounds for random-scan Gibbs samplers based on a hierarchical structure}},
  author={Qin, Qian and Wang, Guanyang},
  journal={The Annals of Applied Probability},
  volume={34},
  number={1B},
  pages={1319--1349},
  year={2024},
  publisher={Institute of Mathematical Statistics}
}

@article{qin2022convergence,
  title={{Convergence rates of two-component MCMC samplers}},
  author={Qin, Qian and Jones, Galin L},
  journal={Bernoulli},
  volume={28},
  number={2},
  pages={859--885},
  year={2022},
  publisher={Bernoulli Society for Mathematical Statistics and Probability}
}

@article{roberts1997geometric,
  title={{Geometric ergodicity and hybrid Markov chains}},
  author={Roberts, Gareth and Rosenthal, Jeffrey},
  year={1997}
}

@article{B97,
  title={{Isoperimetric constants for product probability measures}},
  author={Bobkov, Sergey G and Houdr{\'e}, Christian},
  journal={The Annals of Probability},
  pages={184--205},
  year={1997},
  publisher={JSTOR}
}

@book{LP17,
  title={{Markov chains and mixing times}},
  author={Levin, David A and Peres, Yuval},
  volume={107},
  year={2017},
  publisher={American Mathematical Soc.}
}

@article{CJ23,
  title={{A calculus for Markov chain Monte Carlo: studying approximations in algorithms}},
  author={Caprio, R. and Johansen, A.M.},
  journal={arXiv preprint arXiv:2310.03853},
  year={2023}
}

@article{AL22,
  title={{Explicit convergence bounds for Metropolis Markov chains: Isoperimetry, spectral gaps and profiles}},
  author={Andrieu, Christophe and Lee, Anthony and Power, Sam and Wang, Andi Q},
  journal={The Annals of Applied Probability},
  volume={34},
  number={4},
  pages={4022--4071},
  year={2024},
  publisher={Institute of Mathematical Statistics}
}

@article{papaspiliopoulos2007general,
  title={{A General Framework for the Parametrization of Hierarchical Models}},
  author={Papaspiliopoulos, Omiros and Roberts, Gareth O. R. and Sk{\"o}ld, Martin},
  journal={Statistical Science},
  pages={59--73},
  year={2007},
  publisher={JSTOR}
}

@article{AZ23,
  title={{Dimension-free mixing times of {G}ibbs samplers for {B}ayesian hierarchical models}},
  author={Ascolani, Filippo and Zanella, Giacomo},
  journal={Ann. Statist.},
  volume = {In press},
  year={2024}
}

@article{amit1996convergence,
  title={{Convergence properties of the {G}ibbs sampler for perturbations of {G}aussians}},
  author={Amit, Yali},
  journal={The Annals of Statistics},
  volume={24},
  number={1},
  pages={122--140},
  year={1996},
  publisher={Institute of Mathematical Statistics}
}

@article{B13,
  title={{Optimal tuning of the hybrid Monte Carlo algorithm}},
  author={Beskos, A. and Pillai, N. and Roberts, G. and Sanz-Serna, J. and Stuart, A.},
  journal={Bernoulli},
  volume={19},
  pages={1501--1534},
  year={2013}
}

@article{besag1993spatial,
  title={{Spatial statistics and Bayesian computation}},
  author={Besag, Julian and Green, Peter J},
  journal={Journal of the Royal Statistical Society Series B: Statistical Methodology},
  volume={55},
  number={1},
  pages={25--37},
  year={1993},
  publisher={Oxford University Press}
}

@article{gelfand1990sampling,
  title={{Sampling-based approaches to calculating marginal densities}},
  author={Gelfand, Alan E and Smith, Adrian FM},
  journal={Journal of the American statistical association},
  volume={85},
  number={410},
  pages={398--409},
  year={1990},
  publisher={Taylor \& Francis}
}

@book{B11,
  title={{Handbook of Markov Chain Monte Carlo}},
  author={Brooks, S. and Gelman, A. and Jones, G. L. and Meng, X.},
  year={2011},
  publisher={Chapman and Hall}
}

@article{C92,
  title={{Explaining the {G}ibbs Sampler}},
  author={Casella, G. and George, E. I.},
  journal={Am. Stat.},
  volume={46},
  pages={167--174},
  year={1992}
}

@article{CLM23,
  title={{Solidarity of Gibbs Samplers: the spectral gap}},
  author={Chlebicka, Iwona and Latuszy{\'n}ski, Krzysztof and Miasojedow, B.},
  journal={The Annals of Applied Probability},
  volume={35},
  number={1},
  pages={142--157},
  year={2025},
  publisher={Institute of Mathematical Statistics}
}

@article{D17,
  title={{Theoretical Guarantees for Approximate Sampling from Smooth and Log-Concave Densities}},
  author={Dalalyan, A. S.},
  journal={J. R. Stat. Soc. Ser. B.},
  volume={79},
  pages={651--676},
  year={2017}
}

@article{D08,
  title={{Gibbs Sampling, Exponential Families and Orthogonal Polynomials}},
  author={Diaconis, P. and Khare, K. and Saloff-Coste, L.},
  journal={Stat. Sci.},
  volume={23},
  pages={151--178},
  year={2008}
}

@article{DM17,
  title={{Nonasymptotic convergence analysis for the unadjusted Langevin algorithm}},
  author={Durmus, A. and Moulines, E.},
  journal={Ann. Appl. Probab.},
  volume={27},
  pages={1551--1587},
  year={2017}
}

@article{D19,
  title={{Log--concave sampling: Metropolis--Hastings algorithms are fast!}},
  author={Dwivedi, R. and Chen, Y. and Wainwright, M. J. and Yu, B.},
  journal={J. Mach. Learn. Res.},
  volume={20},
  pages={1--42},
  year={2019}
}

@article{F21,
  title={{mcmcse: Monte Carlo Standard Errors for MCMC}},
  author={Flegal, J. M. and Hughes, J. and Vats, D. and Gupta, K. and Maji, U.},
  journal={R package},
  year={2021}
}

@book{G13,
  title={{Bayesian Data Analysis}},
  author={Gelman, A. and Carlin, J. B. and Stern, H. S. and Dunson, D. B. and Vehtari, A. and Rubin, D. B.},
  year={2013},
  publisher={CRC press}
}

@book{GH07,
  title={{Data Analysis Using Regression and Multilevel/Hierarchical Models}},
  author={Gelman, A. and Hill, J. L.},
  year={2007},
  publisher={Cambridge University Press}
}

@article{GW92,
  title={{Adaptive Rejection Sampling for Gibbs Sampling}},
  author={Gilks, W. R. and Wild, P.},
  journal={J. R. Stat. Soc. Ser. C},
  volume={41},
  pages={337--348},
  year={1992}
}

@article{GF15,
  title={{A Practical Sequential Stopping Rule for High-Dimensional Markov Chain Monte Carlo}},
  author={Gong, L. and Flegal, J. M.},
  journal={J. Comput. Graph. Stat.},
  volume={25},
  pages={684--700},
  year={2015}
}

@article{G15,
  title={{Bayesian computation: a summary of the current state, and samples backwards and forwards}},
  author={Green, P. J. and Latuszynski, K. and Pereyra, M. and Robert, C. P.},
  journal={Stat. Comput.},
  volume={25},
  pages={835--862},
  year={2015}
}

@article{JH21,
  title={{Dimension free convergence rates for Gibbs samplers for Bayesian linear mixed models}},
  author={Jin, Z. and Hobert, J. P.},
  journal={Stoch. Process. Their Appl.},
  volume={148},
  pages={25--67},
  year={2022}
}

@article{K09,
  title={{Rates of convergence of some multivariate Markov chains with polynomial eigenfunctions}},
  author={Khare, Kshitij and Zhou, Hua},
  journal={Ann. Appl. Probab.},
  volume={2},
  pages={737--777},
  year={2009}
}

@article{K14,
  title= {{Local consistency of Markov chain Monte Carlo methods}},
  author={Kamatani, K.},
  journal={Ann. Inst. Stat. Math.},
  volume={66},
  pages={63--74},
  year={2014}
}

@article{LS93,
  title={{Random Walks in a Convex Body and an Improved Volume Algorithm}},
  author={Lov\'asz, L. and Simonovits, M.},
  journal={Random Struct. and Alg.},
  volume={4},
  pages={359--412},
  year={1993}
}

@article{MFR23,
  title={{Computing Bayes: From then ‘til now}},
  author={Martin, Gael M and Frazier, David T and Robert, Christian P},
  journal={Statistical Science},
  volume={39},
  number={1},
  pages={3--19},
  year={2024},
  publisher={Institute of Mathematical Statistics}
}

@inproceedings{PRS03,
  title={{Non-Centered Parameterizations for Hierarchical Models and Data Augmentation (with discussion)}},
  author={Papaspiliopoulos, O. and Roberts, G. O. and Sk\"old, M.},
  booktitle={Bayesian Statistics (J. M. Bernardo, M. J. Bayarri, J. O. Berger, A. P. Dawid, D. Heckerman, A. F. M. Smith and M. West, eds.)},
  pages={307--326},
  year={2003}
}

@article{SPZ23,
  title={Scalable Bayesian computation for crossed and nested hierarchical models},
  author={Papaspiliopoulos, Omiros and Stumpf-F{\'e}tizon, Timoth{\'e}e and Zanella, Giacomo},
  journal={Electronic Journal of Statistics},
  volume={17},
  number={2},
  pages={3575--3612},
  year={2023},
  publisher={The Institute of Mathematical Statistics and the Bernoulli Society}
}

@article{Q19,
  title={{Convergence complexity analysis of Albert and Chib's algorithm for Bayesian probit regression}},
  author={Qin, Q. and Hobert, J. P.},
  journal={Ann. Statist.},
  volume={47},
  pages={2320--2347},
  year={2019}
}

@article{Q22,
  title={{Wasserstein-based methods for convergence complexity analysis of MCMC with applications}},
  author={Qin, Q. and Hobert, J. P.},
  journal={Ann. Appl. Prob.},
  volume={32},
  pages={124--166},
  year={2022}
}

@article{R98,
  title={{Optimal scaling of discrete approximations to Langevin diffusions}},
  author={Roberts, G. O. and Rosenthal, J. S.},
  journal={J. R. Stat. Soc. Ser. B},
  volume={60},
  pages={255--268},
  year={1998}
}

@article{R01,
  title={{Markov Chains and De-Initializing Processes}},
  author={Roberts, G. O. and Rosenthal, J. S.},
  journal={Scand. J. Stat.},
  volume={28},
  pages={489--504},
  year={2001}
}

@article{R97,
  title={{Updating Schemes, Correlation Structure, Blocking and Parameterization for the Gibbs Sampler}},
  author={Roberts, G. O. and Sahu, S. H.},
  journal={J. R. Stat. Soc. Ser. B},
  volume={59},
  pages={291--317},
  year={1997}
}

@article{roberts2001inference,
  title={{On inference for partially observed nonlinear diffusion models using the Metropolis--Hastings algorithm}},
  author={Roberts, Gareth O and Stramer, Osnat},
  journal={Biometrika},
  volume={88},
  number={3},
  pages={603--621},
  year={2001},
  publisher={Oxford University Press}
}

@article{beskos2006retrospective,
  title={{Retrospective exact simulation of diffusion sample paths with applications}},
  author={Beskos, Alexandros and Papaspiliopoulos, Omiros and Roberts, Gareth O},
  journal={Bernoulli},
  volume={12},
  number={6},
  pages={1077--1098},
  year={2006},
  publisher={Bernoulli Society for Mathematical Statistics and Probability}
}

@article{R95,
  title={{Minorization Conditions and Convergence Rates for Markov Chain Monte Carlo}},
  author={Rosenthal, J. S.},
  journal={J. Am. Stat. Assoc},
  volume={90},
  pages={558--566},
  year={1995}
}

@article{T22,
  title={{Computational Complexity of Metropolis-Adjusted Langevin Algorithms for Bayesian Posterior Sampling}},
  author={Tang, R. and Yang, Y.},
  journal={arXiv preprint arXiv:2206.06491},
  year={2022}
}

@book{V00,
  title={{Asymptotic Statistics}},
  author={van der Vaart, A. W.},
  year={2000},
  publisher={Cambridge University Press}
}

@article{Y17,
  title={{Complexity results for MCMC derived from quantitative bounds}},
  author={Yang, J. and Rosenthal, J. S.},
  journal={Ann. Appl. Prob.},
  volume={33},
  pages={1459--1500},
  year={2023}
}

@article{WSC22,
  title={{Minimax Mixing Time of the Metropolis-Adjusted Langevin Algorithm for Log-Concave Sampling}},
  author={Wu, K. and Schmidler, S. and Chen, Y.},
  journal={J. Mach. Learn. Res.},
  volume={23},
  pages={1--63},
  year={2022}
}

@article{PZ20,
  title={{Scalable inference for crossed random effects models}},
  author={Papaspiliopoulos, O. and Roberts, G. and Zanella, G.},
  journal={Biometrika},
  volume={107},
  pages={25--40},
  year={2020}
}

@article{polson1994bayes,
  title={Bayes factors for discrete observations from diffusion processes},
  author={Polson, Nicholas G and Roberts, Gareth O},
  journal={Biometrika},
  volume={81},
  number={1},
  pages={11--26},
  year={1994},
  publisher={Oxford University Press}
}

@article{papaspiliopoulos2013data,
  title={Data augmentation for diffusions},
  author={Papaspiliopoulos, Omiros and Roberts, Gareth O and Stramer, Osnat},
  journal={Journal of Computational and Graphical Statistics},
  volume={22},
  number={3},
  pages={665--688},
  year={2013},
  publisher={Taylor \& Francis}
}

@article{tong2020mala,
  title={{MALA-within-{G}ibbs samplers for high-dimensional distributions with sparse conditional structure}},
  author={Tong, Xin T and Morzfeld, Matthias and Marzouk, Youssef M},
  journal={SIAM Journal on Scientific Computing},
  volume={42},
  number={3},
  pages={A1765--A1788},
  year={2020},
  publisher={SIAM}
}

@article{qin2023spectral,
  title={{Spectral gap bounds for reversible hybrid Gibbs chains}},
  author={Qin, Qian and Ju, Nianqiao and Wang, Guanyang},
  journal={The Annals of Statistics},
  volume={53},
  number={4},
  pages={1613--1638},
  year={2025},
  publisher={Institute of Mathematical Statistics}
}

@article{madras2002markov,
  title={Markov chain decomposition for convergence rate analysis},
  author={Madras, Neal and Randall, Dana},
  journal={Annals of Applied Probability},
  pages={581--606},
  year={2002},
  publisher={JSTOR}
}

@article{SmiRob93,
  title={Bayesian computation via the {G}ibbs sampler and related {M}arkov chain {M}onte {C}arlo methods},
  author={Smith, Adrian FM and Roberts, Gareth O},
  journal={Journal of the Royal Statistical Society: Series B (Methodological)},
  volume={55},
  number={1},
  pages={3--23},
  year={1993},
  publisher={Wiley Online Library}
}
